\documentclass[3p,times,authoryear]{elsarticle}
\usepackage{amsmath}
\usepackage{amsthm}
\usepackage{comment}
\usepackage{graphicx}  
\usepackage[utf8]{inputenc}
\usepackage[capitalise]{cleveref}
\usepackage{xcolor}
\usepackage[subrefformat=parens]{subcaption}
\usepackage{url}
\newtheorem{ass}{Assumption}
\newtheorem{prop}{Proposition}

\usepackage{nccmath}

\everymath{\displaystyle}

\begin{document}

\begin{frontmatter}
\title{Perimeter control in a mixed bimodal bathtub model}  
\address[kana]{Institute of Science and Engineering, Kanazawa University, Kakuma, Kanazawa,
Ishikawa, 920-1192, Japan}
\address[utokyo]{Department of Civil Engineering, Graduate school of Engineering, The University of Tokyo,  7-3-1 Hongo, Bunkyo, Tokyo, 113-8656, Japan}

\author[kana]{Takao Dantsuji\corref{cor}}\ead{dantsuji@se.kanazawa-u.ac.jp}\cortext[cor]{Corresponding author.}
\author[kana]{Yuki Takayama}\ead{ytakayama@se.kanazawa-u.ac.jp}
\author[utokyo]{Daisuke Fukuda}\ead{fukuda@civil.t.u-tokyo.ac.jp }

\begin{abstract}
Perimeter control involves monitoring network-wide traffic and regulating traffic inflow to alleviate hypercongestion. Implementation of transit priority with perimeter control measures, which allow transit into a controlled area without queuing at the perimeter boundary, is an effective strategy in bimodal transportation systems. 
However, travelers' behavior changes in response to perimeter control strategies, such as shifts in their departure times and transportation modes, have not been fully investigated. Therefore, important questions remain, such as the use of transit during perimeter control with transit priority. 
This paper examines the travelers' behavior changes in response to perimeter control with transit priority in a mixed bimodal transportation system with cars and flexible route transit (FRT) vehicles. We model departure time and transportation mode choices in such a transportation system with hypercongestion and discomfort in FRT (called the mixed bimodal bathtub model). Initially, we investigate the properties of dynamic user equilibrium without perimeter control. Then, we study the equilibrium patterns during perimeter control with transit priority. Unlike existing works, we find that the number of FRT passengers decreases with time toward the desired arrival time and that FRT may not be used around the peak of rush hour. Furthermore, transit priority may not be sufficient to promote the use of FRT, and additional incentive such as subsidy for lower fares may be required to encourage FRT use during perimeter control. Finally, we show that operating many FRT vehicles does not always decrease the equilibrium cost, even under perimeter control with transit priority.
\end{abstract}

\begin{keyword}
perimeter control, bimodal bathtub model, hypercongestion,  discomfort, car, flexible route transit
\end{keyword}

\end{frontmatter}

\section{Introduction}
\subsection{Background}
Traffic demand is highly concentrated during the rush hour in urban cities with limited spaces. Consequently, hypercongestion\footnote{The downward sloping part of the inverted-U shaped relationship between flow and density} occurs in road networks of many cities \citep{geroliminis2008existence, loder2017empirics}.
Increased vehicle inflow leads to lower throughput in the hypercongested state, and this externality causes inefficient use of transportation systems. To address this issue, cities have introduced public transportation systems as an alternative mode that can carry more passengers using less urban space than private vehicles. A flexible door-to-door transit system has been proposed as an alternative to transitional fixed-route-based transit to promote the use of transit  \citep{vansteenwegen2022survey}.  However, in some of  these cities, the increases in the use of public transportation is causing considerable inconvenience, such as long waiting times at stations or stops \citep{basso2019efficiency}. In particular, public transit vehicles, such as buses, encounter hypercongestion in road networks in addition to such inconvenience. These externalities play an important role in travelers' decision making \citep{tirachini2014multimodal, basso2019efficiency}. 
Under these circumstances, commuters choose their trip timing (departure time) and transportation mode during rush hour in urban cities, where multiple transportation modes interact with each other and share limited urban spaces. Therefore, efficient multimodal transportation systems need to be developed.

Substantial effort has been devoted to management strategies designed to alleviate these externalities.  A common strategy is congestion pricing, which promotes travel behavior changes, such as trip timing and transportation mode shifts. Despite the large body of literature on this subject  (e.g., \cite{yang2005mathematical}), it is practically difficult to implement such pricing schemes due to public acceptance \citep{giuliano1992assessment, gu2018congestion}.

Perimeter control, which regulates traffic inflow into a targeted area  by gating or traffic signal control to maintain good traffic conditions, is an alternative strategy of mitigating hypercongestion. Real-time perimeter control is enabled by monitoring the network-wide traffic state based on recent advances in traffic theories, more precisely the macroscopic fundamental diagram (MFD) \citep{daganzo2007urban}.   Perimeter control for multimodal transportation systems has been investigated (e.g., \cite{ampountolas2017macroscopic} and \cite{ haitao2019providing}). \cite{haitao2019providing} showed that the combination of perimeter control and transit priority is effective in bimodal systems in terms of passenger flow. However, travelers' behavior changes in response to such strategies, such as shifts in trip timing and transportation modes, have not been fully investigated. Perimeter control is intended to promote trip timing changes (to avoid queuing at perimeter boundaries) and transportation mode shift from cars to transit (which may aggravate discomfort externalities). Although hypercongestion and discomfort in transit are important for traveler decision making in urban cities \citep{arnott2013bathtub, tirachini2014multimodal, de2017economics}, the impact of such externalities on the effectiveness of perimeter control has been rarely studied. A critical gap is the lack of a methodology that connects travelers' departure time and transportation mode choices with the complex traffic flow and passenger dynamics in the implementation of perimeter control.

This paper investigates traveler's behavior changes in response to perimeter control with transit priority. To this end, we consider a bimodal transportation system that has cars and flexible route transit (FRT) vehicles, which have a fixed route but can pick up passengers at their origins in predetermined areas covered by each FRT line. Then, we model departure time and transportation mode choices in the presence of hypercongestion and FRT discomfort (called the mixed bimodal bathtub model). 

First, the dynamic user equilibrium is characterized. Unlike the 
existing works, we find that the number of FRT passengers decreases with time toward the desired arrival time and that FRT use may not be preferred around the peak of rush hour. Then, we investigate the equilibrium patterns in the implementation of perimeter control with transit priority.
Studies indicate that transit priority effectively promotes the transit use \citep{zheng2017macroscopic,dantsuji2021simulation}, but our results show that transit priority at perimeter boundaries may be insufficient and that added incentives, such as subsidy for lower fares, may be required to encourage FRT use during perimeter control. We also show  that operating many FRT vehicles does not always decreases the equilibrium cost, even under perimeter control  with transit priority.

\subsection{Literature review}

MFDs are a powerful tool of describing the network-wide traffic dynamics, which relates network flow (or trip completions) to network density (or accumulation of vehicles). The idea of macroscopic traffic flow theory was proposed by \cite{godfrey1969mechanism} and further investigated by \cite{mahmassani1987performance} and \cite{daganzo2007urban} and then empirically analyzed by \cite{geroliminis2008existence}. Traffic management approaches based on MFDs have been studied, such as congestion pricing \citep{zheng2012dynamic, simoni2015marginal, genser2022dynamic} and route guidance \citep{yildirimoglu2015equilibrium, yildirimoglu2018hierarchical}. MFDs have also been utilized for other purposes, such as the dynamic traffic demand estimation \citep{dantsuji2022novel} and network performance indicator \citep{loder2019understanding, hamm2022network}.  
To capture the interaction of dynamics between different transportation modes (e.g., car and bus), \cite{geroliminis2014three} proposed the three-dimensional MFD (3D-MFD),  which enables the representation of traffic flow in mixed bimodal networks. The multimodal MFD was empirically analyzed  \citep{loder2017empirics, dakic2018use, fu2020empirical, paipuri2021empirical}.  Traffic management strategies for  bimodal systems have also been studied, such as road space allocation \citep{zheng2013distribution, chiabaut2015evaluation, zheng2017macroscopic}, perimeter traffic control \citep{ampountolas2017macroscopic}, and pricing \citep{zheng2020area, dantsuji2021simulation, loder2022optimal}.

Perimeter control is a successful application of MFDs to traffic management. The idea of  perimeter control is to control the entry flow at the perimeter boundary of a targeted area to maximize the trip completion rate \citep{daganzo2007urban,haddad2014robust}.  It has been extended to the multiregion networks \citep{geroliminis2012optimal, haddad2012stability,ramezani2015dynamics}, perimeter control with boundary queue \citep{haddad2017optimal,ni2020city, guo2020macroscopic, li2021robust} and with route guidance \citep{sirmatel2017economic,ding2017traffic}, and in bimodal transportation systems \citep{ampountolas2017macroscopic, haitao2019providing,chen2022passenger}. Considerable effort has been dedicated to perimeter control schemes, but the impacts of traveler behavior has been insufficiently explored. Specifically. the effect of departure time and transportation mode shifts on the efficiency of perimeter control has rarely been investigated.

Bathtub models, namely dynamic user equilibrium model for departure time choice in urban cities with hypercongestion, have been examined \citep{small2003hypercongestion, geroliminis2009cordon, arnott2013bathtub, fosgerau2013hypercongestion,amirgholy2017modeling, arnott2018solving, jin2020generalized,vickrey2020congestion, bao2021leaving, chen2022modeling}. Bathtub models have also been extended to heterogeneity in trip length \citep{fosgerau2015congestion, lamotte2018morning}, cruising-for-parking \citep{geroliminis2015cruising, liu2016modeling}, and  staggered work schedules \citep{yildirimoglu2021staggered}. Simultaneous departure time and transportation mode choices models were also investigated for bimodal transportation systems \citep{gonzales2012morning,gonzales2013evening, gonzales2015coordinated}. However, none of them takes into accounts both hypercongestion and transit discomfort despite the importance of these externalities for traveler behavior, particularly in urban cities with limited spaces \citep{tirachini2014multimodal}.  

\subsection{Contributions}
In summary, our contributions are listed as follows:
\begin{itemize}
    \item We develop  the  mixed bimodal bathtub model, which describes commuters’ departure time and transportation model choices in the presence of hypercongestion and transit discomfort.
    \item We find that the number of FRT passengers decreases with time toward the desired arrival time and that FRT may not be used around the peak of rush hour  in contrast to the findings of the existing works.
    \item We show that transit priority may not be sufficient to promote transit use. Furthermore, we show that  additional incentives, such as subsidy for reduced fares, may be required to encourage FRT use during perimeter control and that operating many FRT vehicles does not always decrease the equilibrium cost, even under perimeter control with transit priority.
\end{itemize}

The remaining of this paper is organized as follows. Section \ref{sec:model} shows the development of the  bimodal bathtub model. In Section \ref{sec:eq_ue}, we characterize the model equilibrium. Section \ref{seq:perimeter} presents the formulation of the bimodal bathtub model during perimeter control with transit priority. The equilibrium conditions during perimeter control are studied in Section \ref{sec:equilibrium_perimeter}. Numerical examples of the model are provided in Section \ref{seq:numerical}, and we conclude in Section \ref{seq:conclusion}.

\section{Bimodal bathtub model} \label{sec:model}
\subsection{Model setting}
Consider a monocentric city that has central business district (CBD) and suburban zones where cars and FRT service\footnote{We use FRT because of the assumption that transit is available for all commuters in the model setting. While the transit network is dense and is accessible to everyone in the CBD zone, the accessibility to the transit service in the suburban zone is limited in reality. For example, \cite{velaga2012transport} showed that it takes 14 -- 26 minutes and  more than 27 minutes  to reach the nearest bus stop in 9 \% and   7 \% of the areas in the suburban zone near the CBD zone, respectively, and that 4 \% of the areas has no bus service in Scotland. Therefore, the assumptions that transit is available for all commuters, and that the trip length of the FRT commuters is homogeneous  are not realistic in  the traditional fixed route transit services.  Instead, we use flexible transit systems  \citep[e.g.,][]{sipetas2021continuous} that enable all commuters in the suburban zone to use transit.} are available for all commuters as transportation modes, as shown in Fig. \ref{fig:model}.  The CBD zone has homogeneous topological characteristics by proper partitioning methods (e.g., \cite{ji2012spatial, dantsuji2019cross}). The congestion pattern is homogeneous over space within the CBD, and the CBD zone exhibits the well-defined bimodal MFD. The congestion dynamics in the CBD is thus described as a bimodal bathtub model, whereas we assume that one can travel at the free-flow speed in the suburban zone. 

A fixed number of FRT vehicles are circulated in multiple directions in the suburban zone  such that  all commuters can access the service.  Each FRT line has a predetermined coverage area in the suburban zone and has a base route; passengers ride FRT vehicles at stops, and these vehicles can pick up passengers at their origins, in predetermined areas and drop them off at their destinations. (e.g., \cite{zheng2018slack,sipetas2021continuous}).  All FRT vehicles are coordinated and controlled by an operator to prevent bus bunching (e.g., \cite{daganzo2011reducing, bartholdi2012self}). Therefore, we assume that the FRT vehicles are distributed homogeneously over the transit lines such that  the number of FRT vehicles in the CBD zone is  fixed and time invariant ($n_F= n_F^s + n_F^c$, where $n_F$, $n_F^s$, and $n_F^c$ are the total number of FRT vehicles and the numbers of FRT vehicles in the suburban and CBD zones, respectively). 

We then assume  a fixed number, $N$,  of homogeneous commuters who have identical preference and desired arrival time. All of them travel from the suburban zone to the CBD zone by choosing departure times and transportation mode to minimize their travel cost.

\begin{figure}[t]
\centering
\includegraphics[width=0.5\columnwidth]{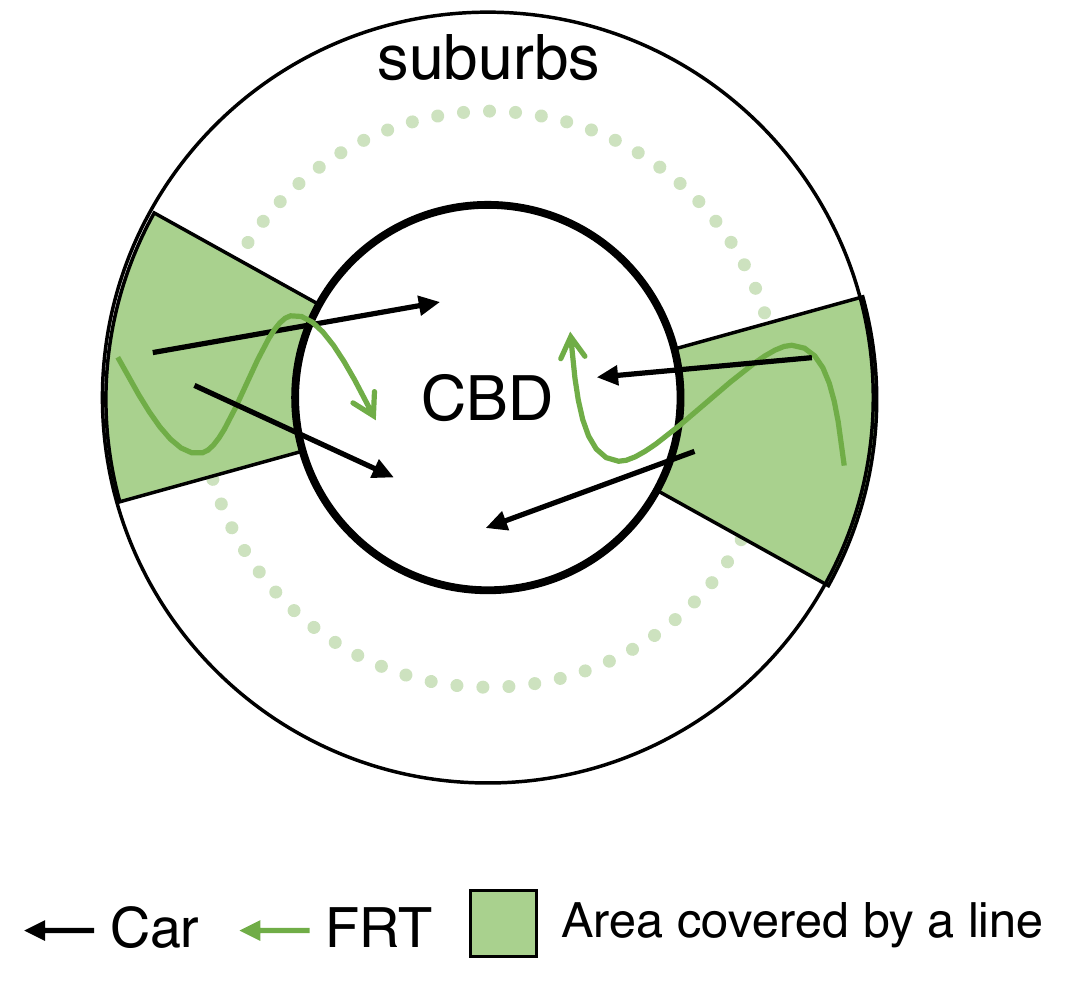}
\caption{Model structure}
\label{fig:model}
\end{figure}

\subsection{Bimodal bathtub congestion and passenger dynamics of FRT  in CBD}
To incorporate the bimodal bathtub model of the CBD zone, we extend the Greenshields model to the bimodal congestion dynamics as follows: 
\begin{eqnarray}
v_c(t) &=& v_f \left( 1 - \frac{\eta n_F(t) + n_c(t)}{n_j} \right) \label{eq:car_speed_bi_original}   \\
v_F(t) &=& m v_c(t) \label{eq:bus_speed_bi} 
\end{eqnarray}
where $v_i(t)$ is the space-mean speed of transportation mode $i$ ($i \in \{c, F\}$ where $c$ and $F$ represent car and FRT, respectively) at time $t$, $v_f$ is the free-flow car space-mean speed, $n_i(t)$ is the accumulation of mode $i$'s vehicles in the CBD zone at time $t$, $\eta$ is the passenger car unit, $n_j$ is jam accumulation and $m$ indicates that transit vehicles travel slower than cars due to passenger boarding and alighting ($0<m<1$), which is  similar to the  modeling of \cite{loder2017empirics}. 

As mentioned above,  we assume that FRT vehicle accumulation in the CBD zone is  fixed; that is $n_F(t)=n_F^c$. Thus, Eq. (\ref{eq:car_speed_bi_original}) can be rewritten as  
\begin{eqnarray}
v_c(t) &=&  v_f' \left(  1 - \frac{ n_c(t)}{n_j'} \right)   \label{eq:car_speed_bi} 
\end{eqnarray}
where $v_f'=v_f \left( 1- \eta n_F^c / n_j \right)$ and $n_j' = n_j \left( 1- \eta n_F^c / n_j \right)$. 

Since the CBD is modeled as a system that has  inflow and outflow and whose traffic conditions are governed by bathtub congestion dynamics \citep{daganzo2007urban}, the time evolution of car accumulation, $\dot{n_c}(t)$, is given by 
\begin{eqnarray}
\dot{n_c}(t) = I_c(t) - G_c(t) \label{eq:TimeEvoCar}
\end{eqnarray}
where $I_c(t)$ is the car inflow rate to the CBD at time $t$ and $G_c(t)$ is the car passenger arrival rate at the destination at time $t$. The car passenger arrival rate is formulated by the network exit function  \citep{gonzales2012morning}  as  
\begin{eqnarray}
G_c(t) = \frac{n_c(t)v_c(t)}{L_c} \label{eq:trip_completion_rate}
\end{eqnarray}
where $L_c$ is the average trip length of the car commuters in the CBD zone. 

From the perspective of passenger flow, we assume that all commuters have their own car, and the car commuters travel by their own car (i.e., the occupancy of cars is $1$ [pax/veh]). With $O_F(t)$ denoting the average number of  passengers per FRT vehicle at time $t$, the FRT passenger arrival rate, $G_F^p(t)$, is given by
\begin{eqnarray}
G_F^p(t) = \frac{O_F(t)n_F^c v_F(t)}{L_F}
\end{eqnarray}
where $L_F$ is the average trip length of the FRT commuters in the CBD zone. 
In analogy with the evolution of car accumulation, the evolution of the average number of the passengers per a FRT vehicle, $\dot{O_F}(t)$,  is given by
\begin{eqnarray}
\dot{O_F}(t) = \frac{1}{n_F} \left(d_F(t) - \frac{O_F(t) n_F^c v_F(t)}{L_F}\right) 
\end{eqnarray}
where $d_F(t)$ is the departure rate of FRT passengers at time $t$. The term in the bracket is the difference between the numbers of the passengers boarding and alighting at time $t$. Therefore, the time evolution of $O_F(t)$ is calculated by dividing it by the total number of FRT vehicles in the monocentric city, $n_F$.

Travel time is time spent in the CBD zone to complete the trip length. Hence, the travel time of each mode $i $ in the CBD zone is
\begin{eqnarray}
\int_{t}^{t+T_i(t)} v_i(u) du = L_i, \; \; \; \; i \in \{c,F\} \label{eq:proper}
\end{eqnarray}
With Eq.(\ref{eq:proper}), the departure time choice model in the bathtub is intractable \citep{arnott2013bathtub}. Thus, we assume that the travel time in the CBD zone of a commuter who arrives at time $t$ by transportation mode $i$ is approximated by
\begin{eqnarray}
T_i(t) \approx \frac{L_i}{v_i(t)} \label{eq:ass}
\end{eqnarray}
There are several ways of the approximation for the tractability in the literature \citep{arnott2013bathtub, fosgerau2013hypercongestion}. In this study,  we assume that the travel time is determined by a single instant of time \citep{small2003hypercongestion, geroliminis2009cordon}.

\subsection{Travel cost}
Given the bathtub congestion dynamics above, the car travel cost incurred by a commuter who arrive at time $t$ is given by 
\begin{eqnarray}
C_c(t) &= &
\begin{cases}
\alpha \left( T_c^s + T_c(t) \right) + \beta \left( t^* - t \right) + p_c  & \text{if }t \leq t^* \\
\alpha  \left( T_c^s + T_c(t) \right) + \gamma \left(  t - t^*  \right) + p_c & \text{if }t > t^*
\end{cases} 
\end{eqnarray}
where $T_c^s$ is the free-flow car travel time in the suburban  zone, $T_c(t)$ is the travel time in the CBD zone, as defined by Eq.(\ref{eq:ass}) and $t^*$ is the desired arrival time. We assume that commuters have ``$\alpha$ - $\beta$ - $\gamma$'' type preference \citep{arnott1993structural}. The first, second, and third terms of the RHS are
the travel time cost, the schedule delay cost, and the car operation cost (e.g., parking fee, gasoline expense), respectively. The travel time in the suburban zone and the car operation cost are time invariant. For  simplicity, we combine the two costs as $F_c$ (i.e., $F_c = \alpha T_c^s + p_c$) to obtain
\begin{eqnarray}
C_c(t)  &=&
\begin{cases}
\alpha  T_c(t)  + \beta \left( t^* - t \right) + F_c  & \text{if }t \leq t^* \\
\alpha  T_c(t) + \gamma \left(  t - t^*  \right) + F_c & \text{if }t > t^*
\end{cases}
\label{eq:gc_car} 
\end{eqnarray}

The FRT travel cost incurred by a commuter who arrives at time $t$ is
\begin{eqnarray}
C_F(t) &=& 
\begin{cases}
\alpha \left( T_F^s + T_F(t) \right) + \beta \left( t^* - t \right) + g(t)  + p_F & \text{if }t \leq t^* \\
\alpha  \left( T_F^s + T_F(t) \right) + \gamma \left(  t - t^*  \right) + g(t) + p_F & \text{if }t > t^*
\end{cases}
\end{eqnarray}
where  $T_F^s$ is the free-flow FRT travel time in the suburban zone, $T_F(t)$ is the FRT travel time in the CBD zone, as defined by Eq. (\ref{eq:ass}), $g(t)$ is the discomfort cost at time $t$ and $p_F$ is the FRT fare. In addition to  travel time  and  schedule delay costs, an FRT commuter incurs the discomfort cost and the FRT fare. For  simplicity, 
the discomfort cost incurred by an FRT commuter is assumed to be proportional to the average number of passengers per  FRT vehicle when this commuter arrives at their destinations.
\begin{eqnarray}
g(t) = \lambda O_F(t)
\end{eqnarray}
where $\lambda$ is the marginal disutility of the discomfort, $O_F(t)$ is the average number of passengers per FRT vehicle at time $t$. 
This assumption indicates that discomfort in  FRT is related to the number of passengers in the FRT vehicles. It can represent the waiting time for pickup \citep{basso2019efficiency} or the boarding and alighting times of other passengers.  A linear function is commonly assumed for the discomfort cost \citep{wu2014equilibrium, xu2018pareto}. Futhermore, the assumption that this cost is proportional to the average number of passengers per FRT vehicle when a commuter arrives at their destination is consistent with the travel time assumption (the travel time is determined by the arrival time at the destination (Eq. (\ref{eq:ass})). 

By combining the travel cost in the suburban zone and the FRT fare (i.e., $F_F=\alpha T_F^s + p_F$), we obtain 
\begin{eqnarray}
C_F(t) &=&
\begin{cases}
\alpha T_F(t) + \beta \left( t^* - t \right) + \lambda O_F(t) + F_F & \text{if }t \leq t^* \\
\alpha   T_F(t) + \gamma \left(  t - t^*  \right) + \lambda O_F(t) + F_F & \text{if }t > t^*
\end{cases}
\label{eq:gc_bus}
\end{eqnarray}

Additionally, to analyze the equilibrium when both modes are used during rush hour, we assume that the average FRT trip length is longer than that of a car (i.e., $L_F > L_c$), and that the marginal cost of earliness is less than that of the travel time (i.e., $\beta < \alpha $)   for the FIFO property.

\section{Equilibrium} \label{sec:eq_ue}
We define dynamic user equilibrium following \cite{wardrop1952road}'s first principle. That is, at equilibrium, no commuter can reduce their travel cost by changing their departure time and transportation mode. The equilibrium conditions are  
\begin{subequations}
 \begin{align}
&\begin{cases}
C_c(t) = c^* & \text{if   } n_c(t) > 0 \\
C_c(t) \geq c^* & \text{if   } n_c(t) = 0 \\
\end{cases} \; \; \;  \forall t \in \mathbb{R}  \label{eq:UE_con_car}\\ 
&\begin{cases}
C_F(t) = c^* & \text{if   } O_F(t) > 0 \\
C_F(t) \geq c^* & \text{if   } O_F(t) = 0 \\
\end{cases}  \; \; \;  \forall t \in \mathbb{R}  \label{eq:UE_con_bus}  \\ 
&\int_{t \in \mathbb{R} } \frac{n_c(s)v_c(s)}{L_c} ds + \int_{t \in \mathbb{R} } \frac{O_F(s)n_F^c v_F(s)}{L_F} ds = N
\label{eq:UE_cons}
 \end{align}
\end{subequations}
where $c^*$ is the equilibrium cost. Condition (\ref{eq:UE_con_car}) states that if the car travel cost at time $t$ is greater than the equilibrium cost, no one will arrive at their destination by car at time $t$. According to Condition (\ref{eq:UE_con_bus}), if the FRT travel cost at time $t$ is greater than the equilibrium cost, no one will take FRT at time $t$. Condition (\ref{eq:UE_cons}) is the conservation law for travel demand: integrating the arrival rate is total number of commuters. 

\subsection{Equilibrium properties}
Let [$t_s^c$, $t_e^c$] and  [$t_s^F$,  $t_e^F$]  be the time windows  between the  first and last commuters who arrive at their destinations by car and FRT, respectively. That is, these time windows represent the rush hour of each transportation mode. We define $\theta \equiv \frac{ (c^* - F_c)  v_f^{'}}{\alpha L_c}$ which is the ratio of the sum of the travel time  and schedule delay costs at equilibrium  to the free-flow travel time (sum of travel time  and schedule delay costs at equilibrium when the number of car commuters is zero) \citep{arnott2013bathtub}. It represents the duration of the car rush hour (or the severity of traffic congestion); that is, the larger the $\theta$, the longer the car rush hour. 
We present the following proposition on the equilibrium when both modes are used. 
\begin{prop} \label{prop:FRT_used_ue}
If  $\theta>1$ and  $\Delta F > \alpha  \Delta T_f $ where $\Delta F = F_c - F_F$ and $\Delta T_f = L_F /  m v_f' - L_c /  v_f'$ , then both modes are used during rush hour, and the FRT rush hour starts earlier and ends later than the car rush hour.  If $\Delta F \leq \alpha  \Delta T_f$,  FRT is never used. If $\theta \leq 1$, cars are never used.
\end{prop}
\begin{proof}
As shown in \cite{small2003hypercongestion}, the time evolution of car accumulation at equilibrium with Eq.~(\ref{eq:ass}) is given by 
\begin{eqnarray}
&& n_c(t) = 
\begin{cases}
n_j' \left(1 - \frac{1}{1+ \frac{\beta}{\alpha T_f^c} \left(t - t_s^c \right)} \right) & \text{if } t \leq t^* \\ 
n_j' \left(1 - \frac{1}{1+ \frac{\gamma}{\alpha T_f^c  }\left(t_e^c - t \right)} \right)& \text{if } t^* < t 
\end{cases} \label{eq:car_evo_ue}
\end{eqnarray}

The first car commuter incurs the travel time cost at free-flow in the CBD. Thus, the car travel cost at time $t_s^c$ is given by 
\begin{eqnarray}
C_c(t_s^c)=\alpha T_f^c  + \beta \left( t^* - t_s^c \right) + F_c 
\end{eqnarray}
where $T_f^c$ is the car free-flow travel time ($=  L_c/ v_f'$). The FRT travel cost at time $t_s^c$ is
\begin{eqnarray}
C_F(t_s^c)=\alpha  T_f^F + \beta \left( t^* - t_s^c \right) + \lambda O_F(t_s^c) + F_F
\end{eqnarray}
where $T_f^F$ is the FRT free-flow travel time ($= L_F / m v_f'$). As seen from Eqs. (\ref{eq:bus_speed_bi}), (\ref{eq:car_speed_bi}), (\ref{eq:ass}), and (\ref{eq:car_evo_ue})   that $\frac{d T_c(t)}{dt}=\frac{\beta}{\alpha}$ and $\frac{d T_F(t)}{dt}=\frac{T_f^F}{T_f^c}\frac{\beta}{\alpha}$ for  earliness, which yields $\frac{d T_c(t)}{dt}<\frac{d T_F(t)}{dt}$ by the assumptions that $L_F > L_c$ and $0<m<1$, which indicates that the FRT travel time cost increases with time more than the car travel time cost, once the car rush hour starts. Thus,  FRT is never used if no one takes FRT at the car rush hour start time, $t_s^c$. When both modes are used, the FRT travel cost without the discomfort cost is less than the car travel cost at  time $t_s^c$. 
\begin{eqnarray}
\alpha  T_f^c + \beta \left( t^* - t_s^c \right) + F_c   > \alpha  T_f^F + \beta \left( t^* - t_s^c \right)  + F_F
\end{eqnarray}
It yields
\begin{eqnarray}
\Delta F > \alpha  \Delta T_f \label{eq:prop_1}
\end{eqnarray}

Before the car rush hour starts (travel time is time-invariant at free-flow), we have   $\frac{dC_F(t)}{dt} = \beta + \lambda \dot{O_F}(t)$ from Eq.~(\ref{eq:gc_bus}). Then, Condition (\ref{eq:UE_con_bus}) yields 
\begin{eqnarray}
\dot{O_F}(t)=\frac{\beta}{\lambda} \label{eq:before_car_rush_der_Od}
\end{eqnarray}
Therefore, the time, $t_s^F$, when the FRT rush hour starts ($O_F(t_s^F)=0$) is earlier than $t_s^c$. In the same way, we can prove for  lateness that $t_e^F>t_e^c$.

Obviously, cars are never used  if $c^* \leq  \alpha L_c /  v_f' + F_c $. This condition can be written by 
\begin{eqnarray}
\theta  \leq 1
\end{eqnarray}

\end{proof}

This proposition indicates that both modes will be used if the fixed cost difference ($\Delta F$) is greater than the free-flow travel time cost difference ($\alpha \Delta T_f$). As the FRT travel cost is lower than the car travel cost outside both modes' rush hour, the FRT rush hour starts earlier and ends later. Then, since  FRT commuters incur the discomfort cost in addition to the costs they have in common with car commuters (travel time, schedule delay, and fixed costs), there is a time when the car travel cost becomes equal to the FRT travel cost. Thus, both modes will be used. Contrarily, if the fixed cost difference is lower than the free-flow travel cost difference, the FRT travel cost will always be larger than the car travel cost; thus, FRT will not be used. Furthermore, if the car fixed cost and free-flow travel time are too high, no one will use their cars and all commuters will take FRT. Since the travel time is time-invariant at free-flow during the FRT rush hour in this case, the number of  FRT passengers increases linearly toward the desired arrival time.  

As this study aims to analyze equilibrium patterns when both modes are used, we impose the following assumption for the remainder of this paper. 
\begin{ass} \label{ass:ass1}
$\theta>1$ and $\Delta F > \alpha  \Delta T_f $ 
\end{ass}

At equilibrium when both modes are used, we have the following proposition. 

\begin{prop} \label{prop:FRT_equi_pattern}
Suppose Assumption \ref{ass:ass1}.  Then, the equilibrium has the following properties. 
\begin{enumerate}
    \item For earliness, the number of FRT passengers increases with time before the car rush hour and begins to decrease toward the desired arrival time once the car rush hour starts. For lateness, it increases after the desired arrival time and starts to decrease after the car rush hour ends. 
    \item If $ \frac{\Delta F}{\alpha\Delta T_f }< \theta$, there will be a time window wherein the FRT is not used during the FRT rush hour, and the time window $[t_{ee}^F, t_{sl}^F]$ is given by
    \begin{eqnarray}
    t_{ee}^F = t_s^c + \frac{\alpha T_f^c}{\beta} \frac{\Delta F - \alpha \Delta T_f}{\alpha \Delta T_f} \label{eq:time_FRT_not_used_earliness} \\
    t_{sl}^F = t_e^c - \frac{\alpha T_f^c}{\gamma} \frac{\Delta F - \alpha \Delta T_f}{\alpha \Delta T_f} \label{eq:time_FRT_not_used_lateness} 
    \end{eqnarray}
    \item If $ \frac{\Delta F}{\alpha\Delta T_f }  \geq \theta $,  FRT will be used from the beginning to the end of the FRT rush hour. 
\end{enumerate}

\end{prop}
\begin{proof}
$\dot{O_F}(t)= \beta / \lambda$ if $ t_s^F < t < t_s^c $, as stated in Eq.~(\ref{eq:before_car_rush_der_Od}). Thus, ${O_F}(t)$ increases linearly before the car rush hour.  
During the car rush hour, according to Conditions~(\ref{eq:UE_con_car}) and (\ref{eq:UE_con_bus}), we have $C_c(t)=C_F(t)$ when both modes are used. For earliness, combining this with Eqs.~(\ref{eq:gc_car}) and (\ref{eq:gc_bus}) yields
\begin{eqnarray} \label{eq:Of_earliness}
O_F(t) = \frac{1}{\lambda} \left(\Delta F - \alpha \left(T_F(t) - T_c(t) \right) \right)
\end{eqnarray}
By substituting Eqs.~(\ref{eq:car_speed_bi}), (\ref{eq:ass}), and (\ref{eq:car_evo_ue}) into Eq.~(\ref{eq:Of_earliness}), we obtain
\begin{eqnarray} \label{eq:od_earliness_ue}
O_F(t) = \frac{1}{\lambda} \left( \Delta F - \alpha \Delta T_f \left( 1 + \frac{\beta}{\alpha T_f^c} \left( t - t_s^c \right) \right) \right) 
\end{eqnarray}
As $\dot{O_F}(t) < 0$ after the car rush hour starts, ${O_F}(t)$ decreases linearly. 

For lateness, we have 
\begin{eqnarray}\label{eq:od_lateness_ue}
O_F(t) = \frac{1}{\lambda} \left( \Delta F - \alpha \Delta T_f \left( 1 + \frac{\gamma}{\alpha T_f^c} \left(t_e^c - t\right) \right) \right) 
\end{eqnarray}
Thus, ${O_F}(t)$ increases linearly before the car rush hour ends. Then, $\dot{O_F}(t)= - \gamma / \lambda$ for the time after the car rush hour. Therefore, for earliness, the number of FRT passengers increases with time before the car rush hour and decreases toward the desired arrival time once the car rush hour starts.  Then, for lateness, it increases after the desired arrival time and decreases after the car rush hour ends.

Since the number of FRT passengers decreases after the car rush hour starts, there may be a time wherein $O_F(t) = 0$ before the desired arrival time. Let $t_{ee}^F$ denote this time. From Eq. (\ref{eq:od_earliness_ue}), we obtain
\begin{eqnarray}
t_{ee}^F = t_s^c + \frac{\alpha T_f^c}{\beta} \frac{\Delta F - \alpha \Delta T_f}{\alpha \Delta T_f}  \label{eq:time_FRT_not_used_earliness2} 
\end{eqnarray}
In the same way, we have $t_{sl}^F = t_e^c - \frac{\alpha T_f^c}{\gamma} \frac{\Delta F - \alpha \Delta T_f}{\alpha \Delta T_f}$ for lateness. 

The use of FRT will continue during the FRT rush hour if $t_{ee}^F>t^*$. Condition (\ref{eq:UE_con_car}) states that  $c^*=\alpha T_f^c + \beta (t^* - t_s^c) + F_c$, combining this condition with Eq.~(\ref{eq:time_FRT_not_used_earliness2}) leads to the following condition for $t_{ee}^F<t^*$: 
\begin{eqnarray}
\frac{\Delta F}{\alpha\Delta T_f }< \theta
\end{eqnarray}
which completes the proof.
\end{proof}

\begin{figure}[t]
\centering
\includegraphics[width=0.8\columnwidth]{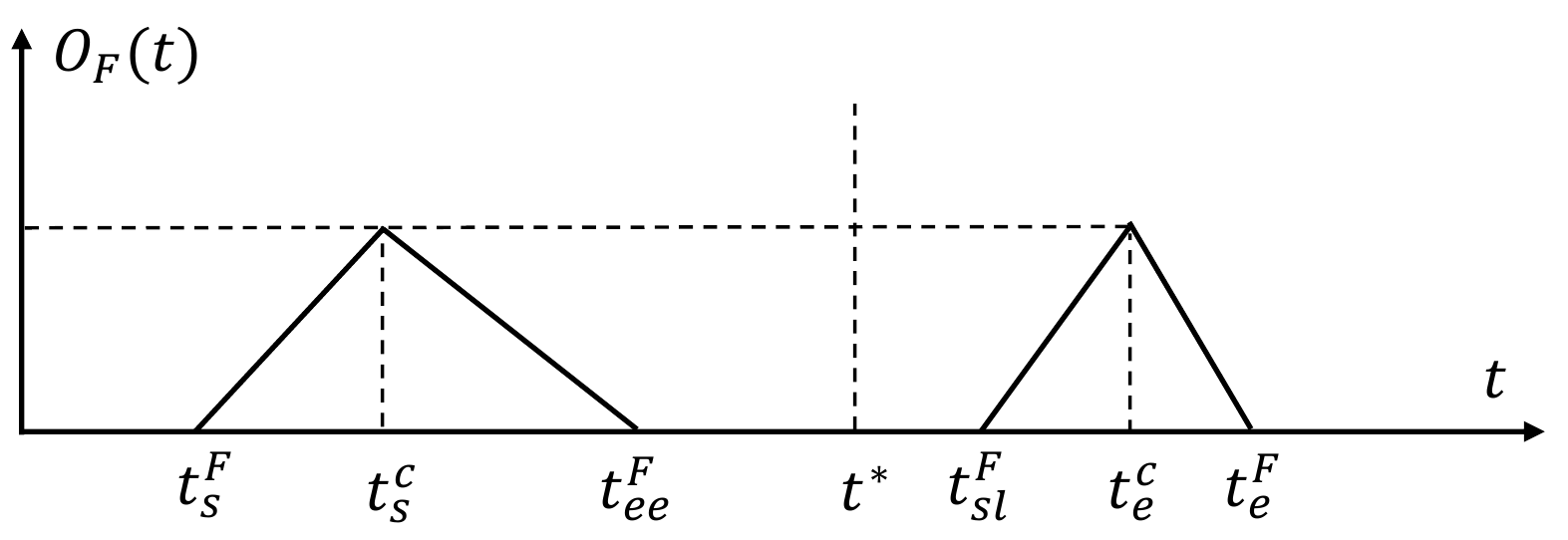}
\caption{Illustration of Proposition \ref{prop:FRT_equi_pattern}.2}
\label{fig:graphical_FRT}
\end{figure}

\begin{figure}[t]
\centering
\includegraphics[width=0.8\columnwidth]{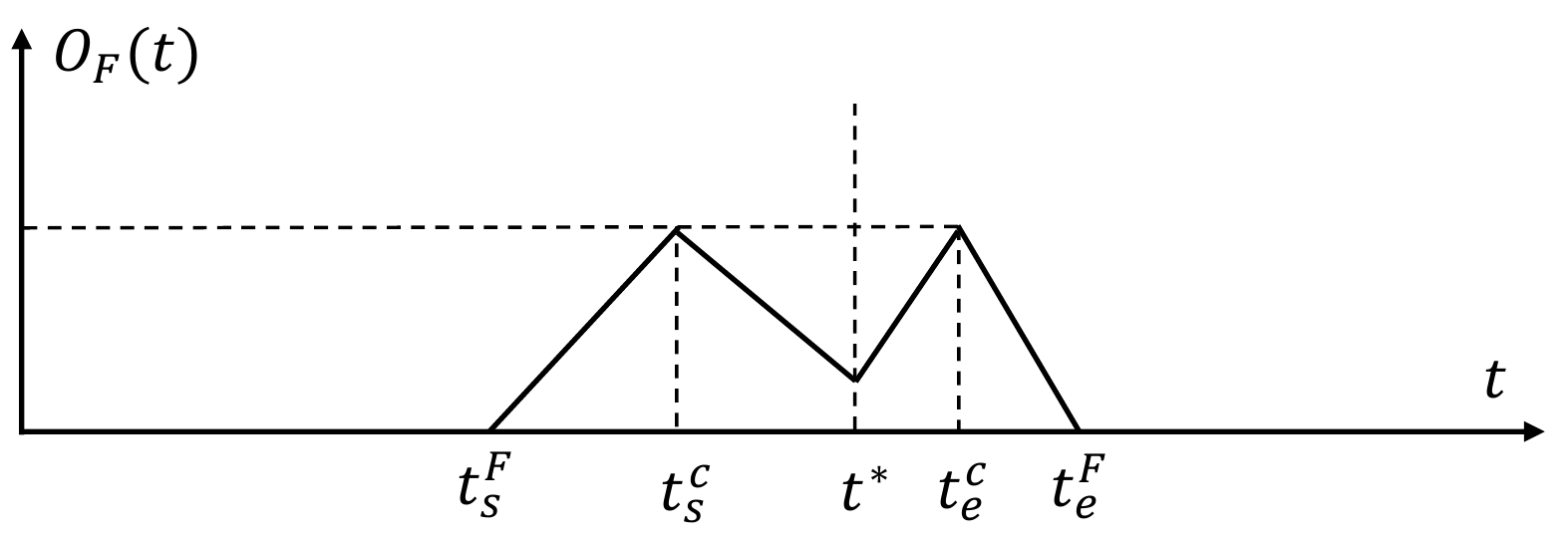}
\caption{Illustration of Proposition \ref{prop:FRT_equi_pattern}.3}
\label{fig:graphical_FRT_continue_ue}
\end{figure} 

Propositions \ref{prop:FRT_equi_pattern}.2 and \ref{prop:FRT_equi_pattern}.3  are depicted in Figs. \ref{fig:graphical_FRT} and \ref{fig:graphical_FRT_continue_ue}, respectively.  For earliness, the number of FRT passengers increases with time before the car rush hour, as shown in Eq.~(\ref{eq:before_car_rush_der_Od}). During the time interval $[t_s^F, t_s^c)$, the only externality that they incur is discomfort. Thus, the trade-off between discomfort and schedule delay causes this increase in FRT passengers with time. As seen from Eq.~(\ref{eq:od_earliness_ue}), once the car rush hour starts,  the  number of FRT passengers begins to decrease; it will become zero at time $t_{ee}^F$ (before the desired arrival time) if the fixed cost difference is relatively low compared with the travel time and schedule delay costs ($  \Delta F/\alpha\Delta T_f < \theta$). Note that as defined above, a large $\theta$ represents a long car rush hour (large travel time around the peak of the rush hour).   Since the travel time cost increases due to the increase in car accumulation  once the car rush hour starts (Eq.~(\ref{eq:car_evo_ue})), the discomfort cost has to be decreased to keep the travel cost at the equilibrium cost, and the number of FRT passengers will decrease. Then, FRT will not be used during the time window $t \in [t_{ee}^F, t_{sl}^F)$ if $  \Delta F/\alpha\Delta T_f < \theta$.  Similarly, for lateness, the average number of passengers per FRT vehicle increases during the time interval $[t_{sl}^F, t_e^c]$ and begins to decrease after the car rush hour. If $\Delta F/\alpha\Delta T_f   \geq \theta $ (e.g., the fixed car cost is significantly higher than the FRT fixed cost ), there will be no time window wherein  FRT is not used, as depicted in Fig.~\ref{fig:graphical_FRT_continue_ue}. 

These results differ from previous findings on bimodal morning commute problems that consider  dynamic congestion for  transit  (or hypercongestion), where  transit is used continuously once the transit rush hour starts and the number of  transit passengers is constant or increases toward the desired arrival time \citep{ gonzales2012morning,basso2019efficiency}. This difference is caused by the two externalities incorporated in the proposed model.  \cite{gonzales2012morning} assumed that the transit travel cost  is constant by dedicated lanes and did not consider any externality for transit.  \cite{basso2019efficiency} considered the externality of the waiting time at bus stops, but transit vehicles in their study were not subjected to bottleneck congestion\footnote{If hypercongestion was ignored for  FRT in the proposed model, results similar to \cite{basso2019efficiency} could be obtained. 
}. Hypercongestion and discomfort in transit (e.g., crowding and waiting time at stops)  in urban cities are not negligible (e.g., \cite{tirachini2014multimodal}). Hence, these externalities play a role in commuters' travel behavior changes. However, empirical measurements that combine urban hypercongestion, discomfort, and transit user travel behavior have not been investigated. Nonetheless, existing estimation methods of hypercongestion (e.g., \cite{geroliminis2008existence}) and massive smart IC card data (e.g., \cite{dantsuji2022novel}) may reveal those relationships, which is an interesting future direction of the current paper. 

\subsection{Equilibrium cost}

Next, we derive the equilibrium cost for the case where there is a time window when FRT is not used  ($ \Delta F/\alpha\Delta T_f < \theta $) . 
As derived in \cite{small2003hypercongestion}, the number of car commuters at equilibrium, $N_c$, is  
\begin{eqnarray}
N_c &=& \int_{t_s^c}^{t_e^c} \frac{n_c(s)v_c(s)}{L_c} ds \nonumber \\ 
&=& n_j'  \left( \frac{\alpha}{\beta} + \frac{\alpha}{\gamma}\right)    \left(   \ln \theta + \frac{1}{\theta} - 1  \right) \label{eq:car_number_ue}
\end{eqnarray}
As shown above, we obtain
\begin{eqnarray} \label{eq:ue_Od_pattern_case1}
O_F(t) =
\begin{cases}
\frac{\beta}{\lambda} (t-t_s^F) & \text{if   }  t_s^F \leq t < t_s^c\\
\frac{1}{\lambda} \left( \Delta F - \alpha \Delta T_f \left( 1 + \frac{\beta}{\alpha T_f^c} \left( t - t_s^c \right) \right) \right)  & \text{if   } t_s^c \leq t <  t_{ee}^F \\
\frac{1}{\lambda} \left( \Delta F - \alpha \Delta T_f \left( 1 + \frac{\gamma}{\alpha T_f^c} \left(t_e^c - t\right) \right) \right) & \text{if   }  t_{ls}^d \leq t < t_e^c \\
\frac{\beta}{\lambda} (t-t_s^F) & \text{if   } t_e^c \leq t < t_e^d
\end{cases}
\end{eqnarray}
Therefore, $O_F(t_s^c) = O_F(t_e^c) =\frac{1}{\lambda} \left( \Delta F - \alpha \Delta T_f \right)$, and the FRT rush hour start and end times are 
\begin{eqnarray}
t_s^c - t_s^F = \frac{1}{\beta} \left( \Delta F - \alpha \Delta T_f \right) \\
 t_e^F - t_e^c = \frac{1}{\gamma} \left( \Delta F - \alpha \Delta T_f \right)
\end{eqnarray}
We can also derive the FRT speed from Eqs.~(\ref{eq:bus_speed_bi}), (\ref{eq:car_speed_bi}), and (\ref{eq:car_evo_ue}) as
\begin{eqnarray}
v_F(t) = 
\begin{cases}
mv_f' & \text{if   }  t < t_s^c \\ 
mv_f' \left( 1 + \frac{\beta}{\alpha T_f^c} (t-t_s^c) \right)^{-1} & \text{if   } t_s^c \leq t < t^*\\
mv_f' \left( 1 + \frac{\gamma}{\alpha T_f^c} (t_e^c-t) \right)^{-1} & \text{if   } t^* \leq t < t_e^c \\
mv_f' & \text{if   }  t_e^c \leq t
\end{cases}
\end{eqnarray}
By integrating the arrival rates in these time windows, we determine the total number of FRT commuters at equilibrium as
\begin{eqnarray}
N_F &=& \int_{t_s^F}^{t_e^F} \frac{O_F(s)n_F^c v_F(s)}{L_F} ds \nonumber \\ 
&=&   \left( \frac{1}{\beta} + \frac{1}{\gamma}\right) \left(   \frac{n_F^c}{2\lambda}  \frac{1}{T_f^F} \left( \Delta F - \alpha \Delta T_f \right)^2 +  n_F^c \frac{\alpha}{\lambda} \frac{T_f^c}{T_f^F}   \left( \Delta F \ln  \frac{\Delta F }{\alpha \Delta T_f} - \left(\Delta F -  \alpha \Delta T_f  \right) \right) \right)
\label{eq:FRT_number_ue}
\end{eqnarray}
From Eqs.~(\ref{eq:car_number_ue}) and (\ref{eq:FRT_number_ue}), we have
\begin{eqnarray} \label{eq:total_number_ue}
N =  \left( \frac{1}{\beta} + \frac{1}{\gamma}\right) \left( \alpha n_j'  \left(   \ln \theta + \frac{1}{\theta} - 1  \right) +  \frac{n_F^c}{2\lambda}  \frac{1}{T_f^F} \left( \Delta F - \alpha \Delta T_f \right)^2 +  n_F^c \frac{\alpha}{\lambda} \frac{T_f^c}{T_f^F}    \left( \Delta F \ln  \frac{\Delta F }{\alpha \Delta T_f} - \left(\Delta F -  \alpha \Delta T_f  \right) \right) \right)
\end{eqnarray}
where $\theta \equiv \frac{ (c^* - F_c)  v_f^{'}}{\alpha L_c}$. As all variables except the equilibrium cost $c^*$ are exogenous, the equilibrium cost can be solved numerically as in \cite{small2003hypercongestion}. 

For the case where FRT is continuously used ( $ \frac{\Delta F}{\alpha\Delta T_f }  \geq \theta$), by replacing $t_{ee}^F$ and $t_{sl}^F$ with $t^*$ in Eq.~(\ref{eq:ue_Od_pattern_case1}), we obtain
\begin{eqnarray}
O_F(t) =
\begin{cases}
\frac{\beta}{\lambda} (t-t_s^F) & \text{if   }  t_s^F \leq t < t_s^c\\
\frac{1}{\lambda} \left( \Delta F - \alpha \Delta T_f \left( 1 + \frac{\beta}{\alpha T_f^c} \left( t - t_s^c \right) \right) \right)  & \text{if   } t_s^c \leq t <  t^* \\
\frac{1}{\lambda} \left( \Delta F - \alpha \Delta T_f \left( 1 + \frac{\gamma}{\alpha T_f^c} \left(t_e^c - t\right) \right) \right) & \text{if   }  t^* \leq t < t_e^c \\
\frac{\beta}{\lambda} (t-t_s^F) & \text{if   } t_e^c \leq t < t_e^F
\end{cases}
\end{eqnarray}
Thus, the total number of FRT commuters is
\begin{eqnarray}
N_F &=& \int_{t_s^F}^{t_e^F} \frac{O_F(s)n_F^c v_F(s)}{L_F} ds \nonumber \\ 
&=&   \left( \frac{1}{\beta} + \frac{1}{\gamma}\right) \left(   \frac{n_F^c}{2\lambda}  \frac{1}{T_f^F} \left( \Delta F - \alpha \Delta T_f \right)^2 +  n_F^c \frac{\alpha}{\lambda} \frac{T_f^c}{T_f^F}   \left( \Delta F \ln \theta  + \alpha \Delta T_f  \left( \theta - 1 \right) \right) \right)
\label{eq:FRT_number_continue_ue}
\end{eqnarray}
As the calculation of the number of car commuters is the same as Eq.~(\ref{eq:car_number_ue}),  we have
\begin{eqnarray} \label{eq:total_number_ue_continous}
N =  \left( \frac{1}{\beta} + \frac{1}{\gamma}\right) \left( \alpha n_j'  \left(   \ln \theta + \frac{1}{\theta} - 1  \right) +  \frac{n_F^c}{2\lambda}  \frac{1}{T_f^F} \left( \Delta F - \alpha \Delta T_f \right)^2 +  n_F^c \frac{\alpha}{\lambda} \frac{T_f^c}{T_f^F}   \left( \Delta F \ln \theta  + \alpha \Delta T_f  \left( \theta   - 1 \right) \right)  \right)
\end{eqnarray}
where $\theta \equiv \frac{ (c^* - F_c)  v_f^{'}}{\alpha L_c}$. As in Eq.~(\ref{eq:total_number_ue}), all parameters except the equilibrium cost $c^*$ are exogenous. Therefore, the equilibrium cost can be solved numerically  with respect to the equilibrium cost $c^*$.  

In summary, the equilibrium cost can be solved numerically,  regardless of FRT usage, and  has the following property:

\begin{prop} \label{prop:uniqueness_ue}
The equilibrium cost is uniquely determined, regardless of the existence of  a time window when FRT is not used during the FRT rush hour. 
\end{prop}
\begin{proof}
See Appendix A. 
\end{proof}

FRT use is restricted due to  hypercongestion and discomfort, particularly around the peak of  rush hour. A large share of transit services, such as FRT, is a key component of efficient and,
sustainable transportation systems. Therefore, in the next section, we discuss the management strategy of  perimeter control with  transit priority to promote FRT use.

\section{Perimeter control with transit priority} \label{seq:perimeter}
During perimeter control with transit priority, the inflow rate to the CBD is restricted to protect the zone from hypercongestion. Therefore, a queue will develop outside the perimeter boundary if the arrival rate of cars at the boundary exceeds the inflow rate to the CBD zone, whereas FRT vehicles can pass the perimeter boundary freely by using dedicated lanes at the perimeter (Fig. \ref{fig:model_pc}). There are certain strategies for perimeter control with transit priority. One example is  similar to the strategy proposed by \cite{haitao2019providing}. In their strategy, a pre-signal is installed at the upstream point of the main traffic signal to control the inflow into the CBD zone, and the lanes are separated between cars and FRT vehicles at the upstream link of the pre-signal point as depicted in Fig. \ref{fig:idea}. The strategy is that  the pre-signal for FRT  stays green, whereas that for cars turns from green to red for perimeter control to maintain the maximum throughput in the CBD zone.  Thus, the FRT vehicles can pass the perimeter control freely regardless of the existence of boundary queues of cars.

\begin{figure}[t]
\centering
\begin{minipage}{0.45\columnwidth}
\includegraphics[width=\columnwidth]{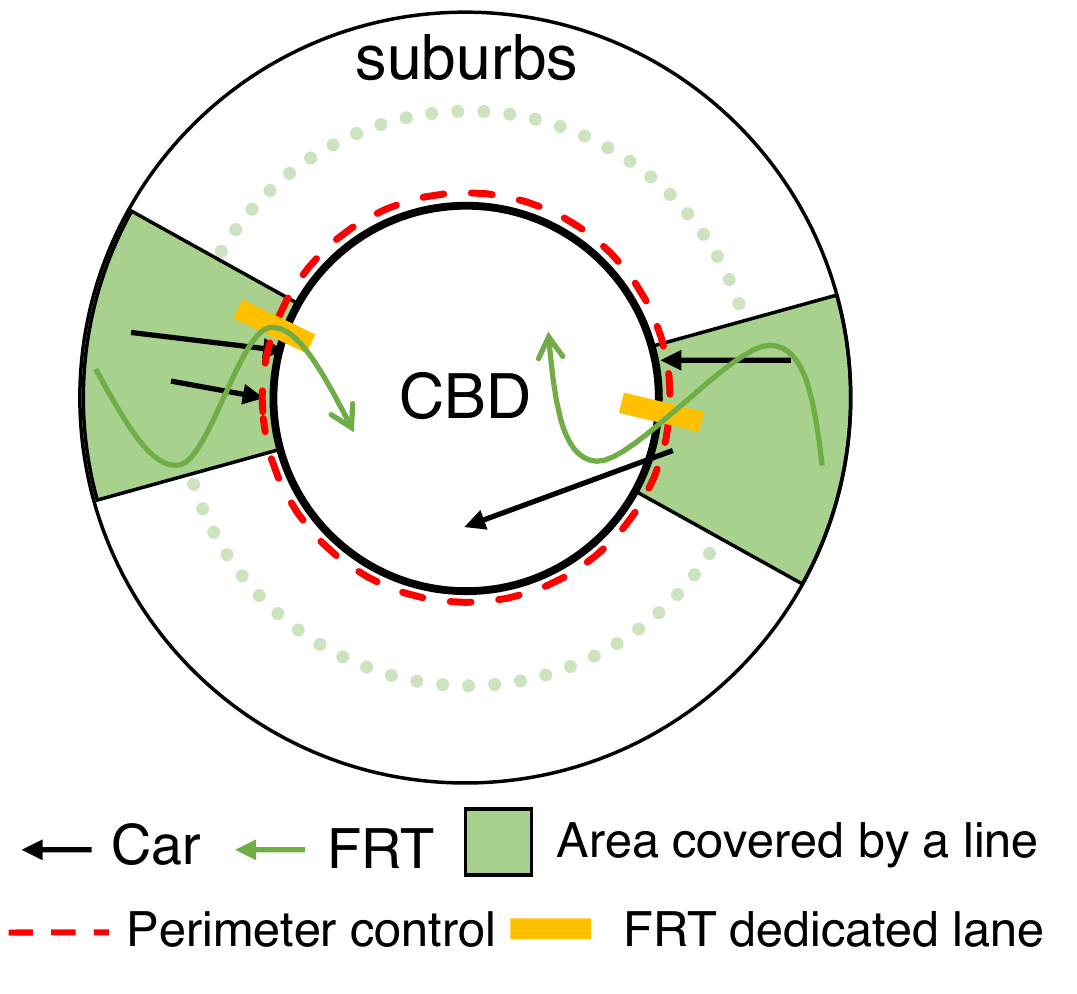}
\caption{Model structure with perimeter control}
\label{fig:model_pc}
\end{minipage}
\begin{minipage}{0.45\columnwidth}
\includegraphics[width=\columnwidth]{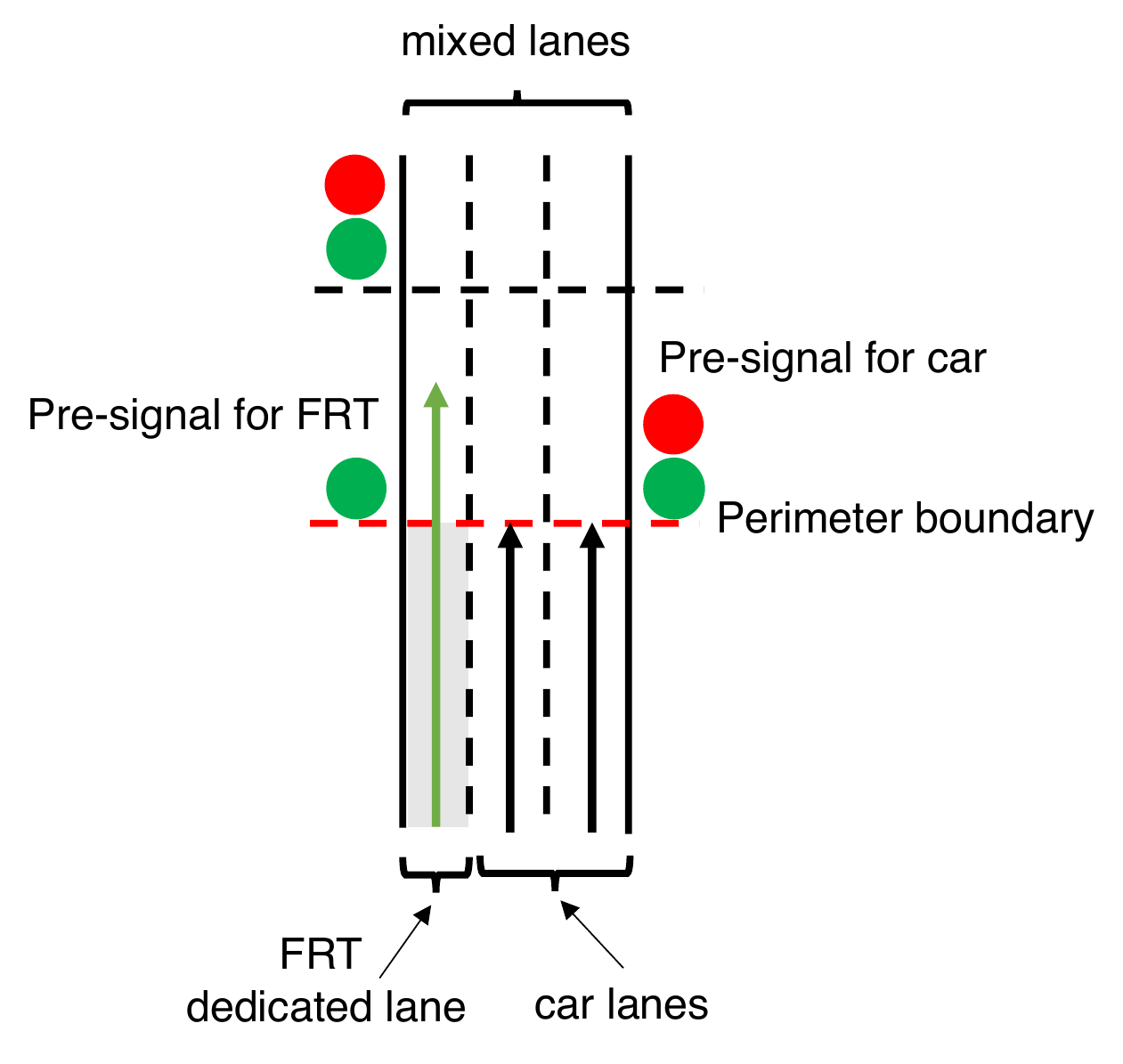}
\caption{One example of the strategies of perimeter control with transit priority}
\label{fig:idea}
\end{minipage}

\end{figure}

\subsection{Queuing dynamics at perimeter boundaries}
Perimeter control aims to maintain the maximum throughput in the CBD zone \citep{geroliminis2012optimal}.  From Eq.~(\ref{eq:car_speed_bi}) we can derive the critical car accumulation, $n_c^{cr}$, where the throughput is maximized ($d\left(n_c(t) v_c(t) \right)/dn=0$).
\begin{eqnarray}
n_c^{cr}=\frac{ n_j'}{2}
\label{eq:critical_density}
\end{eqnarray}
Eq.~(\ref{eq:critical_density}) indicates that the critical accumulation depends not only on car accumulation but also on  FRT accumulation, since $n_j' = n_j \left( 1- \eta n_F^c / n_j \right)$. To prevent  accumulation from exceeding the critical value, the inflow to the CBD zone is restricted at the perimeter boundary by  gating or traffic signals once the accumulation reaches the critical density. The inflow rate to the CBD zone is set to the exit rate from the  zone to maintain the critical accumulation, which means that the CBD traffic state is steady during perimeter control ($\dot{n_c}=0$ and $n_c(t)= n_j'/2$  if $ t_s^p \leq  t < t_e^p $ where $t_s^p$ and $t_e^p$ are the perimeter control start and end times, respectively) . Thus, the inflow rate during  perimeter control is determined from Eqs.~(\ref{eq:car_speed_bi})-(\ref{eq:trip_completion_rate}) and (\ref{eq:critical_density}) by
\begin{eqnarray} \label{eq:bottleneck_capacity}
I_c^p = \frac{ n_j' v_f'}{4L_c}
\end{eqnarray}
Therefore, the control scheme at time $t$ can be written by
\begin{eqnarray} \label{eq:contro_scheme}
I_c(t) =
\begin{cases}
I_c^p  & \text{if   } n_c(t) = \frac{ n_j'}{2} \\
 A_b(t) & \text{if   } n_c(t) < \frac{ n_j'}{2} 
\end{cases}
\end{eqnarray}
where $A_b(t)$ is the car arrival rate at the perimeter boundary at time $t$. If car accumulation is below the critical level, there is no restriction; all vehicles at the boundary can enter the CBD zone. Once accumulation reaches the critical accumulation, the inflow rate is restricted to $I_c^p$. This inflow rate  can be regarded as the bottleneck capacity of the perimeter boundary. Thus, a queue will develop outward from the perimeter boundary if the demand exceeds the inflow rate $I_c^p$. We model the queuing dynamic as a point queue and assume first-arrived-first-in property. Therefore, the waiting time of a commuter who arrives at their destination at time $t$, $T_b(t)$, is 
\begin{eqnarray}
T_b(t) = \frac{q(t)}{I_c^p} \label{eq:waiting_time_perimeter}
\end{eqnarray}
where $q(t)$ is the number of cars queued at the perimeter boundary when a commuter who arrives at their destination at time $t$ reaches the boundary. 

\subsection{Travel cost during perimeter control}
Given the queue dynamics during perimeter control, the car travel cost incurred by a commuter who arrives at their destination at time $t$ is
\begin{eqnarray}
C_c(t) = 
\begin{cases}
\alpha  T_c(t) + \beta \left( t^* - t \right) + F_c   & \text{if }t < t_s^p \\
\alpha \left( \frac{L_c}{ v_f'/2} + T_b\left( t  \right) \right) + \beta \left( t^* - t \right) + F_c  & \text{if }t_s^p\leq t < t^* \\
\alpha  \left(  \frac{L_c}{ v_f'/2} + T_b\left( t \right) \right)  + \gamma \left(  t - t^*  \right)  + F_c& \text{if } t^* \leq  t  < t_e^p \\
\alpha   T_c(t) + \gamma \left(  t - t^*  \right) + F_c & \text{if }  t_e^p \leq  t
\end{cases}
\label{eq:gc_perimeter_car}
\end{eqnarray}
Before and after perimeter control is implemented ( $t<t_s^p$ and $t_e^p \leq t$), the travel cost function is the same as Eq. (\ref{eq:gc_car}). During perimeter control, the CBD travel time is given by $\frac{L_c}{v_f'/2}$ since the trip completion rate is maintained at the maximum. In addition to the travel time cost , the cost of waiting at the perimeter boundary is incurred, as  given by Eq. (\ref{eq:waiting_time_perimeter}). 

As FRT vehicles can pass the perimeter boundary freely by the dedicated lanes, there is no waiting time at the perimeter boundary. Therefore, the FRT travel cost incurred by a commuter who arrives at their destination at time $t$ is
\begin{eqnarray}
C_F(t) = 
\begin{cases}
\alpha  T_F(t)  + \beta \left( t^* - t \right) + \lambda O_F(t) +F_F   & \text{if }t < t_s^p \\
\alpha \frac{L_F}{ mv_f'/2}   + \beta \left( t^* - t \right) + \lambda O_F(t)+ F_F  & \text{if }t_s^p\leq t < t^* \\
\alpha  \frac{L_d}{ mv_f'/2}   + \gamma \left(  t - t^*  \right)  + \lambda O_F(t)+ F_F& \text{if } t^* \leq  t  < t_e^p \\
\alpha   T_F(t)  + \gamma \left(  t - t^*  \right) + \lambda O_F(t)+ F_F & \text{if }  t_e^p \leq  t 
\end{cases}
\label{eq:gc_perimeter_bus}
\end{eqnarray}

\section{Equilibrium under perimeter control}\label{sec:equilibrium_perimeter}
No commuter can reduce their travel cost by changing their departure time and transportation mode at equilibrium during perimeter control.  Thus, the equilibrium conditions are  
\begin{subequations}
\begin{eqnarray}
&&\begin{cases}
C_c(t) = c^*_p & \text{if   } n(t) > 0 \\
C_c(t) \geq c^*_p & \text{if   } n(t) = 0 \\
\end{cases}  \; \; \;  \forall t \in \mathbb{R} \label{eq:condition_car_perimeter} \\ 
&&\begin{cases}
C_F(t) = c^*_p & \text{if   } O_F(t) > 0 \\
C_F(t) \geq c^*_p & \text{if   } O_F(t) = 0 \\
\end{cases}  \; \; \;  \forall t \in \mathbb{R} \label{eq:condition_FRT_perimeter}  \\ 
&&\begin{cases}
n_c(t) = \frac{ n_j' }{2} & \text{if   } q(t) > 0 \\
n_c(t) \leq \frac{ n_j' }{2} & \text{if   } q(t) = 0 
\end{cases}  \; \; \;  \forall t \in \mathbb{R} \label{eq:condition_bottleneck_perimeter}  \\ &&\int_{t \in \mathbb{R} } \frac{n_c(t)v_c(t)}{L_c} dt + \int_{t \in \mathbb{R} } \frac{O_F(t)n_F^c v_F(t)}{L_F} dt = N 
 \label{eq:condition_N_perimeter} 
\end{eqnarray}
\end{subequations}
where $c^*_p$ is the equilibrium cost during perimeter control. Conditions~(\ref{eq:condition_car_perimeter}), (\ref{eq:condition_FRT_perimeter}), and (\ref{eq:condition_bottleneck_perimeter}) are the same as Conditions (\ref{eq:UE_con_car}), (\ref{eq:UE_con_bus}) and (\ref{eq:UE_cons}), respectively. Condition (\ref{eq:condition_bottleneck_perimeter}) reflects the restriction of inflow to the CBD zone during perimeter control; car accumulation is at the critical accumulation if there is a queue at the perimeter boundary. Otherwise, car accumulation is lower than the critical level. Note that the difference in the conditions between user equilibrium and equilibrium under perimeter control  is only Condition (\ref{eq:condition_bottleneck_perimeter}). Moreover, dynamic user equilibrium under perimeter control is also based on \cite{wardrop1952road}'s first principle. 

\subsection{Equilibrium properties}
This study aims to analyse the cases where hypercongestion will occur at user equilibrium. Perimeter control can be implemented under the following condition. 
\begin{ass} \label{ass:ass2}
$\theta > 2 $
\end{ass}

From Condition (\ref{eq:condition_car_perimeter}) and Eq.~(\ref{eq:gc_perimeter_car}), the queue at the perimeter boundary during perimeter control has the same property as the standard bottleneck model. 
\begin{prop} \label{prop:queue}
Suppose Assumption \ref{ass:ass2}.  During perimeter control, a queue  develops at the perimeter boundary, and its length increases toward the desired arrival time. 
\end{prop}
\begin{proof}
Condition (\ref{eq:condition_car_perimeter}) and Eq.~(\ref{eq:gc_perimeter_car}) yield
\begin{eqnarray}
\begin{cases}
\alpha \frac{d T_b(t)}{dt}- \beta = 0  & \text{if }t_s^p\leq t < t^*\\
 \\
\alpha \frac{d T_b(t)}{dt} + \gamma = 0  & \text{if }t^* \leq t < t_e^p
\end{cases}
\end{eqnarray}
Combining this with Eq.~(\ref{eq:waiting_time_perimeter}), where $I_c(t)$ is constant (from Eq.~(\ref{eq:bottleneck_capacity})), we obtain 
\begin{eqnarray}
\begin{cases}
\frac{dq(t)}{dt}= \frac{\beta}{\alpha}& \text{if }t_s^p\leq t < t^* \\
\frac{dq(t)}{dt} = - \frac{\gamma}{\alpha}& \text{if }t^* \leq t < t_e^p
\end{cases}
\end{eqnarray}
which completes the proof.
\end{proof} 
This proposition indicates that some commuters prefer to wait at the perimeter boundary rather than change their departure times or transportation modes.  Since the arrival rate in the CBD is maintained at its maximum during perimeter control and the schedule delay cost decreases toward the desired arrival time, commuters can travel at the equilibrium cost,  even if they have to wait at the perimeter boundary. This mechanism is completely the same as that of the standard bottleneck model \citep{arnott1993structural}.

We define $\theta_p\equiv\frac{(c_p^*-F_c)v_f'}{\alpha L_c}$. Although the effectiveness of the combination of perimeter control and transit priority has been reported (e.g., \cite{haitao2019providing}), we present the following proposition. 

\begin{prop} \label{prop:FRT_not_use_perimeter}
Suppose Assumption \ref{ass:ass2}. The equilibrium with perimeter control has the following properties. 
\begin{enumerate}
    \item If $\Delta F \geq 2 \alpha \Delta T_f$, then FRT is used from the beginning to the end during the FRT rush hour and the number of FRT passengers increases toward the desired arrival time during perimeter control. 
    \item If $\Delta F < 2 \alpha \Delta T_f$ and $\theta_p > \frac{2 \alpha T_f^F - \Delta F}{\alpha T_f^c}$, then FRT is used; however, there will be a time window wherein FRT is not used during  perimeter control. Additionally, the number of FRT passengers increases toward the desired arrival time during perimeter control.  
    \item If $\theta_p \leq \frac{2 \alpha T_f^F - \Delta F}{\alpha T_f^c}$, then FRT is not used during perimeter control. 
\end{enumerate}
\end{prop}

\begin{proof}
Condition (\ref{eq:condition_car_perimeter}) yields Eq. (\ref{eq:car_evo_ue}) for car accumulation before and after perimeter control. Thus, the  start and end times of perimeter control ($t_s^p$ and $t_e^p$, respectively; $n_c(t_s^p)=n(t_e^p)= n_j'/2 $) are
\begin{eqnarray}
t_s^p - t_s^{c} = \frac{\alpha T_f^c}{\beta} \label{eq:time_diff_p_nc_b} \\
t_e^{c} - t_e^p = \frac{\alpha T_f^c}{\gamma} \label{eq:time_diff_p_nc_a}
\end{eqnarray}
Eqs.~(\ref{eq:time_FRT_not_used_earliness}) and (\ref{eq:time_FRT_not_used_lateness}) give the time window when FRT is not used during the FRT rush hour. When the start of perimeter control is earlier than the start of the above time window (i.e., $t_{ee}^F \geq t_s^p$ ), combining Eq. (\ref{eq:time_diff_p_nc_b}) with Eq. (\ref{eq:time_FRT_not_used_earliness}) gives 
\begin{eqnarray}
\Delta F \geq 2 \alpha \Delta T_f
\end{eqnarray}
Note that $2 \alpha \Delta T_f$ is the difference of the car and FRT travel time costs at the critical car accumulation. 
According to Eq.~(\ref{eq:gc_perimeter_bus}) and Condition~(\ref{eq:condition_FRT_perimeter}), $\dot{O_F}(t) = \beta / \lambda$ if $t_s^p \leq t \leq t^*$ and $\dot{O_F}(t) = -  \gamma / \lambda$ if $t^* < t \leq t_e^p$. Thus, the number of FRT passengers increases toward the desired arrival time during perimeter control. Therefore, if $\Delta F \geq 2 \alpha \Delta T_f$, FRT is used throughout the FRT rush hour,  and the number of FRT passengers increases toward the desired arrival time during perimeter control.

Let $[t_{sp}^{F}, t_{ep}^{F}]$ be the time window between the start and end times when FRT is used during perimeter control. If $t_{sp}^{F} < t^*$ (i.e., $t_{ep}^{F} > t^*$), there is a time window when FRT is used during perimeter control. Since $O_F(t_{sp}^{F}) = 0 $, Eq. (\ref{eq:gc_perimeter_bus}) yields
\begin{eqnarray}\label{eq:condition_FRT_use_periemeter}
\alpha \frac{L_F}{ mv_f'/2}   + \beta \left( t^* - t_{sp}^{F} \right) +  F_F = c_p^*
\end{eqnarray}
From Eq. (\ref{eq:condition_FRT_use_periemeter}), $t_{sp}^{F} < t^*$ can be rewritten as
\begin{eqnarray}
\theta_p > \frac{2\alpha T_f^F - \Delta F}{\alpha T_f^c}
\end{eqnarray}
Therefore, even when  $\Delta F < 2 \alpha \Delta T_f$, there is a time window when FRT is used during perimeter control if $\theta_p > \frac{2\alpha T_f^F - \Delta F}{\alpha T_f^c}$. Otherwise, FRT is not used during perimeter control. 

\end{proof} 

Proposition \ref{prop:FRT_not_use_perimeter}.1 indicates that FRT is used from the beginning to the end during  the FRT rush hour if  the fixed cost difference ($\Delta F$) is larger than the travel time cost difference during perimeter control at the critical car accumulation ($2 \alpha \Delta T_f$). Moreover,  the average number of passengers per FRT vehicle begins to increase toward the desired arrival time once perimeter control starts ($t>t_s^p$), as depicted in Fig. \ref{fig:graphical_FRT_perimeter}.  Since the FRT travel time cost is constant during perimeter control (Eq.~(\ref{eq:gc_perimeter_bus})), $\dot{O_F}(t)=\beta / \lambda$ for earliness, which is the same as the time evolution of $O_F(t)$ before the car rush hour starts. A trade-off exists between the schedule delay and discomfort costs.  In the same way,   $\dot{O_F}(t)=-\gamma / \lambda$ for lateness.  Thus, there is  no time interval when FRT is not used during the FRT rush hour.

From  Proposition \ref{prop:FRT_not_use_perimeter}.2, FRT is used if $\Delta F <2 \alpha \Delta T_f $  and the FRT free-flow travel time cost ($T_f^F$) and fixed cost ($F_F$) are relatively low compared to the other costs (i.e., $\theta_p > (2 \alpha T_f^F - \Delta F)/(\alpha T_f^c)$). However, there will be a time window wherein FRT is not used during perimeter control and the FRT rush hour.  As depicted in Fig. \ref{fig:graphical_FRT_perimeter_partially}, the number of FRT commuters becomes zero  before perimeter control starts. Since the travel time cost is constant during perimeter control, the FRT travel cost decreases with a decrease in the schedule delay cost toward the desired arrival time, and FRT is used.  Once the use of FRT begins, the average number of passengers per FRT vehicle increases toward the desired arrival time, as in Proposition \ref{prop:FRT_not_use_perimeter}.1. Since providing transit priority is not effective during the time window wherein FRT is not used, it limits the capacity of the perimeter boundary during this period. Therefore, {\it additional incentives that reduce the FRT travel cost, such as subsidy so that $\Delta F \geq 2 \alpha \Delta T_f $ is satisfied,  may be required to promote the use of FRT from the beginning to the end of perimeter control. }

Although prioritization of transit at the perimeter boundaries is a better strategy than unimodal perimeter control in terms of passenger mobility \citep{haitao2019providing}, Proposition \ref{prop:FRT_not_use_perimeter}.3 states that this is false if $\theta_p \leq (2 \alpha T_f^F - \Delta F)/(\alpha T_f^c)$, and FRT is not used during perimeter control. In this case, the FRT travel cost is higher than the car travel cost during perimeter control despite that FRT vehicles do not queue at the perimeter boundary. Hence, the bimodal transportation system is not fully utilized. As the vacant FRT vehicles circulate in the CBD zones and are prioritized at the perimeter boundary, these vehicles and transit priority (e.g., dedicated lanes) waste urban spaces around the peak of the rush hour. Ultimately, the effect of the strategy is the same as that of unimodal perimeter control in terms of passenger flow. {\it To promote the use of FRT during perimeter control, additional incentives that satisfy  at least  $\theta_p > (2 \alpha T_f^F - \Delta F)/(\alpha T_f^c)$ are required.}

\begin{figure}[t]
\centering
\includegraphics[width=0.8\columnwidth]{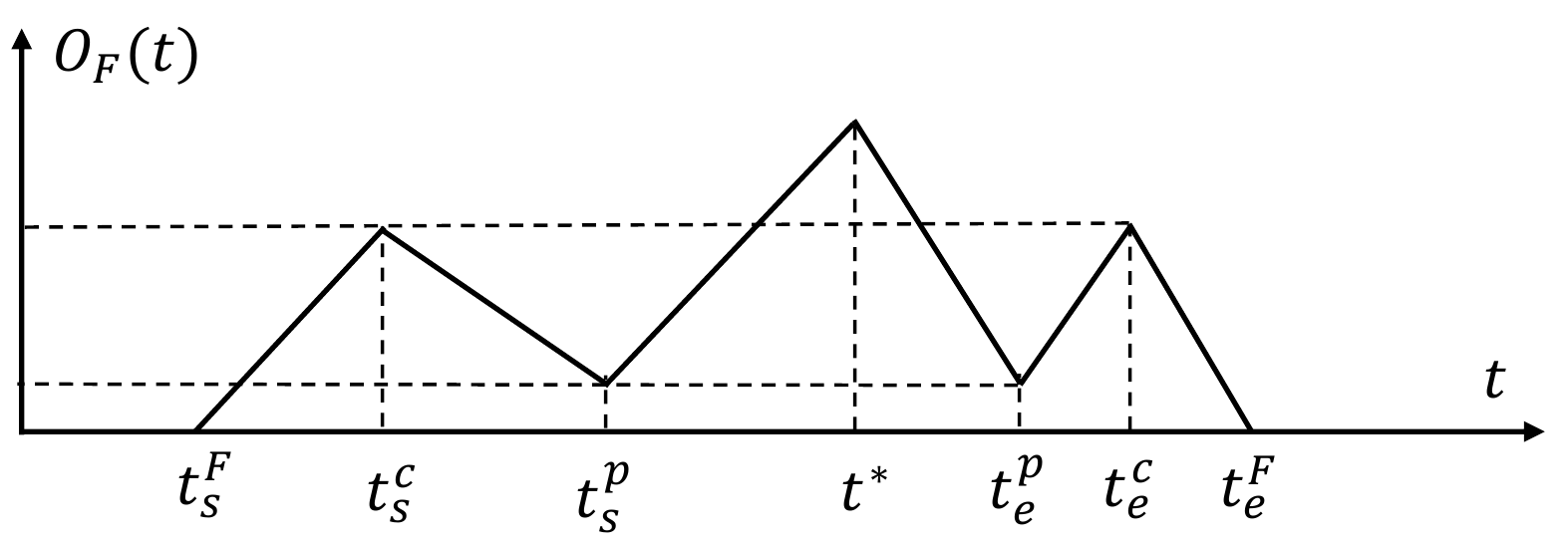}
\caption{Illustration of Proposition \ref{prop:FRT_not_use_perimeter}.1}
\label{fig:graphical_FRT_perimeter}
\end{figure} 
\begin{figure}[t]
\centering
\includegraphics[width=0.8\columnwidth]{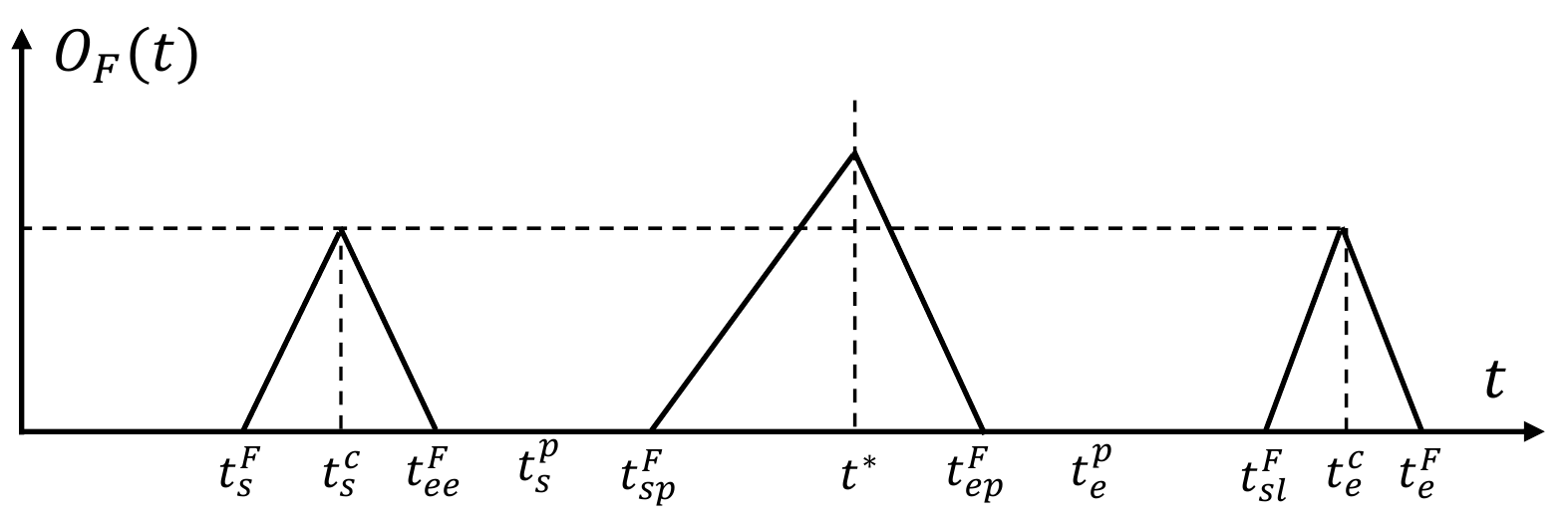}
\caption{Illustration of Proposition \ref{prop:FRT_not_use_perimeter}.2}
\label{fig:graphical_FRT_perimeter_partially}
\end{figure}

\subsection{Equilibrium cost}
We derive the equilibrium cost when FRT is used from the beginning to the end during perimeter control ( $\Delta F > 2 \alpha \Delta T_f$). The number of commuters who take FRT before and after the car rush hour  is the same as that in the case with user equilibrium. 
\begin{eqnarray} \label{eq:num_FRT_peri_before_cr}
N_F^{oc} &=&  \frac{n_F^c}{2\lambda}  \frac{1}{T_f^F} \left( \frac{1}{\beta} + \frac{1}{\gamma}\right) \left( \Delta F - \alpha \Delta T_f \right)^2
\end{eqnarray}
The number of  commuters who take  FRT during the car rush hour outside the perimeter control period ($t \in [t_s^c, t_s^p)$ and $t \in [t_e^p, t_e^c)$ ) is 
\begin{eqnarray} \label{eq:num_FRT_peri_outside}
N_F^{op} &=& \int_{t_s^c}^{t_s^p} \frac{O_F(s)n_F^cv_F(s)}{L_F} ds +  \int_{t_e^p}^{t_e^c} \frac{O_F(s)n_F^cv_F(s)}{L_F} ds \nonumber \\ 
&=&    n_F^c \frac{\alpha}{\lambda} \frac{T_f^c}{T_f^F}  \left( \frac{1}{\beta} + \frac{1}{\gamma}\right)   \left( \Delta F \ln  2  -  \alpha \Delta T_f   \right) 
\end{eqnarray}

When perimeter control starts, the average number of passengers per FRT vehicle is computed using Eqs.~(\ref{eq:gc_perimeter_car}) and (\ref{eq:gc_perimeter_bus}).
\begin{eqnarray}
O_F(t_s^p) = \frac{1}{\lambda} \left( \Delta F - 2 \alpha \Delta T_f \right)
\end{eqnarray}
Furthermore, we obtain $\dot{O_F}(t) = \beta / \lambda$,  $ t \in [t_s^p ,t^*)$ from  Eq.~(\ref{eq:gc_perimeter_bus}). Thus, the time evolution of the average number of passengers per  FRT vehicle is
\begin{eqnarray}
O_F(t) = \frac{1}{\lambda} \left( \Delta F - 2 \alpha \Delta T_f + \beta (t-t_s^p) \right) \; \; \; \text{if} \; \; \; t_s^p \leq t \leq t^* 
\end{eqnarray}
This can be obtained for lateness in the same manner. Thus,  
\begin{eqnarray}
O_F(t) = 
\begin{cases}
\frac{1}{\lambda} \left( \Delta F - 2 \alpha \Delta T_f + \beta (t-t_s^p) \right) \; \; \; \text{if} \; \; \; t_s^p \leq t \leq t^*  \\
\frac{1}{\lambda} \left( \Delta F - 2 \alpha \Delta T_f + \gamma (t_e^p-t) \right) \; \; \; \text{if} \; \; \; t^* < t \leq t_e^p  
\end{cases} \label{eq:number_evolution_passenger_FRT_under_perimeter}
\end{eqnarray}
The total number of FRT commuters who arrive at their destinations during perimeter control is 
\begin{eqnarray} \label{eq:n_F^p}
N_F^p = \int_{t_s^p}^{t^*} \frac{O_F(s)n_F^c m v_f' /2}{L_F} ds +
\int_{t^*}^{t_e^p} \frac{O_F(s)n_F^c m v_f' /2}{L_F} ds 
\end{eqnarray}
We combine Eqs.~(\ref{eq:number_evolution_passenger_FRT_under_perimeter}) and (\ref{eq:n_F^p}) and obtain
\begin{eqnarray}
N_F^p = \frac{n_F^c}{2T_f^F} \frac{1}{\lambda}\left( \frac{1}{\beta} + \frac{1}{\gamma}  \right) \left( c_p^* - 2 \alpha T_f^c - F_c \right) \left( \Delta F - 2 \alpha \Delta T_f + \frac{1}{2}\left( c_p^* - 2 \alpha T_f^c - F_c \right)  \right) 
\label{eq:number_commuters_FRT_under_perimeter}
\end{eqnarray}
Since $c_p^*=2 \alpha T_f^c+\beta \left(t^* - t_s^p \right) + F_c = 2 \alpha T_f^c+\gamma \left(t_e^p - t^*  \right) + F_c$, substituting this into Eq.(\ref{eq:number_commuters_FRT_under_perimeter}), we obtain: 
\begin{eqnarray}\label{eq:number_commuters_FRT_under_perimeter_theta}
N_F^p = \frac{n_F^c}{2} \frac{\alpha}{\lambda} \frac{T_f^c}{T_f^F} \left( \frac{1}{\beta} + \frac{1}{\gamma}  \right) \left( \theta^p - 2  \right) \left( \Delta F - 2 \alpha \Delta T_f + \frac{\alpha T_f^c}{2} \left( \theta^p - 2  \right)   \right) 
\end{eqnarray}
where $\theta^p\equiv\frac{(c_p^*-F_c)v_f'}{\alpha L_c}$.

As the car trip completion rate is maintained at the maximum ($G_c(t)=n_j' v_f'/4 L_c$, $t\in[t_s^p,t_e^p) $), the number of car commuters who arrive at their destination during perimeter control is  
\begin{eqnarray}\label{eq:num_car_p}
N_c^p = \frac{n_j'v_f'}{4L_c} (t_e^p - t_s^p)
\end{eqnarray}
Since $c_p^*=2 \alpha T_f^c+\beta \left(t^* - t_s^p \right) + F_c = 2 \alpha T_f^c+\gamma \left(t_e^p - t^*  \right) + F_c$, substituting it into Eq.(\ref{eq:num_car_p}), we have: 
\begin{eqnarray} \label{eq:num_car_p_theta}
N_c^p = \frac{\alpha n_j'}{4}  \left( \frac{1}{\beta} + \frac{1}{\gamma}  \right) \left( \theta^p - 2 \right)
\end{eqnarray}

The number of car commuters who arrive at their destinations before and after perimeter control is computed using Eqs.~(\ref{eq:car_speed_bi}), (\ref{eq:car_evo_ue}), (\ref{eq:time_diff_p_nc_b}), and (\ref{eq:time_diff_p_nc_a}) by 
\begin{eqnarray}\label{eq:num_car_nc}
N_c^{op} &=& \int_{t_s^{c}}^{t_s^p} \frac{n_c(s)v_c(s)}{L_c} ds + \int_{t_e^{p}}^{t_e^{c}} \frac{n_c(s)v_c(s)}{L_c} ds \nonumber \\
&=& n_j' \left(\frac{\alpha}{\beta} + \frac{\alpha}{\gamma} \right) \left( \ln2 - \frac{1}{2} \right)
\end{eqnarray}
We combine  Eqs.~(\ref{eq:num_FRT_peri_before_cr}), (\ref{eq:num_FRT_peri_outside}), (\ref{eq:number_commuters_FRT_under_perimeter_theta}), (\ref{eq:num_car_p_theta}), and (\ref{eq:num_car_nc}), and the total number of commuters is 
\begin{eqnarray} \label{eq:total_number_perimeter_case2}
N = && \left(  \frac{1}{\beta} + \frac{1}{\gamma}\right) \left[ \frac{\alpha n_j'}{4}   \left( \theta^p - 2 \right) + \alpha n_j' \left( \ln2 - \frac{1}{2} \right) + \nonumber \right. \\
&& \left. \frac{n_F^c}{\lambda} \frac{1}{T_f^F} \left(    \frac{\alpha T_f^c}{2} \left( \theta^p - 2  \right) \left( \Delta F - 2 \alpha \Delta T_f + \frac{\alpha T_f^c}{2} \left( \theta^p - 2  \right) \right) + \frac{1}{2} \left( \Delta F - \alpha \Delta T_f \right)^2 + \alpha T_f^c \left( \Delta F \ln 2 - \alpha \Delta T_f \right) \right)   \right]
\end{eqnarray}
Given that $\theta^p\equiv\frac{(c_p^*-F_c)v_f'}{\alpha L_c}$, this can also  be  solved numerically with respect to the equilibrium cost $c_p^*$.

We derive the equilibrium cost for the case where FRT is used, but there will a time window wherein FRT is not used during perimeter control ($\Delta F \leq 2 \alpha \Delta T_f$ and $\theta_p > \frac{2 \alpha T_f^F - \Delta F}{\alpha T_f^c}$). 
Since $O_F(t_{sp}^F)=O_F(t_{ep}^F)=0$, the start and end times when FRT is used during perimeter control are  given from Eq. (\ref{eq:gc_perimeter_bus}) by 
\begin{eqnarray}\label{eq:start_end_times_FRT_used_perimeter_control}
t^* - t_{sp}^F = \frac{1}{\beta} \left( c^*_p - 2 \alpha T_f^F -  F_F  \right) \\
t_{ep}^F - t^* = \frac{1}{\gamma} \left( c^*_p - 2 \alpha T_f^F -  F_F  \right)
\end{eqnarray}
Since $\dot{O_F}(t) = \beta / \lambda$,  $ t \in [t_{sp}^F ,t^*)$ and $\dot{O_F}(t) = -  \gamma / \lambda$,  $ t \in [ t^*, t_{ep}^F )$ from  Eq.~(\ref{eq:gc_perimeter_bus}),  the total number of FRT commuters who arrive at their destinations during perimeter control is 
\begin{eqnarray} \label{eq:n_F^p_case2}
N_F^p &=& \int_{t_{sp}^F}^{t^*} \frac{O_F(s)n_F^c m v_f' /2}{L_F} ds +
\int_{t^*}^{t_{ep}^F} \frac{O_F(s)n_F^c m v_f' /2}{L_F} ds  \nonumber \\
&=& \frac{n_F^c}{2\lambda} \frac{\left( \alpha T_f^c\right)^2}{2T_f^F}  \left( \frac{1}{\beta} + \frac{1}{\gamma}  \right) \left( \theta^p - \frac{2\alpha T_f^F - \Delta F}{\alpha T_f^c} \right)^2
\end{eqnarray}
The number of FRT commuters before and after perimeter control is given by Eq. (\ref{eq:FRT_number_ue}). As the total number of car commuters is given by the sum of Eqs. (\ref{eq:num_car_p}) and (\ref{eq:num_car_p_theta}), we combine Eqs. (\ref{eq:FRT_number_ue}),(\ref{eq:num_car_p}), (\ref{eq:num_car_p_theta}), and (\ref{eq:n_F^p_case2}), and the total number of commuters satisfies 
\begin{eqnarray} \label{eq:total_number_perimeter_case1}
N &=& \left( \frac{1}{\beta} + \frac{1}{\gamma}\right)  \left[\frac{\alpha n_j'}{4}   \left( \theta^p - 2 \right) + \alpha n_j' \left( \ln2 - \frac{1}{2} \right) + \right. \nonumber \\
&&  \left. \frac{n_F^c}{\lambda} \frac{1}{T_f^F} \left( \frac{\left( \alpha T_f^{c} \right)^2}{4} \left(\theta_p - \frac{2\alpha T_f^F - \Delta F}{\alpha T_f^c}  \right)^2 + \frac{1}{2} \left( \Delta F - \alpha \Delta T_f \right)^2 + \alpha T_f^c \left( \Delta F \ln  \frac{\Delta F }{\alpha \Delta T_f} - \left(\Delta F -  \alpha \Delta T_f  \right)\right) \right) \right]
\end{eqnarray}
Given that $\theta^p\equiv\frac{(c_p^*-F_c)v_f'}{\alpha L_c}$, this can also  be  solved numerically with respect to the equilibrium cost $c_p^*$.  

Even if Assumption 1 is not satisfied, FRT is used during perimeter control provided $\theta_p > \frac{2 \alpha T_f^F - \Delta F}{\alpha T_f^c}$. In this case, FRT is used only during perimeter control. The number of FRT commuters is given by Eq. (\ref{eq:n_F^p_case2}). Thus, the total number of commuters is 
\begin{eqnarray} \label{eq:total_number_perimeter_case3}
N &=& \left( \frac{1}{\beta} + \frac{1}{\gamma}\right)  \left[\frac{\alpha n_j'}{4}   \left( \theta^p - 2 \right) + \alpha n_j' \left( \ln2 - \frac{1}{2} \right) +\frac{n_F^c}{\lambda} \frac{1}{T_f^F} \frac{\left( \alpha T_f^{c} \right)^2}{4} \left(\theta_p - \frac{2\alpha T_f^F - \Delta F}{\alpha T_f^c}  \right)^2 \right]
\end{eqnarray}
where $\theta^p\equiv\frac{(c_p^*-F_c)v_f'}{\alpha L_c}$. This can also be solve numerically.

Next, we derive the equilibrium cost for the case where FRT is not used during perimeter control. If $\theta_p \leq \frac{2 \alpha T_f^F - \Delta F}{\alpha T_f^c}$, then perimeter control is never implemented if  Assumption 1 holds (See Appendix B for the proof ). Thus, when perimeter control is implemented, Assumption 1 is not satisfied, indicating that  the total number of commuters is equal to the number of car commuters since FRT is never used, and is given by  
\begin{eqnarray} \label{eq:total_number_perimeter_case4}
N &=& \left( \frac{1}{\beta} + \frac{1}{\gamma}\right)  \left(\frac{\alpha n_j'}{4}   \left( \theta^p - 2 \right) + \alpha n_j' \left( \ln2 - \frac{1}{2} \right)  \right)
\end{eqnarray}

As with the user equilibrium without perimeter control, the equilibrium cost during perimeter control  has the following property. 

\begin{prop} \label{prop:uniqueness_perimeter_control}
The equilibrium cost during perimeter control is uniquely determined, regardless of whether or not FRT is used during  perimeter control. 
\end{prop}
\begin{proof}
See Appendix C. 
\end{proof}

Although the equilibrium cost is uniquely determined, the impacts of certain parameters  on the equilibrium cost and  mode share are unclear. For example,  Eqs.~(\ref{eq:total_number_perimeter_case2}), (\ref{eq:total_number_perimeter_case1}), (\ref{eq:total_number_perimeter_case3}), and (\ref{eq:total_number_perimeter_case4}) show that the number of FRT vehicles may increase and decrease the equilibrium cost. Therefore,   
the properties of equilibrium with and without perimeter control, such as its cost and the role of certain parameters (e.g., the number of FRT vehicles), are investigated through numerical examples in the next section. 

\begin{figure}[t]
\begin{minipage}{0.5\columnwidth}
\centering
\includegraphics[width=\columnwidth]{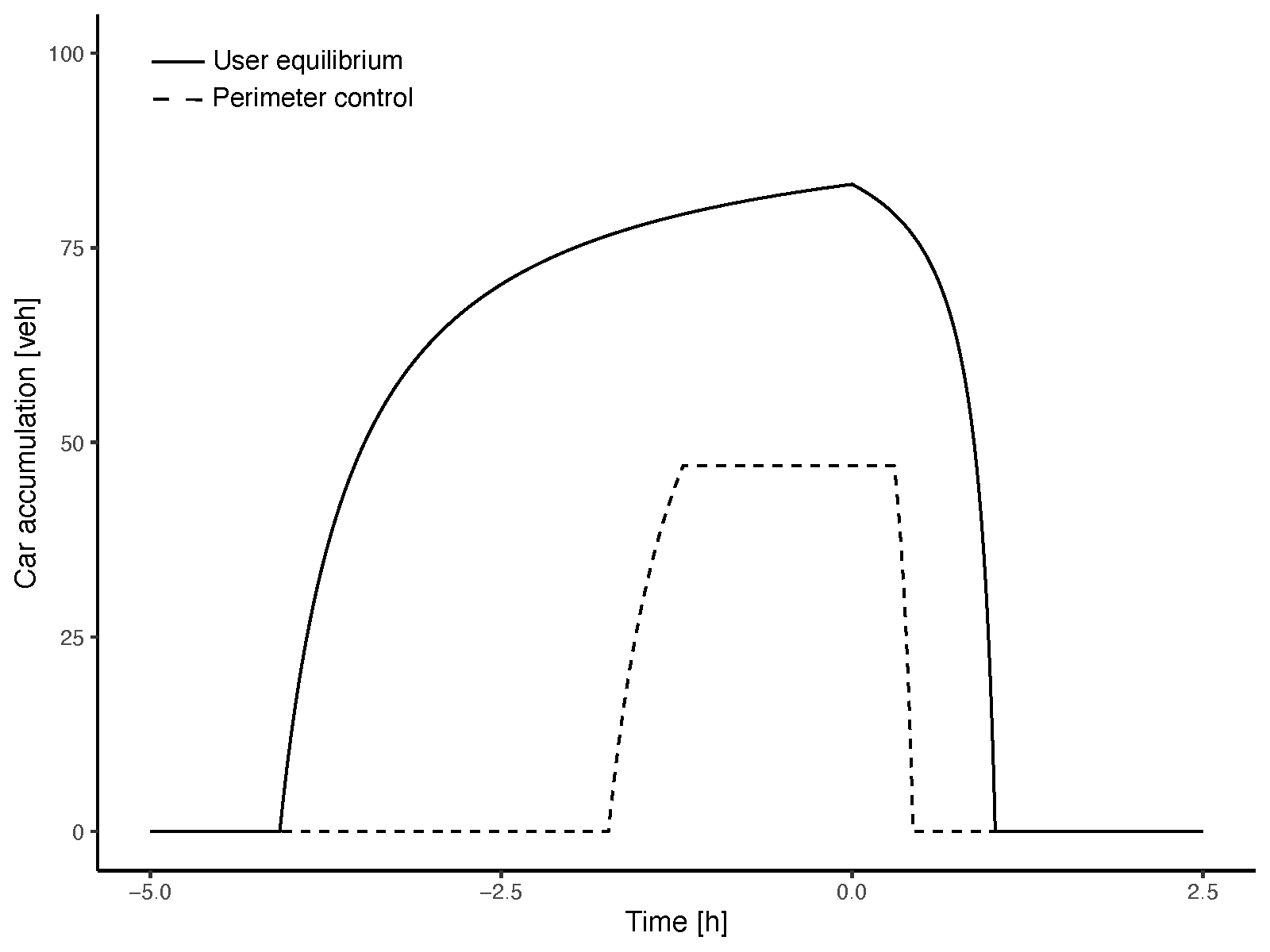}
\subcaption{Case I}
\end{minipage}
\begin{minipage}{0.5\columnwidth}
\centering
\includegraphics[width=\columnwidth]{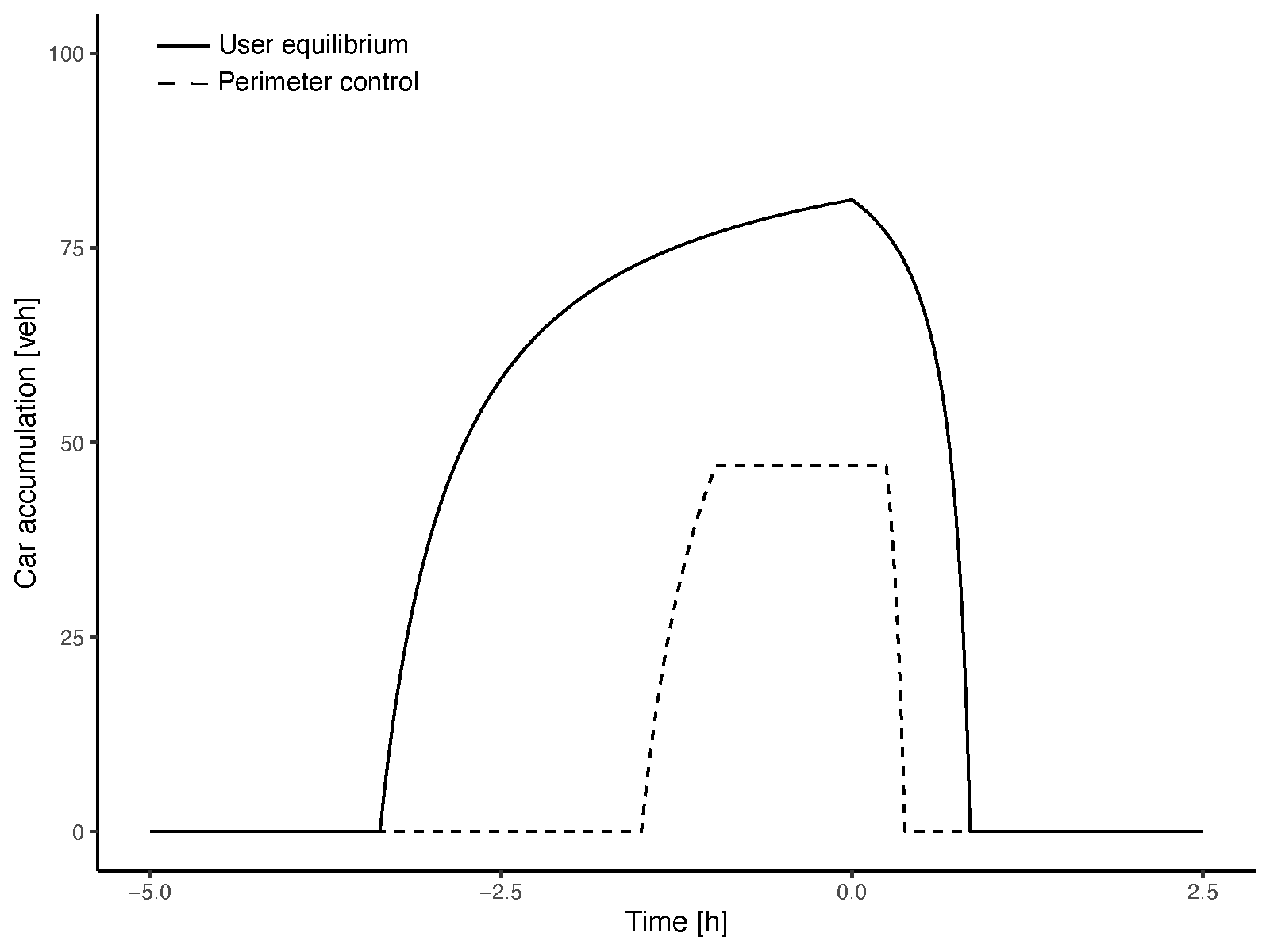}
\subcaption{Case II}
\label{fig:ue_w_FRT}
\end{minipage}
\caption{Car accumulation in the CBD zone}
\label{fig:accumulation}
\begin{minipage}{0.5\columnwidth}
\centering
\includegraphics[width=\columnwidth]{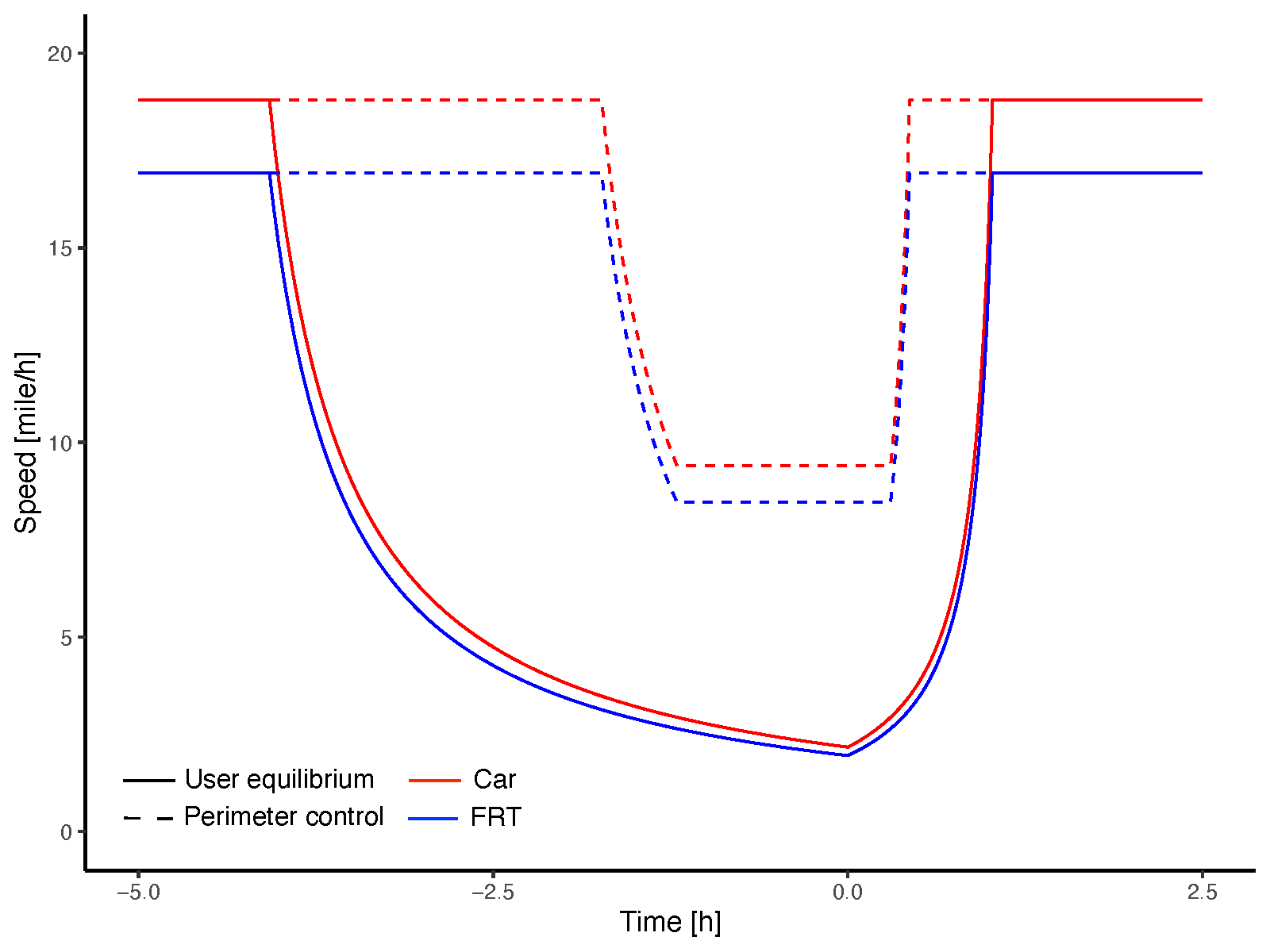}
\subcaption{Case I}
\end{minipage}
\begin{minipage}{0.5\columnwidth}
\centering
\includegraphics[width=\columnwidth]{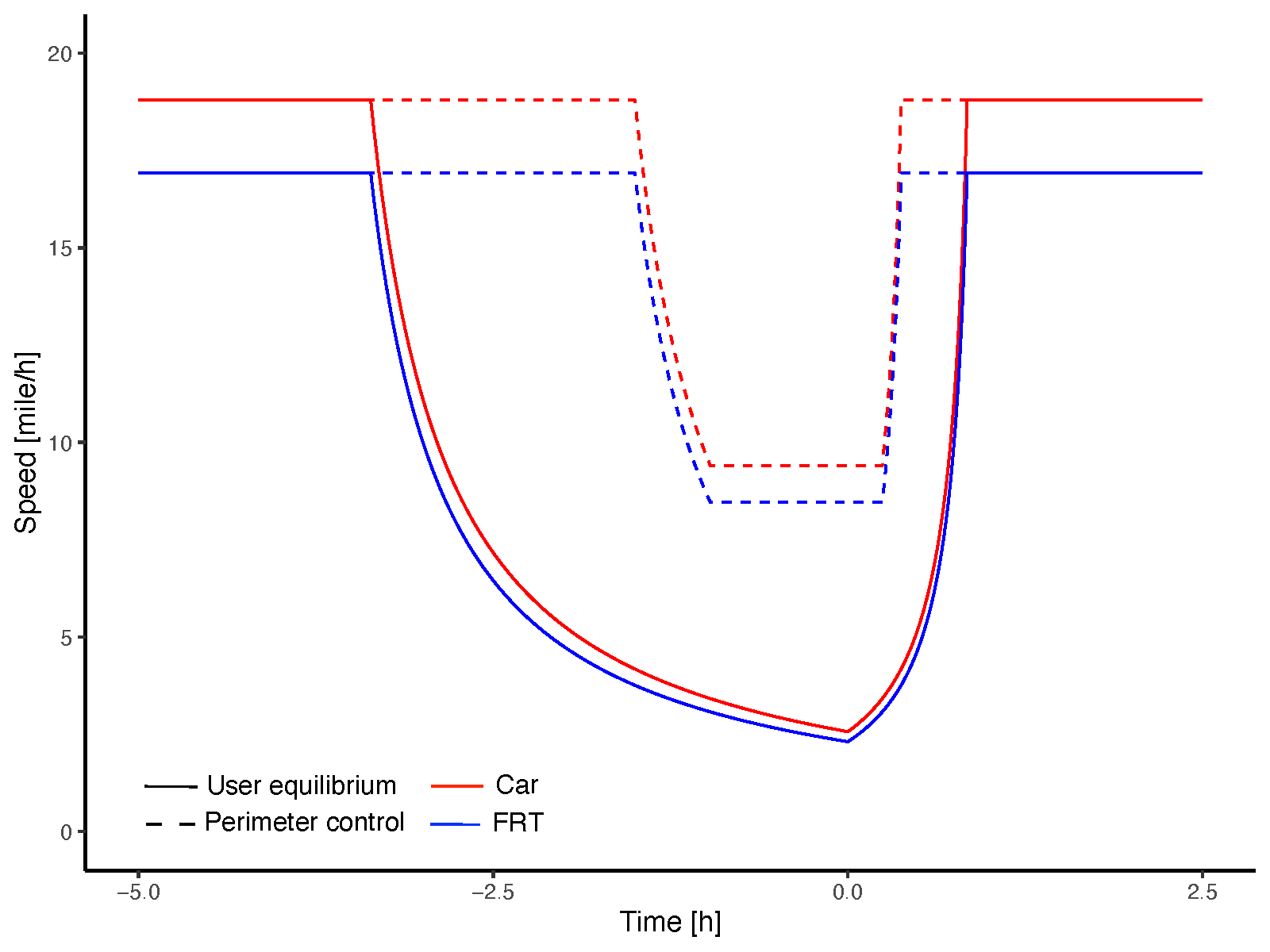}
\subcaption{Case II}
\label{fig:ue_w_FRT}
\end{minipage}
\caption{Speed in the CBD zone}
\label{fig:speed}
\end{figure}

\section{Numerical examples} \label{seq:numerical}
We numerically investigate the equilibrium patterns of the mixed bimodal bathtub model. The parameters are as follows: $v_f=20$ [mile/h], $\eta=1.2$, $n_F^c=5$ [veh], $n_j=100$ [veh], $m=0.9$, $\alpha, \beta, \gamma=20,10,40$ [\$/h], $F_c=5$ [\$], $L_c, L_F=5,6$ [mile], $N=300$ [pax], $\lambda=0.4$ [\$/pax] and $t^*=0$.  Some of the parameters are taken from \cite{arnott2013bathtub}. Note that as we set the desired arrival time to zero,  the negative time indicates the time before the desired arrival time (i.e., earliness), and the positive time indicates the time after the desired arrival time (i.e., lateness).  We compare the equilibrium patterns of two cases: with FRT partial use (Case I) and with FRT use (Case II) during perimeter control. Thus, in addition to above parameters, we set the FRT fixed cost to $F_F=2$ [\$] for Case I and to $F_d=1$ [\$] for case II. 

Car accumulation between Cases I and II is compared in Fig. \ref{fig:accumulation}. In both cases, a similar evolution of car accumulation is obtained; however,  because of the higher FRT fixed cost for Case I, more commuters use their cars, which leads to a longer rush hour.   At user equilibrium, hypercongestion occurs (above $n_j'/2=47$ [veh]), and the highest accumulation is reached at the desired arrival time. These results are consistent with those in the literature on unimodal bathtub models (e.g., \cite{small2003hypercongestion, geroliminis2009cordon, arnott2013bathtub}). 

During perimeter control, critical accumulation is maintained  around the desired arrival time. In both cases, the rush hour is shortened during perimeter control. This is because the traffic condition in the CBD zone never reaches hypercongestion, and the arrival rate is maintained at the maximum during perimeter control. Thus, more commuters can arrive at their destinations near the desired arrival time. These features are visible from the speed profiles in Fig. \ref{fig:speed} as well.

As depicted in Figs.~\ref{fig:arrival_rate} and \ref{fig:Ave_pass}, traffic becomes hypercongested once the maximum car arrival rate is reached at user equilibrium, hence  the car arrival rate decreases toward the desired arrival time. When  perimeter control is implemented, the car arrival rate is maintained at the maximum for both cases. 
The FRT passenger arrival rate and the average number of passengers per FRT vehicle differ between the two cases. In Case I, there is a time window wherein FRT is not used during the FRT rush hour regardless of perimeter control. During perimeter control, FRT is used again and the arrival rate increases toward the desired arrival time, as verified in Proposition \ref{prop:FRT_not_use_perimeter}.2. Even though the FRT vehicular arrival rate is constant during perimeter control, the passenger arrival rate increases because the average number of passengers per  FRT vehicle increases, which is depicted in Fig.  \ref{fig:Ave_pass}  (a) and verified  in Proposition \ref{prop:FRT_not_use_perimeter}.2. At user equilibrium without perimeter control in Case II, the number of FRT commuters is higher than that in Case I due to the lower FRT fixed cost, but its evolution pattern is the same as that in Case I as depicted in Fig. \ref{fig:arrival_rate} (b). It is also shown that at equilibrium with perimeter control, FRT is used from the beginning to the end during the FRT rush hour. Once perimeter control starts, the arrival rate begins to increase because of the increase in the average number of passengers per FRT vehicle, which is depicted in Fig.  \ref{fig:Ave_pass}  (b) and verified  in Proposition \ref{prop:FRT_not_use_perimeter}.1.

\begin{figure}[t]
\begin{minipage}{0.5\columnwidth}
\centering
\includegraphics[width=\columnwidth]{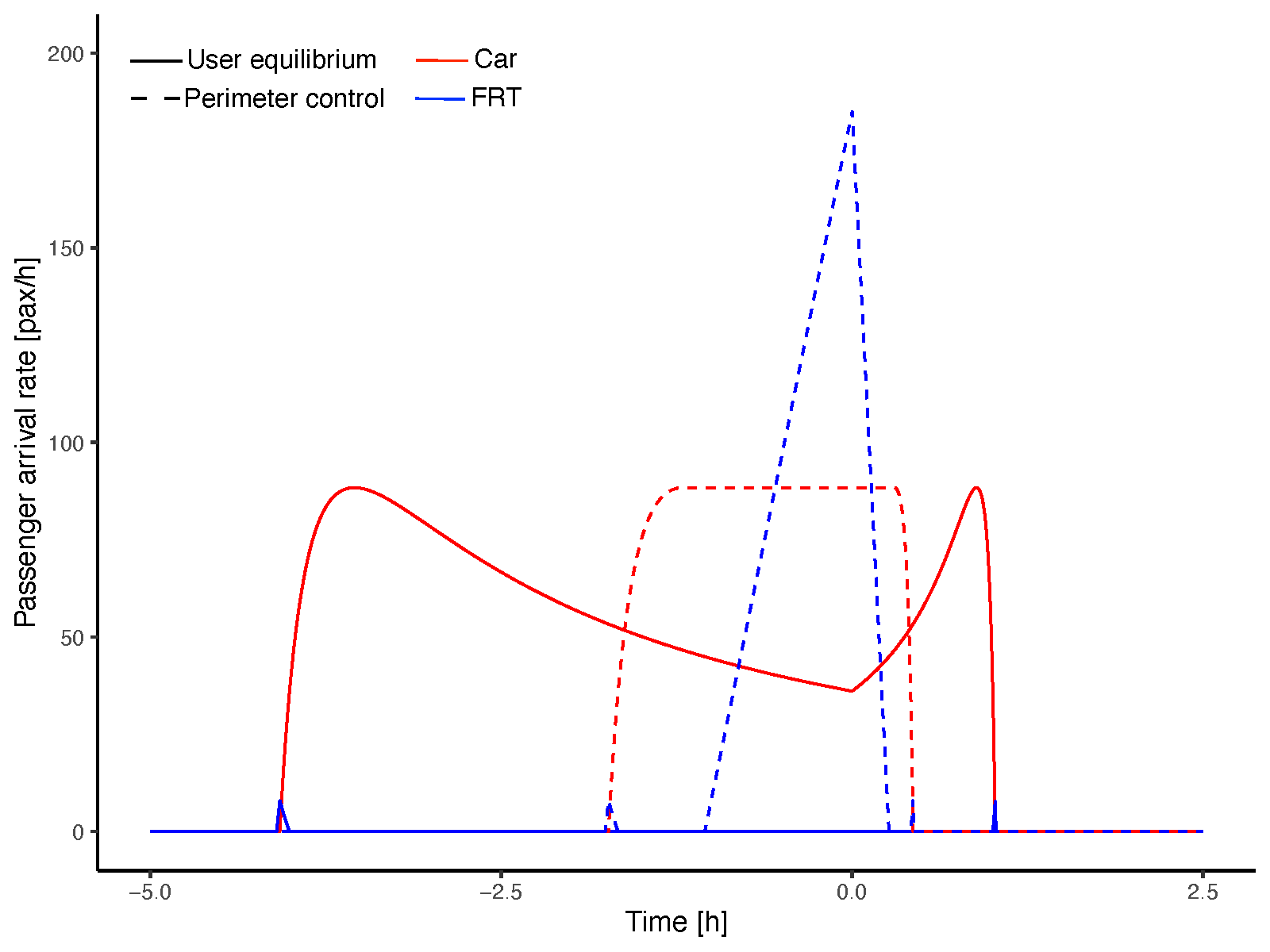}
\subcaption{Case I}
\end{minipage}
\begin{minipage}{0.5\columnwidth}
\centering
\includegraphics[width=\columnwidth]{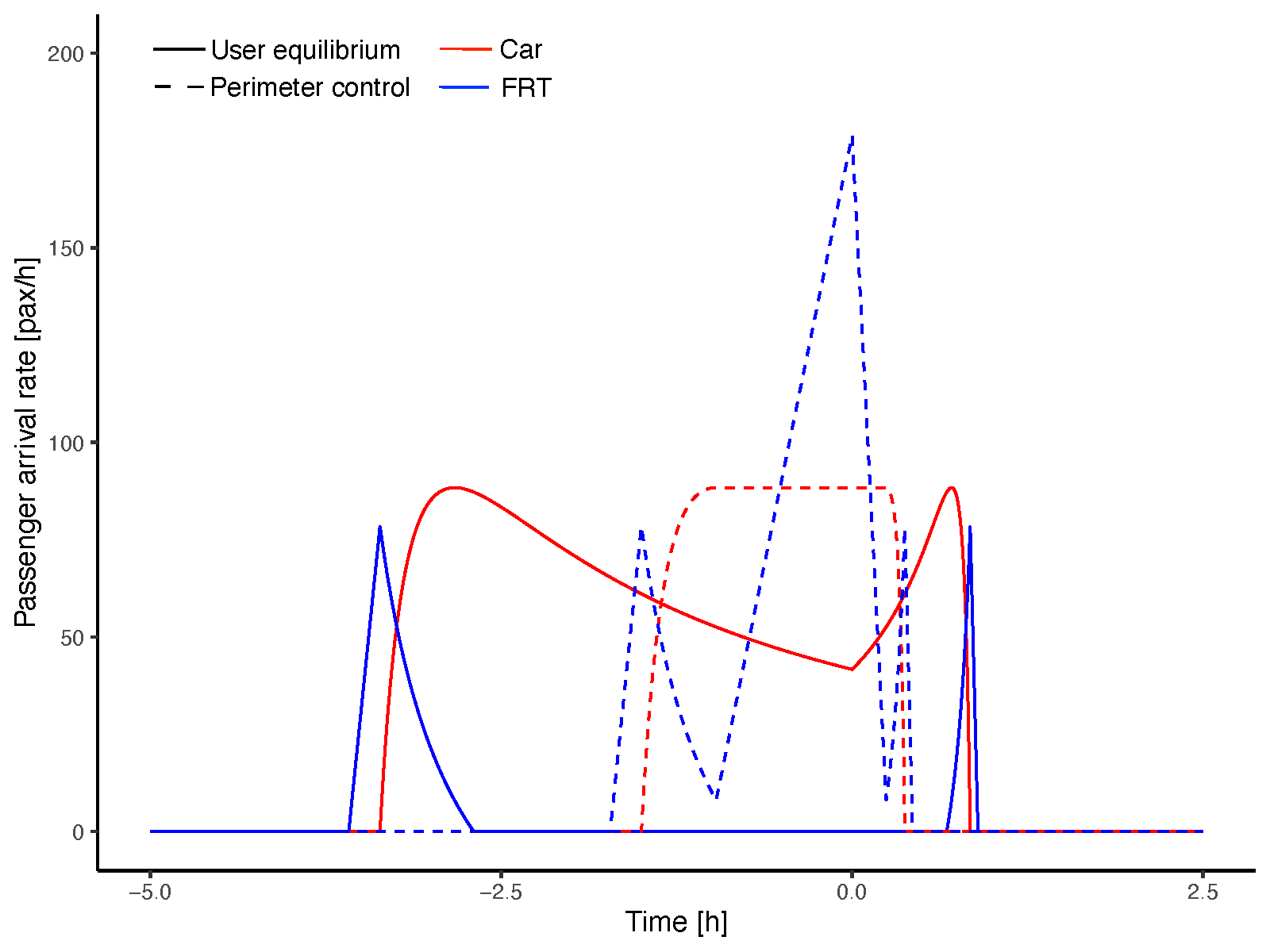}
\subcaption{Case II}
\label{fig:ue_w_FRT}
\end{minipage}
\caption{Passenger arrival rate at their destination}
\label{fig:arrival_rate}
\begin{minipage}{0.5\columnwidth}
\centering
\includegraphics[width=\columnwidth]{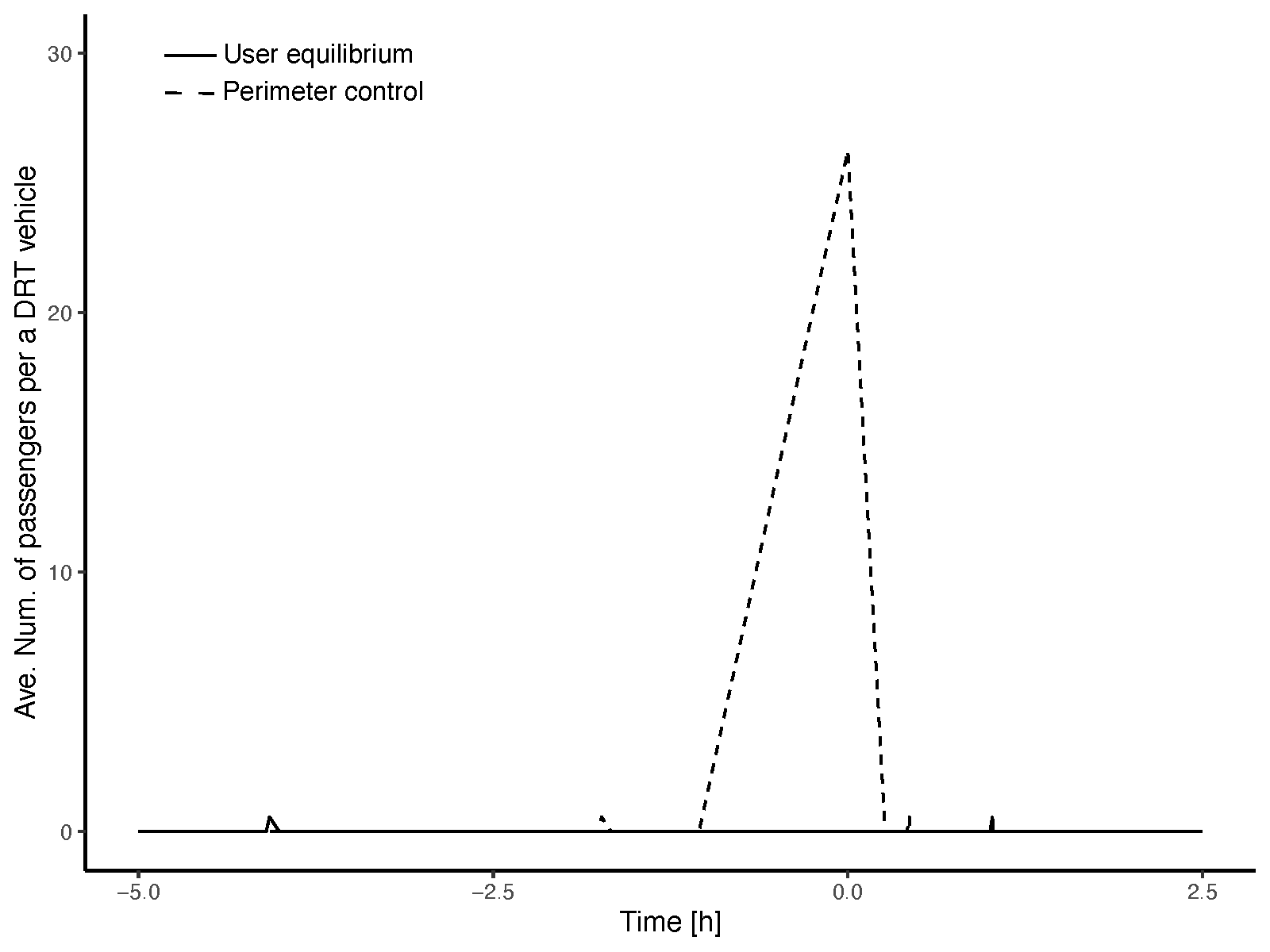}
\subcaption{Case I}
\end{minipage}
\begin{minipage}{0.5\columnwidth}
\centering
\includegraphics[width=\columnwidth]{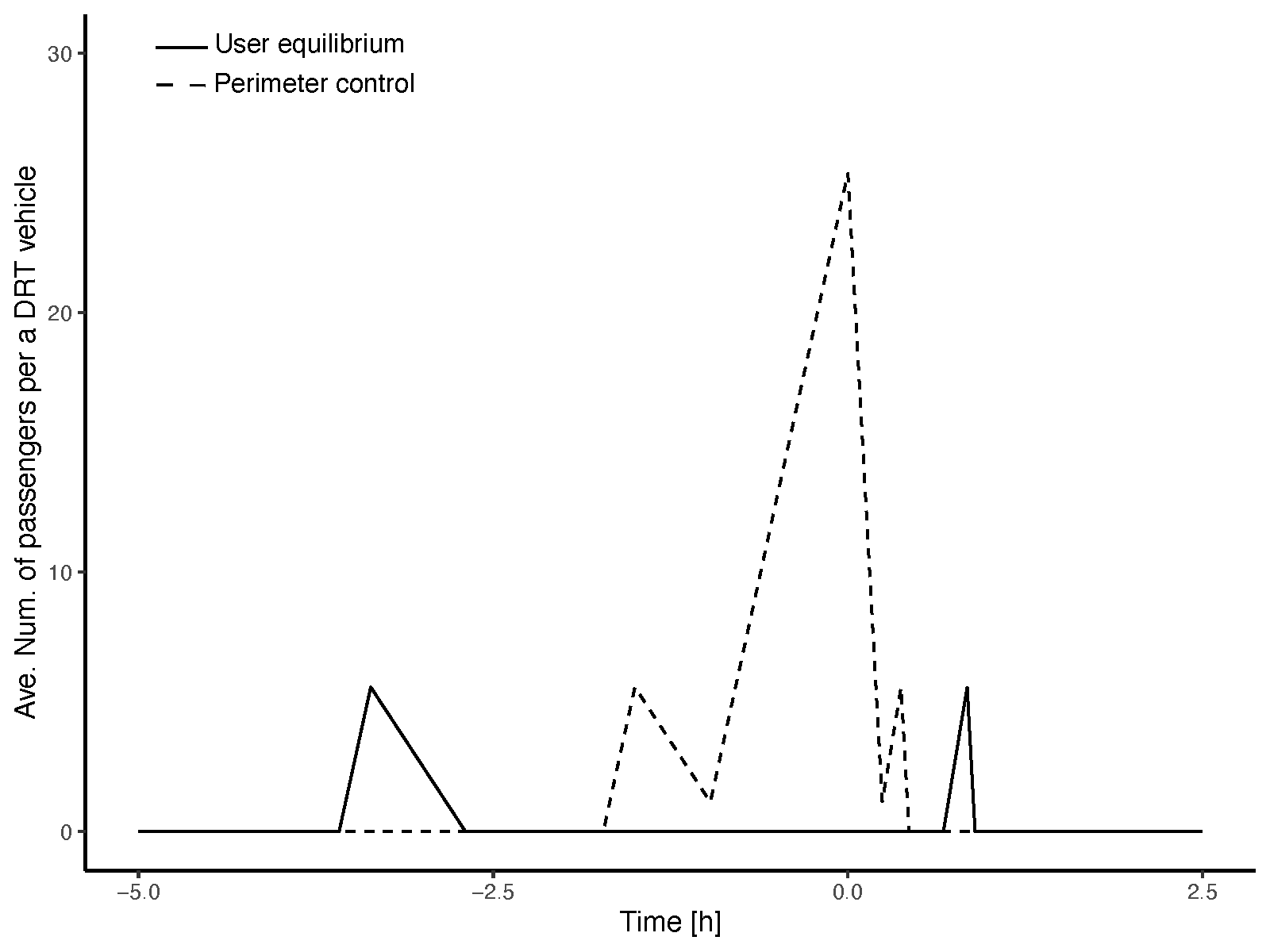}
\subcaption{Case II}
\label{fig:ue_w_FRT}
\end{minipage}
\caption{Average number of passengers per a FRT vehicle}
\label{fig:Ave_pass}
\end{figure}

\begin{figure}[t]
    \centering
    \begin{minipage}{0.45\columnwidth}
    \centering
    \includegraphics[width=\columnwidth]{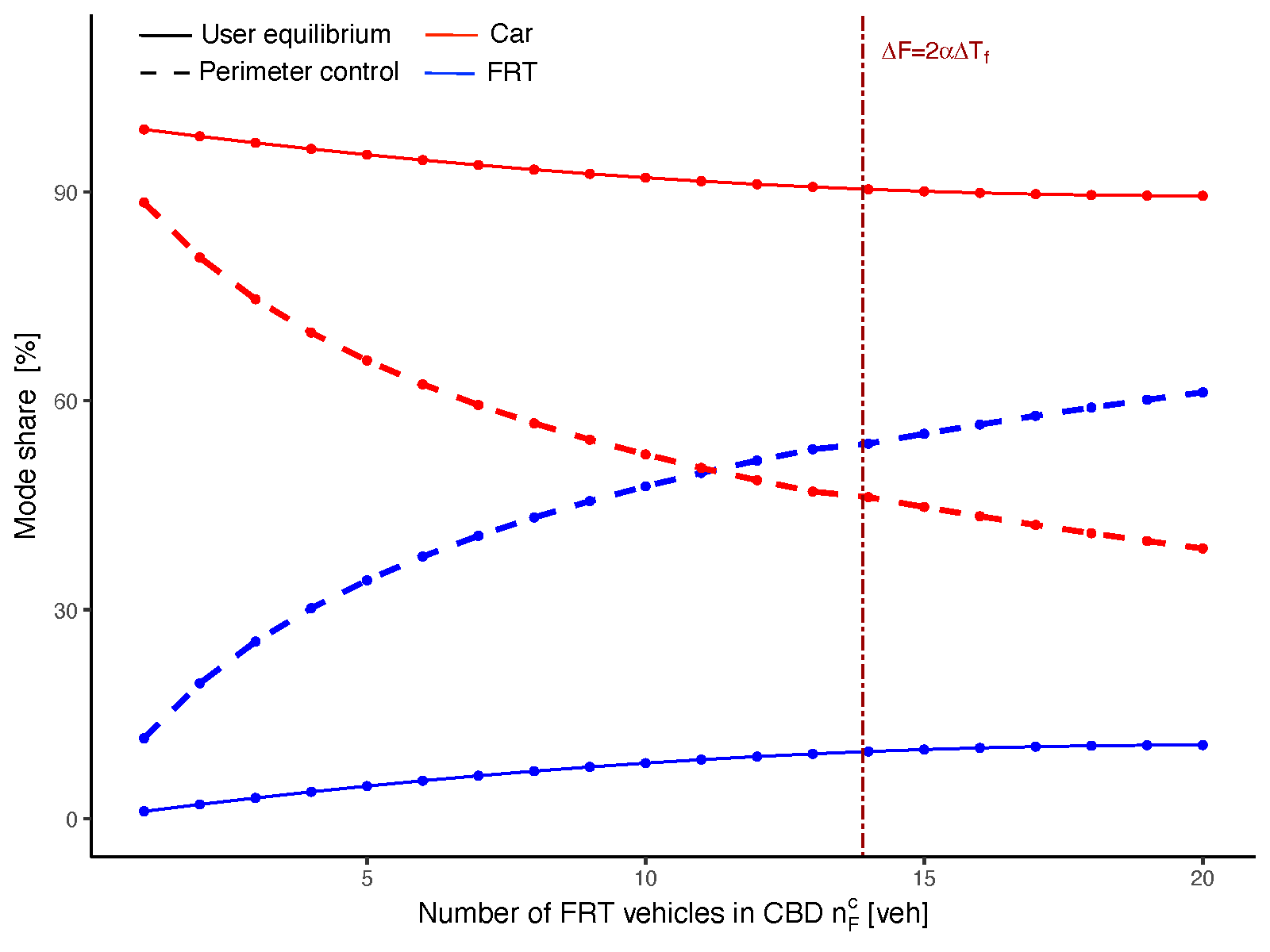}
    \subcaption{Mode share}
    \end{minipage}
        \centering
    \begin{minipage}{0.45\columnwidth}
    \centering
    \includegraphics[width=\columnwidth]{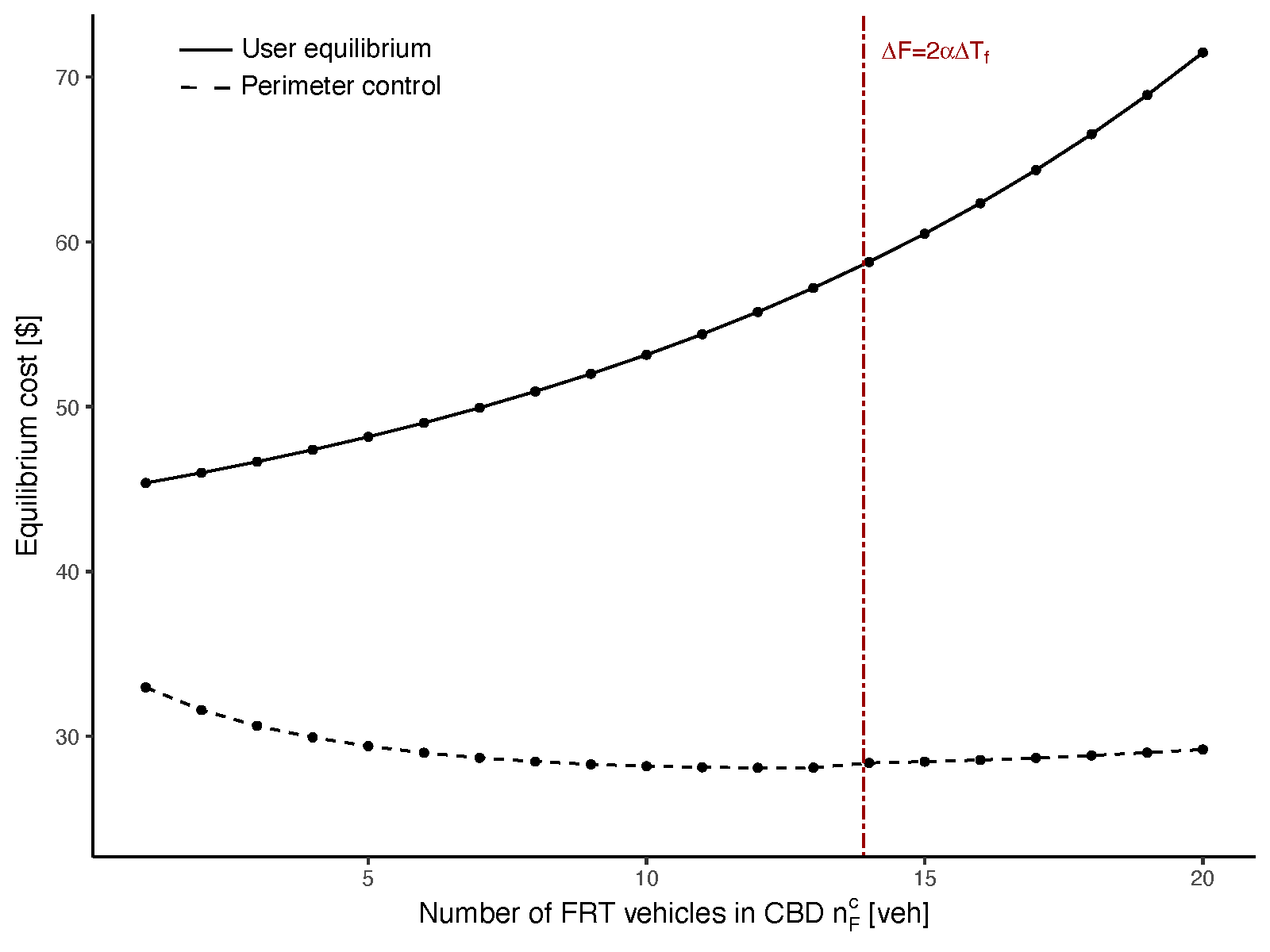}
    \subcaption{Equilibrium cost}
    \end{minipage}
    \caption{Sensitivity analysis  to number of FRT vehicles  (the vertical dotted line represents the number of FRT vehicles; $\Delta F = 2 \alpha \Delta T_f$ )}
    \label{fig:sensitivity_nF}
\end{figure}
\begin{table} 
\caption{Sensitivity analysis to number of FRT vehicles}
\centering
\begin{tabular}{cccccc}
\hline \hline 
& User equilibrium && Perimeter control & &\\
$n_F^c$  & $c^*$ & FRT mode share [\%] & $c_p^*$ & FRT mode share [\%] & $\frac{c_p^*}{c^*}$ \\ \hline
1 & 45.4 & 1.0 & 33.0 & 11.5  & 0.73\\
5 & 48.2 & 4.7 & 29.4 & 34.2  & 0.61\\
10 & 53.1 & 8.0 & 28.2 & 47.7  & 0.53\\
15 & 60.5 & 9.9 & 28.5 & 55.3  & 0.47\\
20 & 71.5 & 10.6 & 29.2 & 61.2 & 0.41\\
\hline 
\end{tabular}\label{tab:sensitivity_nF}
\end{table}

\subsection{Sensitivity analysis with respect to number of FRT vehicles}
We then analyze the sensitivity with respect to the number of FRT vehicles in the CBD zone, $n_F^c$, as shown in Fig.~\ref{fig:sensitivity_nF} and Table.~\ref{tab:sensitivity_nF}. The  vertical dotted line in Fig.~\ref{fig:sensitivity_nF} corresponds to the cost which yields $\Delta F = 2 \alpha \Delta T_f$; $13 < n_F^c < 14$. We find that, as more FRT vehicles are operated, more passengers take FRT. We also find from Fig.~\ref{fig:sensitivity_nF} (b) and Table.~\ref{tab:sensitivity_nF} that the increase in the number of FRT vehicles simply increases the user equilibrium cost, whereas the increase in the number of FRT vehicles decreases the equilibrium cost under perimeter control when a few FRT vehicles are operated; however it increases the cost if too many FRT vehicles are operated. This is because the increase in the number of FRT vehicles not only alleviates the discomfor (discomfort effect hereafter), but also decreases the network capacity (network capacity effect hereafter). The network capacity effect is larger in the presence of hypercongestion and result in an increase in the user equilibrium cost. Conversely, as hypercongestion does not occur under perimeter control, the network capacity effect is relatively smaller. Therefore, the discomfort effect decreases the equilibrium cost until a certain number of vehicles is reached ($13$ [veh] in this example). If the number of FRT vehicles exceeds this threshold, then the network capacity becomes larger and causes an increase in the equilibrium cost. This implies that operating many FRT vehicles to promote  mode shift from cars during  perimeter control may increase the equilibrium cost.  

To investigate the impact of perimeter control, we estimate the ratio of the equilibrium cost with perimeter control to the user equilibrium cost as shown in the last column of Table.~\ref{tab:sensitivity_nF}.  We find that the impact of perimeter control increases with the number of FRT vehicles. This is because the duration of perimeter control is long as higher number of FRT vehicles causes a severe capacity drop (i.e., long duration of hypercongestion).

\begin{figure}[t]
    \centering
        \begin{minipage}{0.45\columnwidth}
    \centering
    \includegraphics[width=\columnwidth]{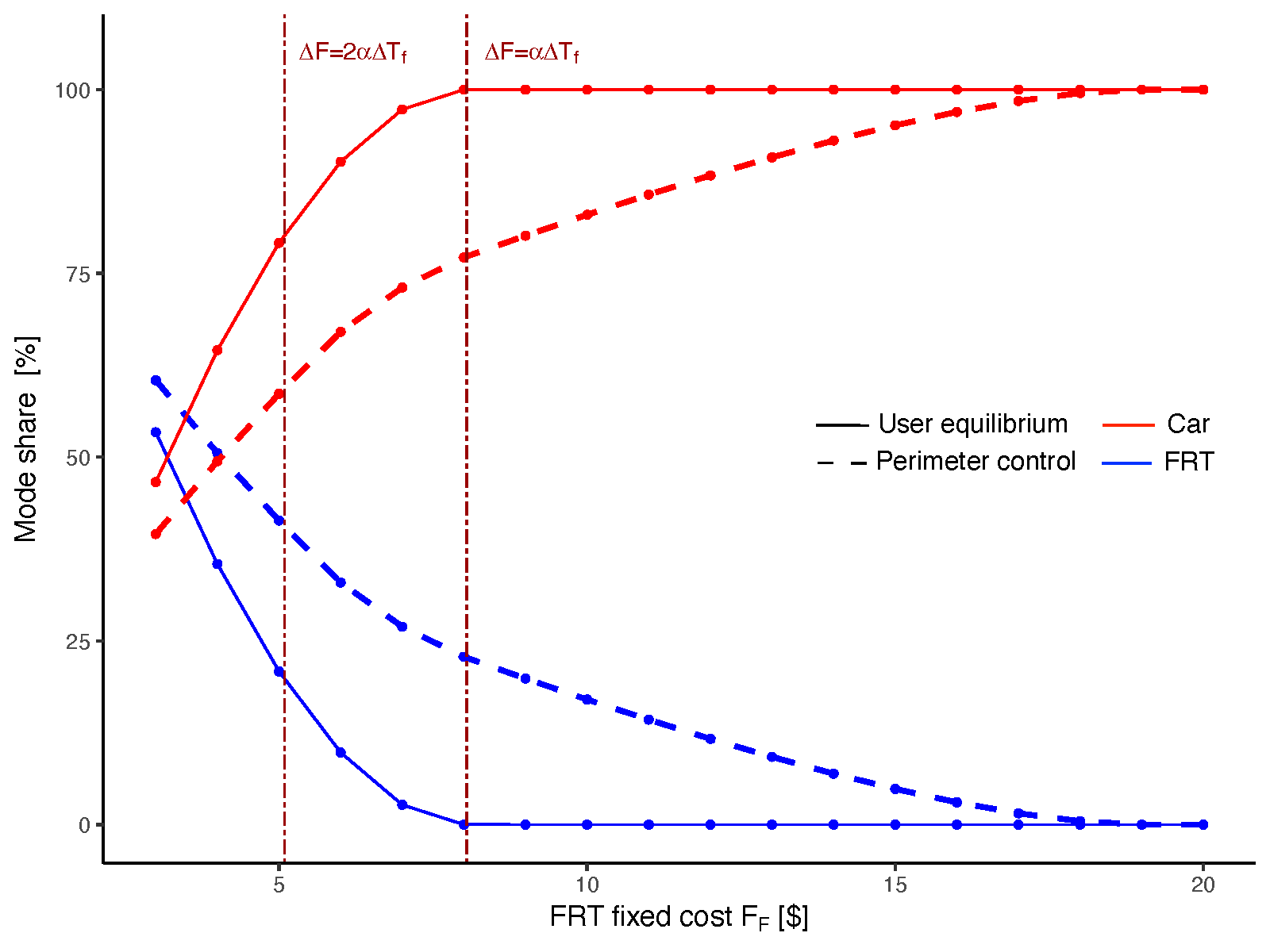}
    \subcaption{Mode share}
    \end{minipage} 
    \begin{minipage}{0.45\columnwidth} 
    \centering
    \includegraphics[width=\columnwidth]{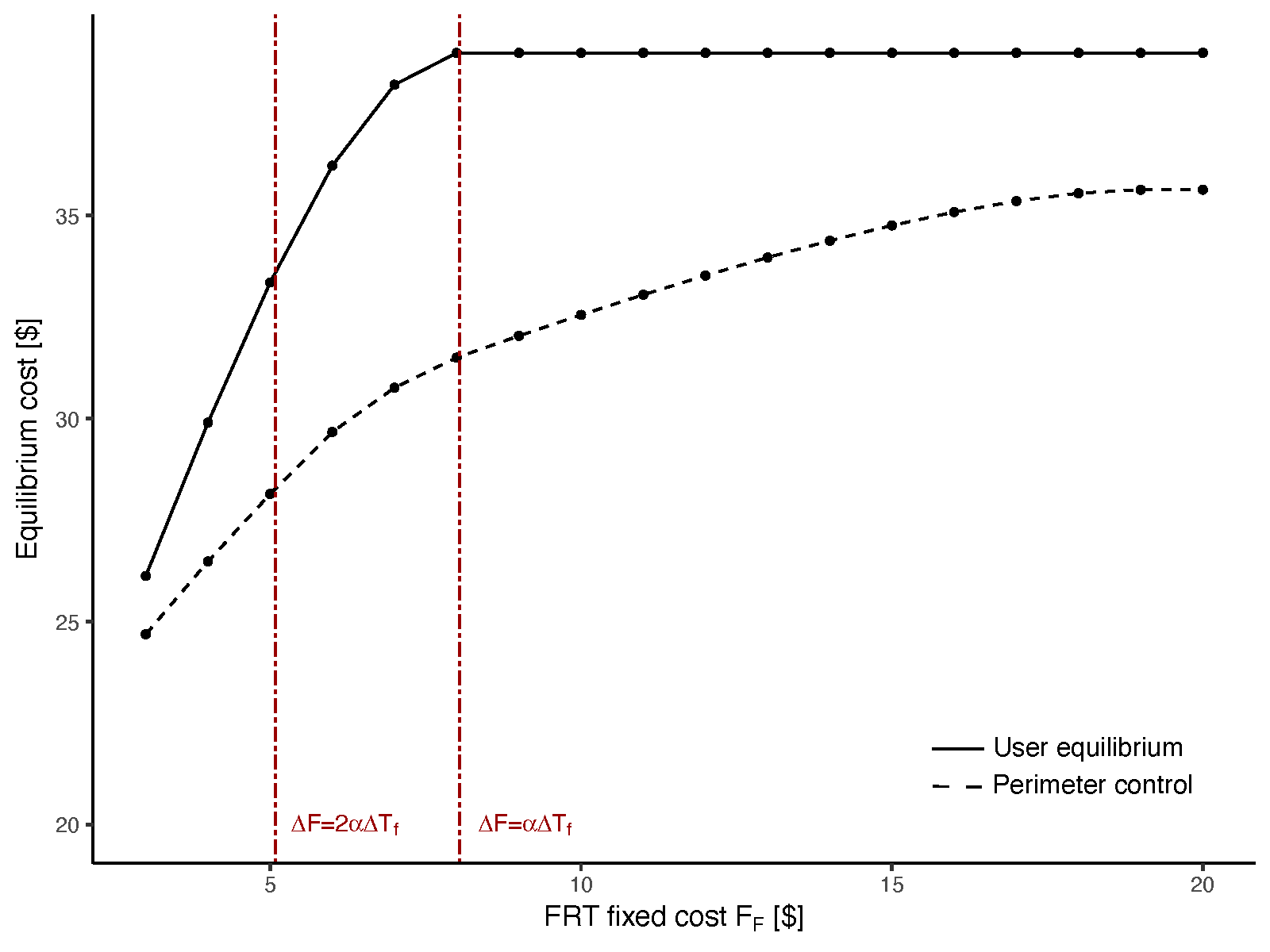}
    \subcaption{Equilibrium cost}
    \end{minipage}
    \centering
    \caption{Sensitivity analysis  to FRT fixed cost   (the left and right vertical dotted lines represent the FRT fixed cost; $\Delta F = 2 \alpha \Delta T_f$ and $\Delta F =  \alpha \Delta T_f$, respectively)}
    \label{fig:sensitivity_FRT_cost}
\end{figure}

\begin{table}[t]
\caption{Sensitivity analysis to FRT fixed cost}
\centering
\begin{tabular}{cccccc}
\hline \hline 
& User equilibrium && Perimeter control & &\\
$F_F$  & $c^*$ & FRT mode share [\%] & $c_p^*$ & FRT mode share [\%] & $\frac{c_p^*}{c^*}$ \\ \hline
3 & 26.1 & 53.3 & 24.7 & 60.5  & 0.95\\
5 & 33.4 & 20.9 & 28.1 & 41.4  & 0.84\\
8 & 39.0 & 0.0 & 31.5 & 22.8  & 0.81\\
10 & 39.0 & 0.0 & 32.6 & 17.0 & 0.83\\
15 & 39.0 & 0.0 & 34.8 & 4.9 &  0.89\\
20 & 39.0 & 0.0 & 35.6 & 0.0 &  0.91 \\
\hline
\end{tabular} \label{tab:sensitivity_FRT_cost}
\end{table}

\subsection{Sensitivity analysis with respect to FRT fixed cost}
Next, we conduct a sensitivity analysis with respect to the FRT fixed cost\footnote{ Note that the model can be easily extended to a case where the heterogeneity exists in the marginal cost of travel time  (i.e.,  $\alpha$ is higher or lower for FRT commuters than  car commuters).  If  the heterogeneity is incorporated, it will have similar effect to the FRT fixed cost. An increase in  $\alpha$ for FRT commuters results in a decrease in  the FRT mode share. }  , as depicted in  Fig. \ref{fig:sensitivity_FRT_cost} and Table. \ref{tab:sensitivity_FRT_cost}. Note that we use $L_F=7$ [mile], $F_c=11$ [\$], $N=200$ [pax] in this analysis. The left and right vertical dotted lines in  Fig. \ref{fig:sensitivity_FRT_cost} indicate the costs that yield $\Delta F = 2 \alpha \Delta T_f$; $F_F=5.09$ and $\Delta F =  \alpha \Delta T_f$; $F_F=8.05$, respectively. The left side of the left line shows the case where FRT is used from the beginning to the end during rush hour, the middle part between the left and right lines shows the case where there is a time window wherein FRT is not used regardless of perimeter control, and the right side of the right line shows the case where FRT is never used at user equilibrium, whereas FRT is used, but there is a time window when FRT is not used at equilibrium with perimeter control. The equilibrium cost can be reduced by  perimeter control in all cases from the user equilibrium cost. An increase in FRT fixed cost results in an increase in equilibrium cost under perimeter control  because it increases the mode share of cars (Fig.~\ref{fig:sensitivity_FRT_cost} (a)) and lengthens the car rush hour. Transit priority is still beneficial, even in the case where FRT is never used at user equilibrium,  In the case where the FRT fixed cost is higher than $8.05$ [\$], all commuters use their cars, and the FRT mode share is zero at user equilibrium (Fig.~\ref{fig:sensitivity_FRT_cost} (a) and Table. \ref{tab:sensitivity_FRT_cost}). Perimeter control with transit priority can increase the mode share of FRT.  However, when the FRT fixed cost is too high (e.g., $F_F=20 $ [\$] as shown in Table.~\ref{tab:sensitivity_FRT_cost}) , FRT is not used during perimeter control because  $\theta_p > (2 \alpha T_f^F - \Delta F)/(\alpha T_f^c)$ is not satisfied, as verified in Proposition \ref{prop:FRT_not_use_perimeter}.3. Since the mode share is unchanged, the effect of perimeter control with transit priority is the same as that of perimeter control alone. Bimodality is  not fully utilized in this case. Thus, additional incentives, such as fare reduction so that $\theta_p > (2 \alpha T_f^F - \Delta F)/(\alpha T_f^c)$ is satisfied, should be provided to promote FRT use during perimeter control.

To investigate the impact of perimeter control, we estimate the ratio of the equilibrium cost  with perimeter control to the user equilibrium cost as shown in the last column of Table.~\ref{tab:sensitivity_FRT_cost}. An increases in the FRT fixed cost increases the impact of perimeter control when  $\Delta F >  \alpha \Delta T_f$. This is because an increase in the FRT fixed cost increases the mode share of cars, leading to a longer duration of hypercongestion at user equilibrium. Thus, this causes longer perimeter control and increasing  the impact of perimeter control. In contract,  the higher the FRT fixed cost is, the lower will be the impact of perimeter control is when $\Delta F \leq  \alpha \Delta T_f$. Since all commuters use their cars, the mode share and the duration of hypercongestion remain unchanged at user equilibrium regardless of the value of the FRT fixed cost. During perimeter control, an increase in the FRT fixed cost leads to a higher mode share of cars, and thus a higher equilibrium cost if $\theta_p > (2 \alpha T_f^F - \Delta F)/(\alpha T_f^c)$. Therefore, the impact of perimeter control decreases.

\begin{figure}[t]
\begin{minipage}{0.5\columnwidth}
\centering
\includegraphics[width=\columnwidth]{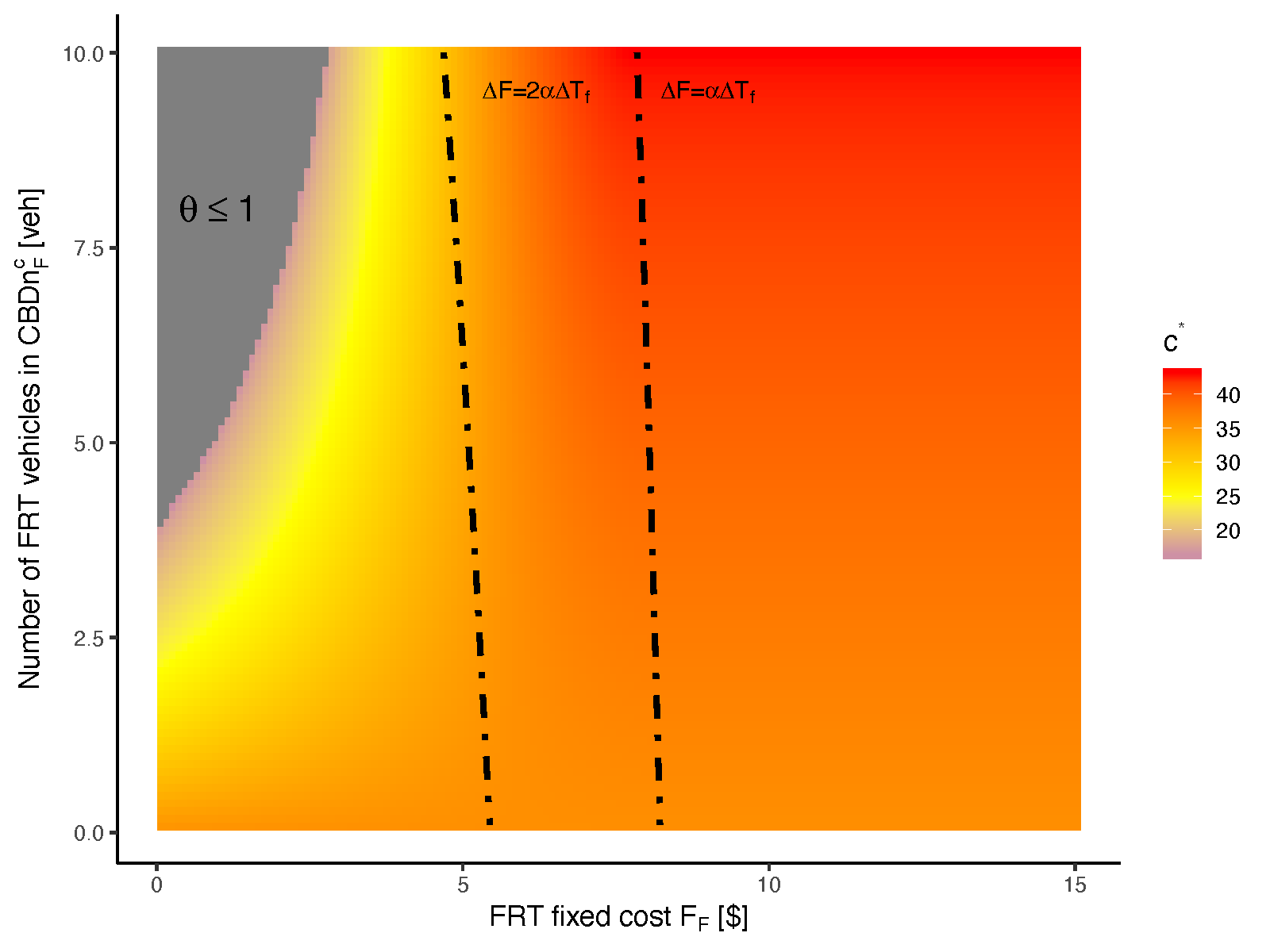}
\subcaption{User equilibrium}
\end{minipage}
\begin{minipage}{0.5\columnwidth}
\centering
\includegraphics[width=\columnwidth]{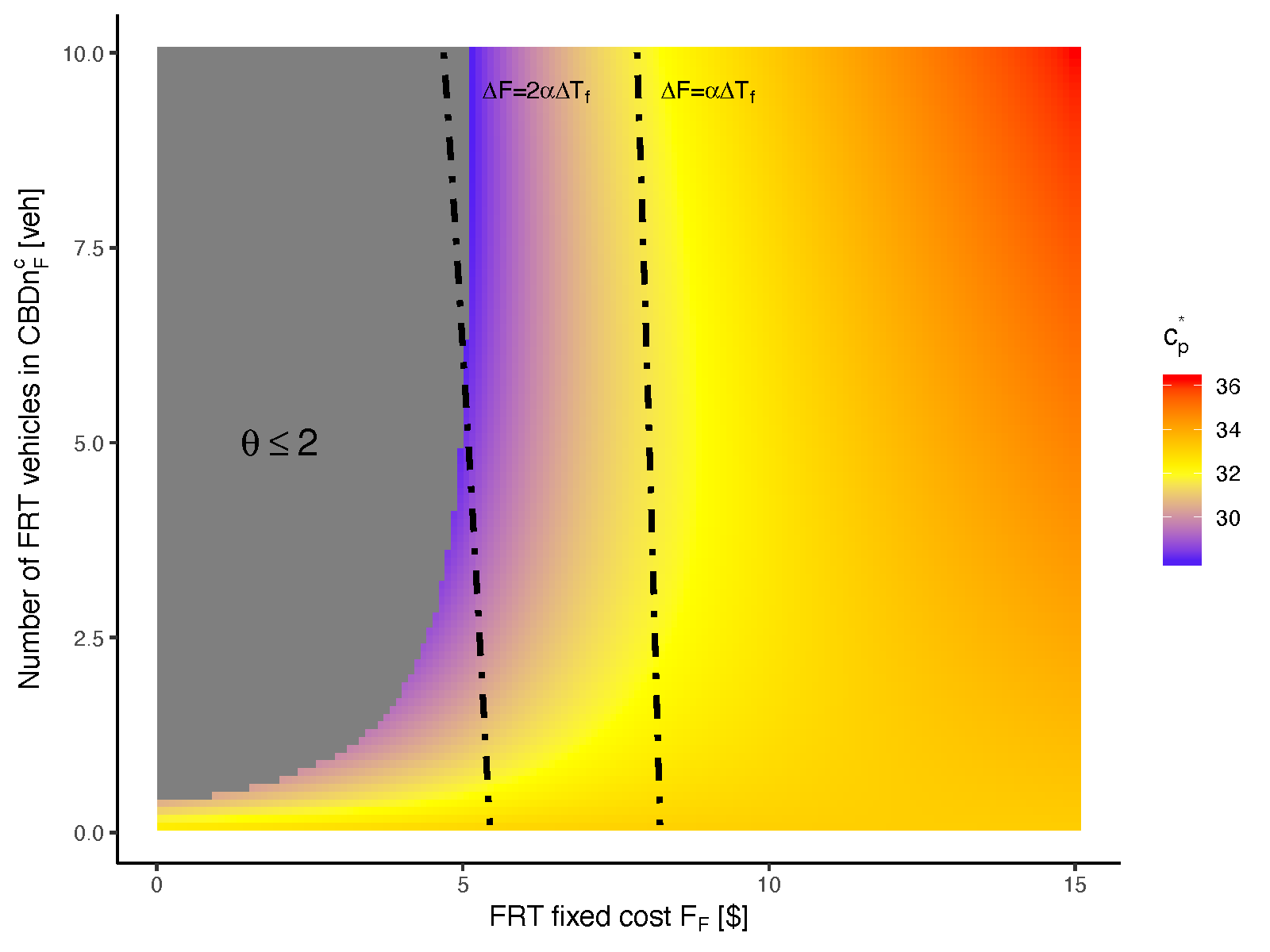}
\subcaption{Perimeter control}
\end{minipage}
\caption{Sensitivity analysis of equilibrium cost}
\label{fig:cost__full_sensitivity}
\begin{minipage}{0.5\columnwidth}
\centering
\includegraphics[width=\columnwidth]{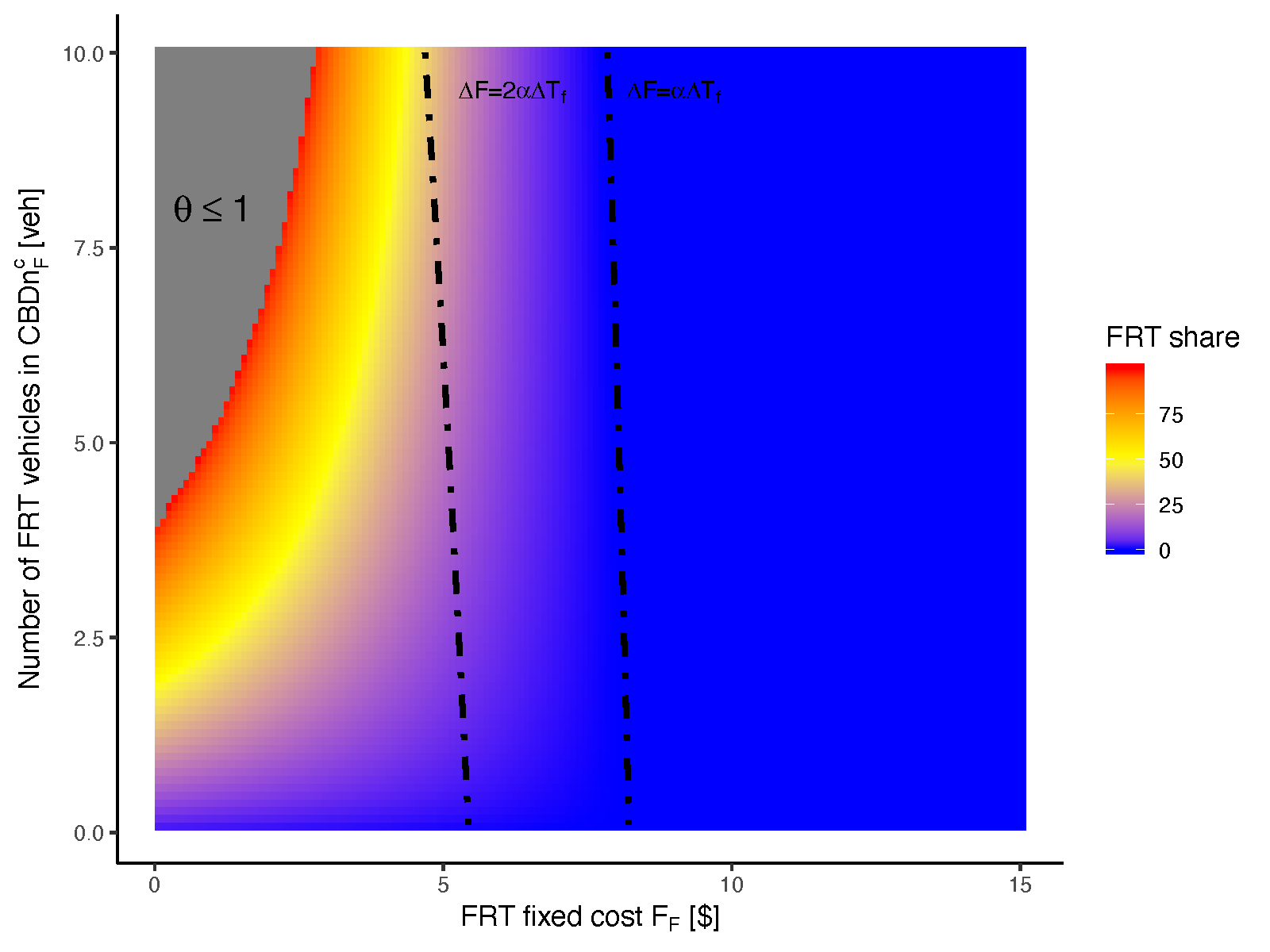}
\subcaption{User equilibrium}
\end{minipage}
\begin{minipage}{0.5\columnwidth}
\centering
\includegraphics[width=\columnwidth]{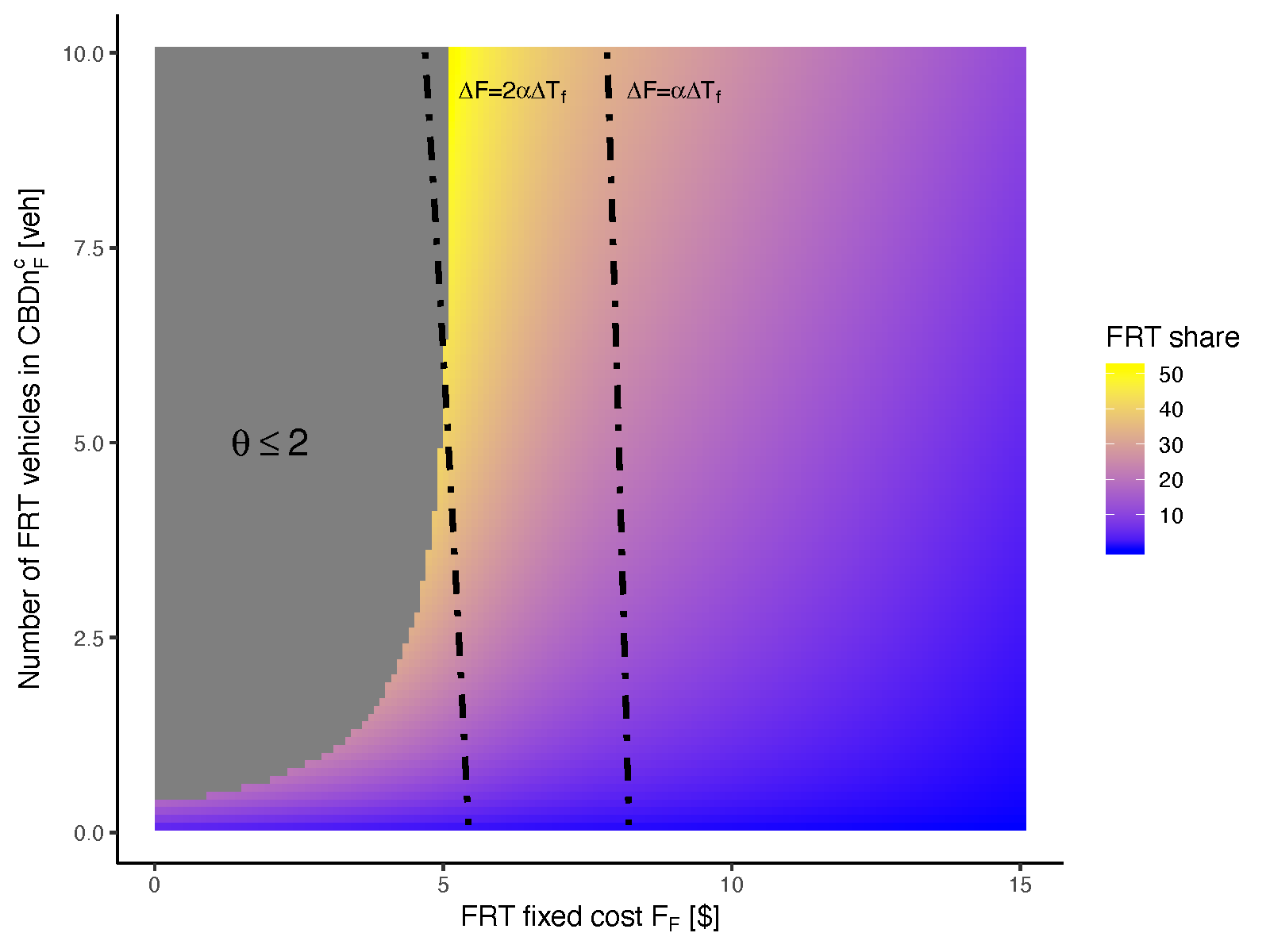}
\subcaption{Perimeter control}
\end{minipage}
\caption{Sensitivity analysis of mode share}
\label{fig:FRTshare_full_sensitivity}
\end{figure}

\subsection{Sensitivity analysis with respect to number of FRT vehicles and FRT fixed cost}
Further, we  conduct a sensitivity analysis with respect to both of the number of FRT vehicles and FRT fixed cost to better understand the implications. Note that we use the same parameters as in Section 6.2.  If  many FRT vehicles are operated and FRT fixed cost is lowered (i.e.,  $\theta \leq 1$ in the gray coloured areas of  Figs.~\ref{fig:cost__full_sensitivity} (a) and \ref{fig:FRTshare_full_sensitivity} (a)), all commuters will take FRT. Although an increase in the FRT fixed cost and a decrease in the number of FRT vehicles increase the mode share of cars, hypercongestion never occurs as long as $\theta \leq 2$ (the gray coloured areas of  Figs.~\ref{fig:cost__full_sensitivity} (b) and \ref{fig:FRTshare_full_sensitivity} (b)). We also find that an increase in the number of FRT vehicles reduces the equilibrium cost when FRT fixed cost is lower  at user equilibrium without perimeter control (see the left side of the dotted line which is the cost yielding $\Delta F = 2 \alpha \Delta T_f$ in  Fig.~\ref{fig:cost__full_sensitivity} (a)). Because hypercongestion never occur or occur for a short period in most cases, an increase in the number of FRT vehicles simply alleviates the discomfort for FRT, thereby  resulting in higher mode share of cars and shortening the  rush hour. However, if the FRT fixed cost is high,  hypercongestion occurs due to  higher mode share of cars. An increase in the number of FRT vehicles leads to a severe capacity drop, resulting in a higher equilibrium cost. Moreover, the mode share of FRT does not increase significantly due to the network capacity effect in this case, as depicted in the middle part of the dotted lines in Fig.~\ref{fig:FRTshare_full_sensitivity} (a)  

When perimeter control is implemented, hypercongestion never occur. Therefore, even when the FRT fixed cost is the cost which yields $\Delta F < 2 \alpha \Delta T_f$ and $\Delta F >  \alpha \Delta T_f$ (the middle part of the dotted lines in Fig.~\ref{fig:FRTshare_full_sensitivity} (b)), the mode share of FRT can be increased due to an increase in the number of FRT vehicles, which results in a lower equilibrium cost as shown in Fig.~\ref{fig:cost__full_sensitivity} (b). However, as discussed in Section 6.1, even though an increase in the number of FRT vehicles increases the mode share of FRT, the equilibrium cost will increase only when the FRT fixed cost is higher (the right side of the dotted line which is the cost yielding $\Delta F =  \alpha \Delta T_f$ in  Fig.~\ref{fig:cost__full_sensitivity} (b)).   

\vspace{5mm}

We showed that bimodality is an important component for efficient perimeter control, and can be enhanced by providing transit priority and incentives to lower the FRT fixed cost (if necessary). However, the operation of too many vehicles that reduces discomfort may produce an unintended result: operating many FRT vehicles may increase the equilibrium cost.

\section{Conclusions} \label{seq:conclusion}
We investigated traveler behavior changes in response to  perimeter control with  transit priority in a bimodal transportation system (cars and flexible route transit (FRT)).  To this end, we developed a bimodal bathtub model; it is a departure time and transportation mode choices model that incorporates the hypercongestion and FRT discomfort. First, some conditions of the equilibrium patterns were characterized: (I) If the difference of the car fixed cost from the FRT fixed cost is higher than the difference of the FRT travel time cost from the car travel time cost  at free-flow, then both modes are used, and the FRT rush hour is longer than the car rush hour. (II) The number of FRT passengers decreases with time toward the desired arrival time. (III) There may be a time interval wherein FRT is not used around the rush hour peak. Then, we examined commuter behavior changes during  perimeter control with  transit priority. We found that if the fixed cost difference is larger than  the  travel time cost after perimeter control starts, then FRT is used during perimeter control; otherwise,  FRT is not used during perimeter control. This implies that transit
priority may not be sufficient to promote the use of FRT and that additional incentives may be required to promote FRT use during perimeter control. Finally, we conducted numerical analyses to investigate the equilibrium patterns of the model. We showed that operating many FRT vehicles to promote mode shift from cars during perimeter control may increase equilibrium cost

There are several directions for future works. Although we showed the theoretical results of the bimodal bathtub model, the empirical measurements that incorporate hypercongestion, discomfort of transit users and travel behavior need to be conducted. This analysis will elucidate our discussion further.   Second, it would be interesting to extend the model to multiregion models. Since the regional route choice model will be considered, multiregion models that consider regional route, departure time, and transportation mode choices will no longer be tractable.  Surrogate techniques \citep[e.g.,][]{ge2020reservoir, dantsuji2022novel} can be used to address this issue. Third, although we assume free-flow conditions in the suburban zone for the tractability,  recent research \citep[e.g.,][]{keyvan2015multiple} has shown the impact of the congestion propagation from the perimeter boundaries on the MFD dynamics in the suburban zone.  Thus, a model extension in this direction would be an interesting.    Finally, the urban structure was assumed to be fixed in this study. As urban structures change in response to traffic demand management \citep{takayama2017bottleneck, fosgerau2018vickrey, takayama2020gains}, this would be an interesting direction.

\section*{Acknowledgements}
We thank Se-il Mun, Eiji Hato and Tsubasa Takeda for their valuable comments on the early version of this work. This work was supported by JST ACT-X, Japan (grant \#JPMJAX21AE), by JST FOREST Program, Japan (grant  \#JPMJFR215M), by JSPS KAKENHI Grant-in-Aid  for Scientific Research (B), Japan  (grant  \#22H01610),  and by the Committee on Advanced Road Technology (CART), Ministry of Land, Infrastructure, Transport, and Tourism, Japan (grant \#2020-2).

\section*{Appendix A: Proof of Proposition \ref{prop:uniqueness_ue} }
We prove that the equilibrium cost is uniquely determined in the case where there is a time window wherein FRT is not used. Let $F(\theta)$ be the RHS of Eq. (\ref{eq:total_number_ue}) (i.e., $N=F(\theta)$). If the function of $F(\theta)$ is strictly monotonic with respect to $\theta$, then the equilibrium cost $c^*$ is uniquely determined. As $\theta$ is defined as the ratio of the sum of the travel time and schedule delay costs at equilibrium to the
free-flow travel time, we have $\theta>1$. Thus, we prove the monotonicity of $F(\theta)$ for $\theta>1$. We derive the first derivative of $F(\theta)$ from Eq.~(\ref{eq:total_number_ue}) as
\begin{eqnarray}
\frac{dF(\theta)}{d \theta} =\alpha n_j' \left( \frac{1}{\beta} + \frac{1}{\gamma} \right) \left( \frac{1}{\theta} - \frac{1}{\theta^2} \right)
\end{eqnarray}
For $\theta>1$, $dF/d\theta$ is strictly monotonically increasing. Thus, the equilibrium cost is uniquely determined. 

The first derivative of $F(\theta)$ from Eq.~(\ref{eq:total_number_ue_continous}) for the case where FRT is used continuously during the FRT rush hour is
\begin{eqnarray}
\frac{dF(\theta)}{d \theta} =\left( \frac{1}{\beta} + \frac{1}{\gamma} \right)  \left( \alpha n_j' \left( \frac{1}{\theta} - \frac{1}{\theta^2} \right) + n_F^c \frac{\alpha}{\lambda} \frac{T_f^c}{T_f^F} \left( \Delta F \frac{1}{\theta} + \alpha \Delta T_f \right) \right) 
\end{eqnarray}
The condition for $dF/d\theta>0$ is
\begin{eqnarray}
 \alpha n_j' \left( \frac{1}{\theta} - \frac{1}{\theta^2} \right) + n_F^c \frac{\alpha}{\lambda} \frac{T_f^c}{T_f^F} \left( \Delta F \frac{1}{\theta} + \alpha \Delta T_f \right)  > 0
\end{eqnarray}
Multiplication of $\theta^2$ and some tedious calculations yield
\begin{eqnarray}
\alpha \Delta T_f  n_F^c \frac{\alpha}{\lambda} \frac{T_f^c}{T_f^F} \theta^2 + \left(  \alpha n_j' + \Delta F  n_F^c \frac{\alpha}{\lambda} \frac{T_f^c}{T_f^F} \right) \theta -  \alpha n_j' > 0
\end{eqnarray} 
This condition is satisfied for $\theta>1$. Thus, the equilibrium cost is uniquely determined. 

\section*{Appendix B: Proof that perimeter control is never implemented if Assumption 1 holds and $\theta_p \leq (2 \alpha T_f^F - \Delta F) /\alpha T_f^c$ }
As  $\theta^p\equiv\frac{(c_p^*-F_c)v_f'}{\alpha L_c}$,   $\theta_p \leq \frac{2 \alpha T_f^F - \Delta F}{\alpha T_f^c}$ can  be rewritten by 
\begin{eqnarray} \label{eq: proof1}
\frac{(c_p^*-F_c)v_f'}{\alpha L_c} \leq \frac{2 \alpha T_f^F - \Delta F}{\alpha T_f^c}
\end{eqnarray}
Since $c_p^* = \alpha T_f^c + \beta (t^* - t_s^c) +F_c$, Eq. (\ref{eq: proof1}) yields, 
\begin{eqnarray} 
\beta (t^* - t_s^c)  - \alpha T_f^F \leq \alpha \Delta T_f - \Delta F 
\end{eqnarray}
If Assumption 1 holds (i.e., $\Delta F > \alpha \Delta T_f$), 
\begin{eqnarray} 
\beta (t^* - t_s^c)  - \alpha T_f^F < 0 
\end{eqnarray}
If perimeter control is implemented, the travel cost  is equal to the equilibrium cost when car accumulation reaches  the critical level. $c_p^* = 2\alpha T_f^c + \beta (t^* - t_s^p) +F_c = \alpha T_f^c + \beta (t^* - t_s^c) +F_c$ yields, 
\begin{eqnarray} 
\beta (t^* - t_s^p)  = \frac{1}{\beta} \left( \beta \left( t^* - t_s^c \right) - \alpha T_f^c \right)
\end{eqnarray}
As $\alpha T_f^F > \alpha T_f^c$ and $\beta (t^* - t_s^c)  - \alpha T_f^F < 0 $, the desired arrival time is earlier than the time when accumulation reaches  the critical level (i.e., $(t^* - t_s^p)< 0$) if Assumption 1 holds. Thus, perimeter control is never implemented in this case.

\section*{Appendix C: Proof of Proposition \ref{prop:uniqueness_perimeter_control} }
We prove that the equilibrium cost during perimeter control is uniquely determined in the case where  FRT is used from the beginning to the end  during  the FRT rush hour.  Let $F(\theta^p)$ be the RHS of Eq.~(\ref{eq:total_number_perimeter_case2}). Then, we have
\begin{eqnarray}
\frac{dF(\theta^p)}{d\theta^p} = \left( \frac{1}{\beta} + \frac{1}{\gamma} \right)  \left( \frac{\alpha n_j'}{4} +  n_F^c \frac{\alpha}{\lambda} \frac{T_f^c}{2T_f^F} \alpha T_f^c (\theta^p-2)   \right)
\end{eqnarray}
It yields
\begin{eqnarray}
\frac{dF(\theta^p)}{d\theta^p} = 
\begin{cases}
\leq 0  & \text{if } \theta^p \leq 2 \left( 1 - \frac{\alpha n_j'}{4} \frac{1}{n_F^c} \frac{\lambda}{\alpha} \frac{T_f^F}{T_f^c} \frac{1}{\alpha T_f^c} \right) \\
> 0  & \text{otherwise } 
\end{cases}
\end{eqnarray}
Combining this with the fact that $\theta^p>1$ and $F(1)=0$, since $\theta^p=1$ states that the equilibrium cost is equal to the free-flow travel time and fixed costs (equilibrium cost when $N=0$), produces a unique equilibrium cost. 

Next, we prove that the equilibrium cost during perimeter control is uniquely determined in the case where FRT is  used, but there will be a time window wherein FRT is not used  during the FRT rush hour. Let $F(\theta^p)$ be the RHS of Eq. (\ref{eq:total_number_perimeter_case1}). Then, the first derivative of $F(\theta^p)$ is 
\begin{eqnarray}
\frac{dF(\theta^p)}{d\theta^p} = \left( \frac{1}{\beta} + \frac{1}{\gamma} \right)  \left( \frac{\alpha n_j'}{4} +  n_F^c \frac{\alpha}{\lambda} \frac{T_f^c}{2T_f^F} \alpha T_f^c \left(\theta^p- \frac{2\alpha T_f^F - \Delta F}{\alpha T_f^c} \right)   \right)
\end{eqnarray}
It is  positive since $\theta^p > \frac{2\alpha T_f^F - \Delta F}{\alpha T_f^c}$. Thus,  $dF/d\theta^p>0$  yields a unique equilibrium cost.

Similarly, we can prove that the equilibrium cost is uniquely determined for Eqs.~(\ref{eq:total_number_perimeter_case3}) and (\ref{eq:total_number_perimeter_case4}).

\section*{Appendix D: List of variables}
\begin{table}[h]
\centering
\caption{The list of notations in the paper}
\label{tab:notation}

\begin{tabular}{cl}
 \hline \hline
Notation  & Description \\ \hline
$n_F$ &  Total number of FRT vehicles   \\ 
$n_F^s, n_F^c$ & Number of FRT vehicles in the suburban and CBD zones, respectively  \\ 
$v_i(t)$ & Space-mean speed of transportation mode $i$ in the CBD zone  \\ 
$v_f$ & Free-flow car space-mean speed in the CBD zone \\
$n_i(t)$ & Accumulation of mode $i$'s vehicles in the CBD zone   \\ 
$\eta$ & Passenger car unit  \\ 
$n_j$ & Jam accumulation  \\
$n_c^{cr}$ & Critical car accumulation ($=n_j^{'} /2$) \\
$m$ &  Parameter that captures lower FRT's speed than car's speed in the CBD zone   \\
$v_f'$ &  $=v_f \left( 1- \eta n_F^c / n_j \right)$ \\ 
$n_j'$ &  $= n_j \left( 1- \eta n_F^c / n_j \right)$ \\ 
$I_c(t)$ &  Car inflow rate to the CBD zone at time $t$ \\
$I_c^p$ & Inflow rate during perimeter control ($=n_j^{'}v_f^{'}/(4L_c)$) \\
$G_c(t), G_F^p(t)$  & Car and FRT passenger arrival rates at the destination at time $t$  \\
$L_i$ & Average trip length of mode $i$'s commuters in the CBD zone \\
$d_F(t)$ &  Departure rate of FRT passengers at time $t$ \\
$C_i(t)$ & Mode $i$'s travel cost incurred by a commuter who arrives at their destination at time $t$ \\
$T_i(t)$ & Travel time in the CBD zone of a commuter who arrives at time $t$ by transportation mode $i$ \\
$T_f^i$ & Mode $i$'s free-flow travel time in the CBD zone  \\
$T_i^s$ & Free-flow mode $i$'s travel time in the suburban zone \\
$p_c, p_F$ & Car operation cost and FRT fare, respectively  \\
$F_i$ & Fixed cost for transportation mode $i$  \\
$\alpha, \beta, \gamma, \lambda$ & Marginal costs for travel time, earliness,  lateness, and discomfort, respectively \\
$g(t)$ & Discomfort cost at time $t$ \\ 
$\Delta F$ & $=F_c - F_F$ \\
$\Delta T_f$ & $= T_f^F - T_f^c$ \\
$c^*, c_p^*$ & Equilibrium cost without and with perimeter control \\
$\theta$ & $\equiv (c^* - F_c)v_f^{'}/(\alpha L_c)$ \\
$\theta_p$ & $\equiv (c^*_p - F_c)v_f^{'}/(\alpha L_c)$ \\
$t_s^i, t_e^i$ & Start and end times of mode $i$'s rush hour, respectively \\
$t_{ee}^F, t_{sl}^F$ & Start and end times of a time window wherein FRT is not used during the FRT rush hour, respectively \\
$t_s^p, t_e^p$ & Start and end times of perimeter control implemented \\
$N$ & Total number of commuters \\
$N_i$ & Total number of mode $i$'s commuters \\
$N_i^p$ & Number of mode $i$'s commuters who arrives at their destination during perimeter control \\
$N_F^{oc}$ & Number of commuters who take FRT outside the car rush hour \\
$N_i^{op}$ & Number of mode $i$'s commuters outside the perimeter control period \\ 
$A_b(t)$ & Car arrival rate at the perimeter boundary at time $t$ \\
$T_b(t)$ & Waiting time at the perimeter boundary of a commuter who arrives at their destination at time $t$ \\ 
$q(t)$ & Number of cars queued at the perimeter boundary when a commuter who arrives \\
& at their destination at time $t$ reaches the boundary \\
\hline
\end{tabular}
\end{table}

\bibliographystyle{elsarticle-harv}\bibliography{ref}

\begin{thebibliography}{77}
\expandafter\ifx\csname natexlab\endcsname\relax\def\natexlab#1{#1}\fi
\expandafter\ifx\csname url\endcsname\relax
  \def\url#1{\texttt{#1}}\fi
\expandafter\ifx\csname urlprefix\endcsname\relax\def\urlprefix{URL }\fi

\bibitem[{Amirgholy and Gao(2017)}]{amirgholy2017modeling}
Amirgholy, M., Gao, H.~O., 2017. Modeling the dynamics of congestion in large
  urban networks using the macroscopic fundamental diagram: User equilibrium,
  system optimum, and pricing strategies. Transportation Research Part B:
  Methodological 104, 215--237.

\bibitem[{Ampountolas et~al.(2017)Ampountolas, Zheng, and
  Geroliminis}]{ampountolas2017macroscopic}
Ampountolas, K., Zheng, N., Geroliminis, N., 2017. Macroscopic modelling and
  robust control of bi-modal multi-region urban road networks. Transportation
  Research Part B: Methodological 104, 616--637.

\bibitem[{Arnott(2013)}]{arnott2013bathtub}
Arnott, R., 2013. A bathtub model of downtown traffic congestion. Journal of
  Urban Economics 76, 110--121.

\bibitem[{Arnott and Buli(2018)}]{arnott2018solving}
Arnott, R., Buli, J., 2018. Solving for equilibrium in the basic bathtub model.
  Transportation Research Part B: Methodological 109, 150--175.

\bibitem[{Arnott et~al.(1993)Arnott, de~Palma, and
  Lindsey}]{arnott1993structural}
Arnott, R., de~Palma, A., Lindsey, R., 1993. A structural model of peak-period
  congestion: A traffic bottleneck with elastic demand. The American Economic
  Review, 161--179.

\bibitem[{Bao et~al.(2021)Bao, Verhoef, and Koster}]{bao2021leaving}
Bao, Y., Verhoef, E.~T., Koster, P., 2021. Leaving the tub: The nature and
  dynamics of hypercongestion in a bathtub model with a restricted downstream
  exit. Transportation Research Part E: Logistics and Transportation Review
  152, 102389.

\bibitem[{Bartholdi~III and Eisenstein(2012)}]{bartholdi2012self}
Bartholdi~III, J.~J., Eisenstein, D.~D., 2012. A self-co{\"o}rdinating bus
  route to resist bus bunching. Transportation Research Part B: Methodological
  46~(4), 481--491.

\bibitem[{Basso et~al.(2019)Basso, Feres, and Silva}]{basso2019efficiency}
Basso, L.~J., Feres, F., Silva, H.~E., 2019. The efficiency of bus rapid
  transit (brt) systems: A dynamic congestion approach. Transportation Research
  Part B: Methodological 127, 47--71.

\bibitem[{Chen et~al.(2022{\natexlab{a}})Chen, Fu, Wu, Wang, and
  Qiao}]{chen2022passenger}
Chen, S., Fu, H., Wu, N., Wang, Y., Qiao, Y., 2022{\natexlab{a}}.
  Passenger-oriented traffic management integrating perimeter control and
  regional bus service frequency setting using 3d-pmfd. Transportation Research
  Part C: Emerging Technologies 135, 103529.

\bibitem[{Chen et~al.(2022{\natexlab{b}})Chen, Wu, Huang, and
  Shang}]{chen2022modeling}
Chen, Z., Wu, W.-X., Huang, H.-J., Shang, H.-Y., 2022{\natexlab{b}}. Modeling
  traffic dynamics in periphery-downtown urban networks combining vickrey's
  theory with macroscopic fundamental diagram: user equilibrium, system
  optimum, and cordon pricing. Transportation Research Part B: Methodological
  155, 278--303.

\bibitem[{Chiabaut(2015)}]{chiabaut2015evaluation}
Chiabaut, N., 2015. Evaluation of a multimodal urban arterial: The passenger
  macroscopic fundamental diagram. Transportation Research Part B:
  Methodological 81, 410--420.

\bibitem[{Daganzo(2007)}]{daganzo2007urban}
Daganzo, C.~F., 2007. Urban gridlock: Macroscopic modeling and mitigation
  approaches. Transportation Research Part B: Methodological 41~(1), 49--62.

\bibitem[{Daganzo and Pilachowski(2011)}]{daganzo2011reducing}
Daganzo, C.~F., Pilachowski, J., 2011. Reducing bunching with bus-to-bus
  cooperation. Transportation Research Part B: Methodological 45~(1), 267--277.

\bibitem[{Dakic and Menendez(2018)}]{dakic2018use}
Dakic, I., Menendez, M., 2018. On the use of lagrangian observations from
  public transport and probe vehicles to estimate car space-mean speeds in
  bi-modal urban networks. Transportation Research Part C: Emerging
  Technologies 91, 317--334.

\bibitem[{Dantsuji et~al.(2021)Dantsuji, Fukuda, and
  Zheng}]{dantsuji2021simulation}
Dantsuji, T., Fukuda, D., Zheng, N., 2021. Simulation-based joint optimization
  framework for congestion mitigation in multimodal urban network: a
  macroscopic approach. Transportation 48~(2), 673--697.

\bibitem[{Dantsuji et~al.(2020)Dantsuji, Hirabayashi, Ge, and
  Fukuda}]{dantsuji2019cross}
Dantsuji, T., Hirabayashi, S., Ge, Q., Fukuda, D., 2020. Cross comparison of
  spatial partitioning methods for an urban transportation network.
  International Journal of Intelligent Transportation Systems Research 18,
  412--421.

\bibitem[{Dantsuji et~al.(2022)Dantsuji, Hoang, Zheng, and
  Vu}]{dantsuji2022novel}
Dantsuji, T., Hoang, N.~H., Zheng, N., Vu, H.~L., 2022. A novel metamodel-based
  framework for large-scale dynamic origin--destination demand calibration.
  Transportation Research Part C: Emerging Technologies 136, 103545.

\bibitem[{de~Palma et~al.(2017)de~Palma, Lindsey, and
  Monchambert}]{de2017economics}
de~Palma, A., Lindsey, R., Monchambert, G., 2017. The economics of crowding in
  rail transit. Journal of Urban Economics 101, 106--122.

\bibitem[{Ding et~al.(2017)Ding, Guo, Zheng, and Zhang}]{ding2017traffic}
Ding, H., Guo, F., Zheng, X., Zhang, W., 2017. Traffic guidance--perimeter
  control coupled method for the congestion in a macro network. Transportation
  Research Part C: Emerging Technologies 81, 300--316.

\bibitem[{Fosgerau(2015)}]{fosgerau2015congestion}
Fosgerau, M., 2015. Congestion in the bathtub. Economics of Transportation
  4~(4), 241--255.

\bibitem[{Fosgerau et~al.(2018)Fosgerau, Kim, and Ranjan}]{fosgerau2018vickrey}
Fosgerau, M., Kim, J., Ranjan, A., 2018. Vickrey meets alonso: Commute
  scheduling and congestion in a monocentric city. Journal of Urban Economics
  105, 40--53.

\bibitem[{Fosgerau and Small(2013)}]{fosgerau2013hypercongestion}
Fosgerau, M., Small, K.~A., 2013. Hypercongestion in downtown metropolis.
  Journal of Urban Economics 76, 122--134.

\bibitem[{Fu et~al.(2020)Fu, Wang, Tang, Zheng, and
  Geroliminis}]{fu2020empirical}
Fu, H., Wang, Y., Tang, X., Zheng, N., Geroliminis, N., 2020. Empirical
  analysis of large-scale multimodal traffic with multi-sensor data.
  Transportation Research Part C: Emerging Technologies 118, 102725.

\bibitem[{Ge et~al.(2020)Ge, Fukuda, Han, and Song}]{ge2020reservoir}
Ge, Q., Fukuda, D., Han, K., Song, W., 2020. Reservoir-based surrogate modeling
  of dynamic user equilibrium. Transportation Research Part C: Emerging
  Technologies 113, 350--369.

\bibitem[{Genser and Kouvelas(2022)}]{genser2022dynamic}
Genser, A., Kouvelas, A., 2022. Dynamic optimal congestion pricing in
  multi-region urban networks by application of a multi-layer-neural network.
  Transportation Research Part C: Emerging Technologies 134, 103485.

\bibitem[{Geroliminis(2015)}]{geroliminis2015cruising}
Geroliminis, N., 2015. Cruising-for-parking in congested cities with an mfd
  representation. Economics of Transportation 4~(3), 156--165.

\bibitem[{Geroliminis and Daganzo(2008)}]{geroliminis2008existence}
Geroliminis, N., Daganzo, C.~F., 2008. Existence of urban-scale macroscopic
  fundamental diagrams: Some experimental findings. Transportation Research
  Part B: Methodological 42~(9), 759--770.

\bibitem[{Geroliminis et~al.(2012)Geroliminis, Haddad, and
  Ramezani}]{geroliminis2012optimal}
Geroliminis, N., Haddad, J., Ramezani, M., 2012. Optimal perimeter control for
  two urban regions with macroscopic fundamental diagrams: A model predictive
  approach. IEEE Transactions on Intelligent Transportation Systems 14~(1),
  348--359.

\bibitem[{Geroliminis and Levinson(2009)}]{geroliminis2009cordon}
Geroliminis, N., Levinson, D.~M., 2009. Cordon pricing consistent with the
  physics of overcrowding. In: Transportation and Traffic Theory 2009: Golden
  Jubilee. Springer, pp. 219--240.

\bibitem[{Geroliminis et~al.(2014)Geroliminis, Zheng, and
  Ampountolas}]{geroliminis2014three}
Geroliminis, N., Zheng, N., Ampountolas, K., 2014. A three-dimensional
  macroscopic fundamental diagram for mixed bi-modal urban networks.
  Transportation Research Part C: Emerging Technologies 42, 168--181.

\bibitem[{Giuliano(1992)}]{giuliano1992assessment}
Giuliano, G., 1992. An assessment of the political acceptability of congestion
  pricing. Transportation 19~(4), 335--358.

\bibitem[{Godfrey(1969)}]{godfrey1969mechanism}
Godfrey, J., 1969. The mechanism of a road network. Traffic Engineering \&
  Control 8~(8).

\bibitem[{Gonzales(2015)}]{gonzales2015coordinated}
Gonzales, E.~J., 2015. Coordinated pricing for cars and transit in cities with
  hypercongestion. Economics of Transportation 4~(1-2), 64--81.

\bibitem[{Gonzales and Daganzo(2012)}]{gonzales2012morning}
Gonzales, E.~J., Daganzo, C.~F., 2012. Morning commute with competing modes and
  distributed demand: user equilibrium, system optimum, and pricing.
  Transportation Research Part B: Methodological 46~(10), 1519--1534.

\bibitem[{Gonzales and Daganzo(2013)}]{gonzales2013evening}
Gonzales, E.~J., Daganzo, C.~F., 2013. The evening commute with cars and
  transit: Duality results and user equilibrium for the combined morning and
  evening peaks. Transportation Research Part B: Methodological 57~(C),
  286--299.

\bibitem[{Gu et~al.(2018)Gu, Liu, Cheng, and Saberi}]{gu2018congestion}
Gu, Z., Liu, Z., Cheng, Q., Saberi, M., 2018. Congestion pricing practices and
  public acceptance: A review of evidence. Case Studies on Transport Policy
  6~(1), 94--101.

\bibitem[{Guo and Ban(2020)}]{guo2020macroscopic}
Guo, Q., Ban, X.~J., 2020. Macroscopic fundamental diagram based perimeter
  control considering dynamic user equilibrium. Transportation Research Part B:
  Methodological 136, 87--109.

\bibitem[{Haddad(2017)}]{haddad2017optimal}
Haddad, J., 2017. Optimal perimeter control synthesis for two urban regions
  with aggregate boundary queue dynamics. Transportation Research Part B:
  Methodological 96, 1--25.

\bibitem[{Haddad and Geroliminis(2012)}]{haddad2012stability}
Haddad, J., Geroliminis, N., 2012. On the stability of traffic perimeter
  control in two-region urban cities. Transportation Research Part B:
  Methodological 46~(9), 1159--1176.

\bibitem[{Haddad and Shraiber(2014)}]{haddad2014robust}
Haddad, J., Shraiber, A., 2014. Robust perimeter control design for an urban
  region. Transportation Research Part B: Methodological 68, 315--332.

\bibitem[{Haitao et~al.(2019)Haitao, Yang, Liang, Menendez, and
  Guler}]{haitao2019providing}
Haitao, H., Yang, K., Liang, H., Menendez, M., Guler, S.~I., 2019. Providing
  public transport priority in the perimeter of urban networks: A bimodal
  strategy. Transportation Research Part C: Emerging Technologies 107,
  171--192.

\bibitem[{Hamm et~al.(2022)Hamm, Loder, Tilg, Menendez, and
  Bogenberger}]{hamm2022network}
Hamm, L.~S., Loder, A., Tilg, G., Menendez, M., Bogenberger, K., 2022. Network
  inefficiency: Empirical findings for six european cities. Transportation
  Research Record, 03611981221082588.

\bibitem[{Ji and Geroliminis(2012)}]{ji2012spatial}
Ji, Y., Geroliminis, N., 2012. On the spatial partitioning of urban
  transportation networks. Transportation Research Part B: Methodological
  46~(10), 1639--1656.

\bibitem[{Jin(2020)}]{jin2020generalized}
Jin, W.-L., 2020. Generalized bathtub model of network trip flows.
  Transportation Research Part B: Methodological 136, 138--157.

\bibitem[{Keyvan-Ekbatani et~al.(2015)Keyvan-Ekbatani, Yildirimoglu,
  Geroliminis, and Papageorgiou}]{keyvan2015multiple}
Keyvan-Ekbatani, M., Yildirimoglu, M., Geroliminis, N., Papageorgiou, M., 2015.
  Multiple concentric gating traffic control in large-scale urban networks.
  IEEE Transactions on Intelligent Transportation Systems 16~(4), 2141--2154.

\bibitem[{Lamotte and Geroliminis(2018)}]{lamotte2018morning}
Lamotte, R., Geroliminis, N., 2018. The morning commute in urban areas with
  heterogeneous trip lengths. Transportation Research Part B: Methodological
  117, 794--810.

\bibitem[{Li et~al.(2021)Li, Yildirimoglu, and Ramezani}]{li2021robust}
Li, Y., Yildirimoglu, M., Ramezani, M., 2021. Robust perimeter control with
  cordon queues and heterogeneous transfer flows. Transportation Research Part
  C: Emerging Technologies 126, 103043.

\bibitem[{Liu and Geroliminis(2016)}]{liu2016modeling}
Liu, W., Geroliminis, N., 2016. Modeling the morning commute for urban networks
  with cruising-for-parking: An mfd approach. Transportation Research Part B:
  Methodological 93, 470--494.

\bibitem[{Loder et~al.(2017)Loder, Amb{\"u}hl, Menendez, and
  Axhausen}]{loder2017empirics}
Loder, A., Amb{\"u}hl, L., Menendez, M., Axhausen, K.~W., 2017. Empirics of
  multi-modal traffic networks--using the 3d macroscopic fundamental diagram.
  Transportation Research Part C: Emerging Technologies 82, 88--101.

\bibitem[{Loder et~al.(2019)Loder, Amb{\"u}hl, Menendez, and
  Axhausen}]{loder2019understanding}
Loder, A., Amb{\"u}hl, L., Menendez, M., Axhausen, K.~W., 2019. Understanding
  traffic capacity of urban networks. Scientific reports 9~(1), 1--10.

\bibitem[{Loder et~al.(2022)Loder, Bliemer, and Axhausen}]{loder2022optimal}
Loder, A., Bliemer, M.~C., Axhausen, K.~W., 2022. Optimal pricing and
  investment in a multi-modal city—introducing a macroscopic network design
  problem based on the mfd. Transportation Research Part A: Policy and Practice
  156, 113--132.

\bibitem[{Mahmassani et~al.(1987)Mahmassani, Williams, and
  Herman}]{mahmassani1987performance}
Mahmassani, H., Williams, J.~C., Herman, R., 1987. Performance of urban traffic
  networks. In: Proceedings of the 10th International Symposium on
  Transportation and Traffic Theory. Vol.~14. Elsevier Amsterdam, The
  Netherlands, pp. 1--20.

\bibitem[{Ni and Cassidy(2020)}]{ni2020city}
Ni, W., Cassidy, M., 2020. City-wide traffic control: modeling impacts of
  cordon queues. Transportation research part C: emerging technologies 113,
  164--175.

\bibitem[{Paipuri et~al.(2021)Paipuri, Barmpounakis, Geroliminis, and
  Leclercq}]{paipuri2021empirical}
Paipuri, M., Barmpounakis, E., Geroliminis, N., Leclercq, L., 2021. Empirical
  observations of multi-modal network-level models: Insights from the pneuma
  experiment. Transportation Research Part C: Emerging Technologies 131,
  103300.

\bibitem[{Ramezani et~al.(2015)Ramezani, Haddad, and
  Geroliminis}]{ramezani2015dynamics}
Ramezani, M., Haddad, J., Geroliminis, N., 2015. Dynamics of heterogeneity in
  urban networks: aggregated traffic modeling and hierarchical control.
  Transportation Research Part B: Methodological 74, 1--19.

\bibitem[{Simoni et~al.(2015)Simoni, Pel, Waraich, and
  Hoogendoorn}]{simoni2015marginal}
Simoni, M.~D., Pel, A.~J., Waraich, R.~A., Hoogendoorn, S.~P., 2015. Marginal
  cost congestion pricing based on the network fundamental diagram.
  Transportation Research Part C: Emerging Technologies 56, 221--238.

\bibitem[{Sipetas and Gonzales(2021)}]{sipetas2021continuous}
Sipetas, C., Gonzales, E.~J., 2021. Continuous approximation model for hybrid
  flexible transit systems with low demand density. Transportation Research
  Record 2675~(8), 198--214.

\bibitem[{Sirmatel and Geroliminis(2017)}]{sirmatel2017economic}
Sirmatel, I.~I., Geroliminis, N., 2017. Economic model predictive control of
  large-scale urban road networks via perimeter control and regional route
  guidance. IEEE Transactions on Intelligent Transportation Systems 19~(4),
  1112--1121.

\bibitem[{Small and Chu(2003)}]{small2003hypercongestion}
Small, K.~A., Chu, X., 2003. Hypercongestion. Journal of Transport Economics
  and Policy (JTEP) 37~(3), 319--352.

\bibitem[{Takayama(2020)}]{takayama2020gains}
Takayama, Y., 2020. Who gains and who loses from congestion pricing in a
  monocentric city with a bottleneck? Economics of Transportation 24, 100189.

\bibitem[{Takayama and Kuwahara(2017)}]{takayama2017bottleneck}
Takayama, Y., Kuwahara, M., 2017. Bottleneck congestion and residential
  location of heterogeneous commuters. Journal of Urban Economics 100, 65--79.

\bibitem[{Tirachini et~al.(2014)Tirachini, Hensher, and
  Rose}]{tirachini2014multimodal}
Tirachini, A., Hensher, D.~A., Rose, J.~M., 2014. Multimodal pricing and
  optimal design of urban public transport: The interplay between traffic
  congestion and bus crowding. Transportation research part b: methodological
  61, 33--54.

\bibitem[{Vansteenwegen et~al.(2022)Vansteenwegen, Melis, Akta{\c{s}},
  Montenegro, Vieira, and S{\"o}rensen}]{vansteenwegen2022survey}
Vansteenwegen, P., Melis, L., Akta{\c{s}}, D., Montenegro, B. D.~G., Vieira,
  F.~S., S{\"o}rensen, K., 2022. A survey on demand-responsive public bus
  systems. Transportation Research Part C: Emerging Technologies 137, 103573.

\bibitem[{Velaga et~al.(2012)Velaga, Beecroft, Nelson, Corsar, and
  Edwards}]{velaga2012transport}
Velaga, N.~R., Beecroft, M., Nelson, J.~D., Corsar, D., Edwards, P., 2012.
  Transport poverty meets the digital divide: accessibility and connectivity in
  rural communities. Journal of Transport Geography 21, 102--112.

\bibitem[{Vickrey(2020)}]{vickrey2020congestion}
Vickrey, W., 2020. Congestion in midtown manhattan in relation to marginal cost
  pricing. Economics of Transportation 21, 100152.

\bibitem[{Wardrop(1952)}]{wardrop1952road}
Wardrop, J.~G., 1952. Road paper. some theoretical aspects of road traffic
  research. Proceedings of the institution of civil engineers 1~(3), 325--362.

\bibitem[{Wu and Huang(2014)}]{wu2014equilibrium}
Wu, W.-X., Huang, H.-J., 2014. Equilibrium and modal split in a competitive
  highway/transit system under different road-use pricing strategies. Journal
  of Transport Economics and Policy (JTEP) 48~(1), 153--169.

\bibitem[{Xu et~al.(2018)Xu, Liu, Liu, and Huang}]{xu2018pareto}
Xu, S.-X., Liu, R., Liu, T.-L., Huang, H.-J., 2018. Pareto-improving policies
  for an idealized two-zone city served by two congestible modes.
  Transportation Research Part B: Methodological 117, 876--891.

\bibitem[{Yang and Huang(2005)}]{yang2005mathematical}
Yang, H., Huang, H.-J., 2005. Mathematical and economic theory of road pricing.

\bibitem[{Yildirimoglu et~al.(2021)Yildirimoglu, Ramezani, and
  Amirgholy}]{yildirimoglu2021staggered}
Yildirimoglu, M., Ramezani, M., Amirgholy, M., 2021. Staggered work schedules
  for congestion mitigation: A morning commute problem. Transportation Research
  Part C: Emerging Technologies 132, 103391.

\bibitem[{Yildirimoglu et~al.(2015)Yildirimoglu, Ramezani, and
  Geroliminis}]{yildirimoglu2015equilibrium}
Yildirimoglu, M., Ramezani, M., Geroliminis, N., 2015. Equilibrium analysis and
  route guidance in large-scale networks with mfd dynamics. Transportation
  Research Part C 59, 404--420.

\bibitem[{Yildirimoglu et~al.(2018)Yildirimoglu, Sirmatel, and
  Geroliminis}]{yildirimoglu2018hierarchical}
Yildirimoglu, M., Sirmatel, I.~I., Geroliminis, N., 2018. Hierarchical control
  of heterogeneous large-scale urban road networks via path assignment and
  regional route guidance. Transportation Research Part B: Methodological 118,
  106--123.

\bibitem[{Zheng et~al.(2017)Zheng, Dantsuji, Wang, and
  Geroliminis}]{zheng2017macroscopic}
Zheng, N., Dantsuji, T., Wang, P., Geroliminis, N., 2017. Macroscopic approach
  for optimizing road space allocation of bus lanes in multimodal urban
  networks through simulation analysis. Transportation Research Record
  2651~(1), 42--51.

\bibitem[{Zheng and Geroliminis(2013)}]{zheng2013distribution}
Zheng, N., Geroliminis, N., 2013. On the distribution of urban road space for
  multimodal congested networks. Transportation Research Part B: Methodological
  57~(C), 326--341.

\bibitem[{Zheng and Geroliminis(2020)}]{zheng2020area}
Zheng, N., Geroliminis, N., 2020. Area-based equitable pricing strategies for
  multimodal urban networks with heterogeneous users. Transportation Research
  Part A: Policy and Practice 136, 357--374.

\bibitem[{Zheng et~al.(2012)Zheng, Waraich, Axhausen, and
  Geroliminis}]{zheng2012dynamic}
Zheng, N., Waraich, R.~A., Axhausen, K.~W., Geroliminis, N., 2012. A dynamic
  cordon pricing scheme combining the macroscopic fundamental diagram and an
  agent-based traffic model. Transportation Research Part A: Policy and
  Practice 46~(8), 1291--1303.

\bibitem[{Zheng et~al.(2018)Zheng, Li, and Qiu}]{zheng2018slack}
Zheng, Y., Li, W., Qiu, F., 2018. A slack arrival strategy to promote
  flex-route transit services. Transportation Research Part C: Emerging
  Technologies 92, 442--455.

\end{thebibliography}
\end{document}